\documentclass[a4paper,11pt]{article}
\usepackage[T1]{fontenc}
\usepackage{lmodern}
\usepackage[english]{babel}
\usepackage[margin=27mm]{geometry}
\usepackage{microtype}
\usepackage{mathtools,amssymb,amsthm}
\usepackage{booktabs,xspace}
\usepackage[ruled,vlined,linesnumbered]{algorithm2e}
\usepackage{cite}
\usepackage[hidelinks]{hyperref}

\newtheorem{theorem}{Theorem}[section]
\newtheorem{lemma}[theorem]{Lemma}
\newtheorem{corollary}[theorem]{Corollary}
\newtheorem{proposition}[theorem]{Proposition}
\theoremstyle{definition}
\newtheorem{definition}[theorem]{Definition}
\theoremstyle{remark}
\newtheorem{remark}[theorem]{Remark}

\newcommand{\fpt}{\ensuremath{\mathsf{FPT}}\xspace}
\newcommand{\Pclass}{\ensuremath{\mathsf{P}}\xspace}
\newcommand{\NPclass}{\ensuremath{\mathsf{NP}}\xspace}
\newcommand{\Direct}{\textnormal{\textsc{Direct}}\xspace}
\newcommand{\Rooted}{\textnormal{\textsc{Rooted}}\xspace}
\newcommand{\epas}{\textnormal{\textsf{EPAS}}\xspace}
\newcommand{\clique}{\textnormal{\textsc{Clique}}}
\newcommand{\maxwtcsp}{\textnormal{\textsc{Max $k$-Weight CSP}}\xspace}
\newcommand{\OPT}{\mathrm{OPT}}
\newcommand{\wt}{\operatorname{wt}}
\newcommand{\scope}{\mathrm{sc}}
\newcommand{\supp}{\operatorname{supp}}
\newcommand{\maxsat}{\textnormal{\textsc{Max-SAT}}\xspace}
\newcommand{\maxcoverage}{\textnormal{\textsc{Maximum Coverage}}\xspace}
\newcommand{\wone}{\ensuremath{\mathsf{W[1]}}\xspace}
\newcommand{\maxcsp}{\textnormal{\textsc{Max-CSP}}\xspace}
\newcommand{\mincsp}{\textnormal{\textsc{Min-CSP}}\xspace}
\newcommand{\minones}{\textnormal{\textsc{Min-Ones}}\xspace}
\newcommand{\maxones}{\textnormal{\textsc{Max-Ones}}\xspace}
\newcommand{\gapeth}{\textnormal{Gap-ETH}\xspace}

\title{Parameterized Approximation of Max \texorpdfstring{$k$}{k}-Weight CSP:\\
A Dichotomy for Boolean Languages}
\author{Rudranarayan Kar\textsuperscript{1}\quad
Souvik Saha\textsuperscript{2}\quad
Abhishek Sahu\textsuperscript{1}\\[0.8em]
\small\textsuperscript{1}School of Computer Sciences, National Institute of Science Education\\
\small and Research (NISER), Bhubaneswar, India\\
\small (An Off-Campus Centre of Homi Bhabha National Institute (HBNI), Mumbai)\\
\small\textsuperscript{2}Chennai Mathematical Institute, Chennai, India}
\date{}
\hypersetup{
  pdftitle={Parameterized Approximation of Max k-Weight CSP: A Dichotomy for Boolean Languages},
  pdfauthor={Rudranarayan Kar, Souvik Saha, Abhishek Sahu},
  pdfsubject={Parameterized approximation and Boolean constraint satisfaction},
  pdfkeywords={Max k-Weight CSP, parameterized approximation, Boolean constraint language, dichotomy}
}

\begin{document}
\raggedbottom
\maketitle
\begin{abstract}
We study \textnormal{\textsc{Max $k$-Weight CSP}}, where Boolean constraints
carry nonnegative rational weights and the goal is to maximize the satisfied
weight using at most $k$ variables set to $1$. Each constraint uses at most
$d$ pairwise distinct variables and a relation from a finite Boolean constraint
language $\Gamma$. The framework includes \textsc{Partial Vertex Cover},
\textsc{Densest $k$-Subgraph}, bounded-frequency \textsc{Maximum Coverage},
and cardinality-constrained \textsc{Max-SAT}.

We give a complete classification of the fixed finite Boolean constraint
languages $\Gamma$ of arity at most $d$ for which this problem admits an
efficient parameterized approximation scheme (EPAS) parameterized by $k+d$. Only nonempty relations that exclude the all-zero
tuple affect the classification. In such a relation, a coordinate is a
\emph{unit witness} if the tuple with a single $1$ in that coordinate is
accepted, and a \emph{control} otherwise. A relation is \emph{Direct} if
setting a selected control to $0$, with all other coordinates unchanged,
preserves satisfaction. It is \emph{Rooted} if every satisfying tuple has a
selected unit witness such that keeping this witness, setting all other
unit witnesses to $0$, and setting any subset of the selected controls
to $0$ preserves satisfaction.

If all the relations are Direct, or all are Rooted, we give a deterministic
$(1-\varepsilon)$-approximation for every rational $\varepsilon\in(0,1)$
in time $2^{O(d+k\log(1+dk/\varepsilon))}\ell_{\mathrm{in}}^{O(1)}$,
where $\ell_{\mathrm{in}}$ is the total input length.
The same algorithms work for every value of $d$ and give an EPAS with
parameter $k+d$, provided that relation membership can be tested in
polynomial time.

If neither condition holds, no EPAS exists unless
$\mathsf{FPT}=\mathsf{W[1]}$, even for unweighted instances where no two
constraints involve the same set of variables.

\end{abstract}

\section{Introduction}

We study Boolean constraint satisfaction problems (CSPs). An instance
consists of a finite set $X$ of variables, each taking a value in $\{0,1\}$,
and a finite collection $\mathcal C$ of constraints. Let $d$ be an upper bound
on all constraint arities. A constraint $C$ of arity $r\le d$ is specified
by a tuple $(x_1,\ldots,x_r)$ of variables and a relation $R_C\subseteq\{0,1\}^r$. We call the positions $1,\ldots,r$ the coordinates
of the relation and of the constraint. An assignment $\tau:X\to\{0,1\}$ satisfies $C$ if
$(\tau(x_1),\ldots,\tau(x_r))\in R_C$. We write $\mathbf 0_r$ and
$\mathbf 1_r$ for the all-zero and all-one tuples in $\{0,1\}^r$, omitting
the subscript when the arity is clear. Throughout this paper, the variables
$x_1,\ldots,x_r$ of every constraint $C$ are pairwise distinct. This
assumption affects the classification; see Remark~\ref{rem:repeated-variables}. We write
$\scope(C):=\{x_1,\ldots,x_r\}$ for the set of variables of $C$ and call it
the \emph{underlying scope set} of $C$; thus $|\scope(C)|=r$. A Boolean
constraint language $\Gamma$ is a finite set of relations, and an instance is
over $\Gamma$ if $R_C\in\Gamma$ for every $C\in\mathcal C$.

The input to \maxwtcsp is a CSP instance $(X,\mathcal C)$ together with a
nonnegative rational weight $w(C)$ for every constraint $C\in\mathcal C$
and a nonnegative integer budget $k$. The \emph{Hamming weight}
$\wt(\tau):=|\tau^{-1}(1)|$ of an assignment $\tau$ is the number of
variables it sets to $1$. The goal is to find an assignment
$\tau$ with $\wt(\tau)\le k$ that maximizes the total weight of the
constraints it satisfies. The unweighted case has $w(C)=1$ for every
$C\in\mathcal C$.
For a constraint language $\Gamma$, \maxwtcsp{}$(\Gamma)$ is the restriction
to instances over $\Gamma$.

For an instance, let $n:=|X|$ and $m:=|\mathcal C|$.
All accuracy parameters are rational.
Let $\ell_{\mathrm{in}}$ be the total number of bits used to encode
the input, including the weights, budget, accuracy parameters when present,
and any relation descriptions supplied as part of the input.
All running times count bit operations.

An \emph{efficient
parameterized approximation scheme} (\epas) for \maxwtcsp{}$(\Gamma)$,
for a fixed $\Gamma$, parameterized by $k$, is an algorithm that, given an
instance and $\varepsilon\in(0,1)$, returns an assignment $\tau$ with
$\wt(\tau)\le k$ of value at least $(1-\varepsilon)$ times the optimum,
in time $f(k,\varepsilon)\ell_{\mathrm{in}}^{O(1)}$ for a computable function $f$.
When $d$ varies, an \epas{} parameterized by $k+d$ uses one algorithm
running in time $f(k,d,\varepsilon)\ell_{\mathrm{in}}^{O(1)}$.

\maxwtcsp{} contains \textsc{Partial Vertex Cover}~\cite{
DBLP:journals/cj/Marx08,DBLP:conf/soda/Manurangsi19},
\maxcoverage~\cite{DBLP:conf/aaai/SkowronF15,
DBLP:conf/compgeom/BadanidiyuruKL12,DBLP:conf/esa/Sellier23},
cardinality-constrained \maxsat~\cite{jain_parameterized_2023,
DBLP:journals/tcs/Manurangsi25}, and \textsc{Max Cut} with at most
$k$ vertices on one side~\cite{DBLP:journals/algorithmica/BonnetEPT15}
as special cases (Section~\ref{sec:intro-special-cases}).
Earlier work also gives approximation schemes for several of these
problems through common frameworks~\cite{jain_parameterized_2023,
FominGIK24,DBLP:conf/icalp/00020LS0U24,KunduSSU26}. Can the
parameterized approximability of these problems be determined uniformly by
the permitted relations alone? More precisely, for every $d\ge1$, for which
Boolean constraint languages $\Gamma$ of arity at most $d$ does
\maxwtcsp{}$(\Gamma)$ admit an \epas?
Theorem~\ref{thm:arityd-complete-dichotomy} gives a complete
classification. It also shows that two languages can each admit
an \epas{} while their union admits none.

\subsection{Our results}
\label{sec:intro-results}

Fix $d\ge 1$ and a finite Boolean constraint language $\Gamma$ whose
relations have arity at most $d$. Throughout, $d$ denotes this global arity
bound, while $r$ denotes the arity of an individual relation. Only the
nonempty relations that exclude the all-zero tuple affect the classification.

A relation $R\subseteq\{0,1\}^r$ is \emph{zero-valid} if
$\mathbf 0\in R$ and \emph{zero-free} otherwise. Let
\[
  \Gamma_+:=\{R\in\Gamma:R\ne\varnothing,\ \mathbf 0\notin R\}.
\]
Constraints with an empty relation are never satisfied and can be discarded.
The all-zero assignment satisfies every constraint whose relation contains
$\mathbf 0$, although selecting variables may later falsify such constraints.
The zero-valid lifting argument handles this possible loss. The dichotomy
therefore depends only on $\Gamma_+$.

Fix $R\in\Gamma_+$, and let $r\le d$ be its arity. Let $e_i$ be the tuple
whose only $1$ is in coordinate $i$. Coordinate $i$ is a \emph{unit witness} if
$e_i\in R$ and a \emph{control} otherwise.
The classification is determined by two \emph{repair rules}, that is, two
ways of deleting $1$s from a satisfying tuple.

\begin{itemize}
\item $R$ is \Direct if, for every $\alpha\in R$ and every control
      coordinate $i$ with $\alpha_i=1$, setting $\alpha_i$ to $0$ while
      leaving all other coordinates unchanged produces another tuple in $R$.

\item $R$ is \Rooted if every $\alpha\in R$ has a unit-witness coordinate
      $u$ with $\alpha_u=1$ for which the following operation
      preserves satisfaction. First keep $u$ at $1$, set every other
      unit-witness coordinate to $0$, and leave all controls unchanged. From
      this modified tuple, set any subset of its $1$-valued controls to $0$.
      Every tuple obtained in this way must belong to $R$.
\end{itemize}

Both rules can also be described by what they exclude. A control
coordinate $c$ is \emph{helpful} if some $\beta\notin R$ with
$\beta_c=0$ satisfies $\beta+e_c\in R$. A \emph{pinned conjunction}
consists of disjoint sets $A,T\subseteq[r]$, where $|A|\ge2$ and
$T$ contains only controls: after fixing the coordinates in $T$
to $1$ and those outside $A\cup T$ to $0$, a tuple is accepted
exactly when every coordinate in $A$ is $1$. The set $T$ may be
empty. Then $R$ is \Direct{} exactly when it has no helpful
control, and \Rooted{} exactly when it has no pinned conjunction
(Lemma~\ref{arityd-dist-hard-lem-pinned-and}). Thus the hard case
occurs when $\Gamma_+$ contains both obstructions, possibly in
different relations.

In the hardness part of the following theorem, an instance has
\emph{multiplicity one} if no two of its constraints have the same
underlying scope set.

\begin{theorem}[Fixed-arity dichotomy]
\label{thm:arityd-complete-dichotomy}
Let $d\ge 1$, and let $\Gamma$ be a fixed finite Boolean constraint language
whose relations have arity at most $d$. The following statements hold.
\begin{enumerate}
\item If every relation in $\Gamma_+$ is \Direct, then
      \maxwtcsp{}$(\Gamma)$ admits a deterministic \epas{} that, for every
      $\varepsilon\in(0,1)$, runs in time
      \[
        2^{O(d+k\log(1+dk/\varepsilon))}\ell_{\mathrm{in}}^{O(1)}.
      \]

\item If every relation in $\Gamma_+$ is \Rooted, then
      \maxwtcsp{}$(\Gamma)$ admits a deterministic \epas{} that, for every
      $\varepsilon\in(0,1)$, runs in time
      \[
        2^{O(k\log(1+dk/\varepsilon))}\ell_{\mathrm{in}}^{O(1)}.
      \]

\item If neither of the preceding conditions on $\Gamma_+$ holds, then
      \maxwtcsp{}$(\Gamma)$ admits no \epas{} unless
      $\fpt{}=\wone{}$, even on unweighted multiplicity-one
      instances.
\end{enumerate}
\end{theorem}

\begin{remark}
\label{rem:arityd-uniform-algorithms}
The deterministic algorithms in
Theorem~\ref{thm:arityd-complete-dichotomy} also apply when
relation descriptions are supplied as part of the input and
the arity bound $d$ varies.
Assume that membership in each relation can be tested in time
polynomial in $\ell_{\mathrm{in}}$, with a fixed exponent
for all inputs.

If every relation in $\Gamma_+$ is \Direct, or every relation
in $\Gamma_+$ is \Rooted, the corresponding algorithm satisfies
the running-time bound in the theorem.
This gives an \epas{} parameterized by $k+d$.
The algorithm assumes the corresponding condition on $\Gamma_+$;
the time needed to check that condition is not included.
\end{remark}

The two rules capture two ways in which control coordinates are
redundant for satisfaction. In a \Direct relation, a $1$-valued control can
be set to $0$ while every other coordinate, including every unit witness
currently set to $1$, remains unchanged. Since $R$ is zero-free, every
satisfying tuple must set some unit witness to $1$; otherwise, repeatedly
deleting its $1$-valued controls would produce $\mathbf 0\in R$. In a
\Rooted relation, the definition explicitly chooses a unit witness set to
$1$. It keeps this witness, removes all competing unit witnesses, and then
permits any subset of the remaining $1$-valued controls to be deleted. Call
a relation $R\in\Gamma_+$ \emph{accessible} if every tuple in $R$ has a $1$
in some unit-witness coordinate, and \emph{inaccessible} otherwise. Thus
every \Direct relation and every \Rooted relation is accessible.

Neither rule implies the other, and the \epas{} side requires the same rule
for all relations in $\Gamma_+$. The relation
\[
  R_{\mathrm{dir}}:=\{\beta\in\{0,1\}^3:
    (\beta_1\vee\beta_2)\wedge(\beta_3\to\beta_1\wedge\beta_2)\}
  =\{(1,0,0),(0,1,0),(1,1,0),(1,1,1)\}
\]
is \Direct but not \Rooted. Its unit witnesses are coordinates $1$ and $2$,
and its only control is coordinate $3$. It is \Direct because $(1,1,1)$ is
its only tuple with $\beta_3=1$, and setting $\beta_3$ to $0$ gives
$(1,1,0)\in R_{\mathrm{dir}}$. It is not \Rooted because of $(1,1,1)$:
keeping coordinate $1$ and setting coordinate $2$ to $0$ gives
$(1,0,1)\notin R_{\mathrm{dir}}$, and symmetrically
$(0,1,1)\notin R_{\mathrm{dir}}$. The relation
\[
\begin{aligned}
  R_{\mathrm{root}}&:=\{\beta\in\{0,1\}^3:
    (\beta_1\vee\beta_2)\wedge(\beta_1\wedge\beta_2\to\beta_3)\}\\
  &=\{(1,0,0),(0,1,0),(1,0,1),(0,1,1),(1,1,1)\}
\end{aligned}
\]
is \Rooted but not \Direct, with the same unit witnesses and control. It is
\Rooted because every tuple with exactly one of $\beta_1,\beta_2$ equal to
$1$ lies in $R_{\mathrm{root}}$, and for every $\beta\in R_{\mathrm{root}}$,
keeping a unit witness $u$ with $\beta_u=1$ and setting the other unit
witness to $0$ produces such a tuple, whatever value the control then
takes. It is not \Direct because setting the control of $(1,1,1)$ to $0$
gives $(1,1,0)\notin R_{\mathrm{root}}$. By
Theorem~\ref{thm:arityd-complete-dichotomy}, \maxwtcsp{}$(\{R_{\mathrm{dir}}\})$
and \maxwtcsp{}$(\{R_{\mathrm{root}}\})$ each admit an \epas, whereas
\maxwtcsp{}$(\{R_{\mathrm{dir}},R_{\mathrm{root}}\})$ admits none unless
$\fpt{}=\wone{}$.

The conditions are decidable by finite inspection. For a relation of
arity $r$, testing \Direct takes $O(r2^r)$ membership queries and testing
\Rooted takes $O(r4^r)$. Thus $O(|\Gamma|\,d\,4^d)$ queries determine the
case of the theorem; Section~\ref{sec:arityd-complete-distinct} gives the tests.

\paragraph*{The arity-two boundary.} For languages whose relations have arity at most two,
Theorem~\ref{thm:arityd-complete-dichotomy} depends on a single relation,
$R_{\mathrm{AND}}:=\{(1,1)\}$. It has no unit witness, so it is inaccessible
and hence neither \Direct nor \Rooted. Every other nonempty zero-free
relation of arity at most two is both \Direct and \Rooted
(Proposition~\ref{prop:arityd-binary-specialization}).

\begin{corollary}[Arity-two boundary]
\label{cor:arity-two}
Let $\Gamma$ be a fixed finite Boolean constraint
language of arity at most two.
For \maxwtcsp{}$(\Gamma)$ parameterized by the budget $k$:
\begin{enumerate}
\item If $R_{\mathrm{AND}}\notin\Gamma$, then there is
      a deterministic \epas{} that, for every rational
      $\varepsilon\in(0,1)$, runs in time $2^{O(k\log(1+k/\varepsilon))}
        \ell_{\mathrm{in}}^{O(1)}.$
      
\item If $R_{\mathrm{AND}}\in\Gamma$, then there is no
      \epas{} unless $\fpt{}=\wone{}$.
\end{enumerate}
\end{corollary}

\subsection{Problems covered by the framework}
\label{sec:intro-special-cases}
We briefly survey several special cases of our framework and the known
approximation results for these problems.
In the problems below, $k$ bounds the number of selected vertices, sets, or
true variables. All our \epas{} consequences use the \Rooted case of
Theorem~\ref{thm:arityd-complete-dichotomy} and its uniform version: they
are deterministic, allow nonnegative rational weights, and run in time
$(1+dk/\varepsilon)^{O(k)}\ell_{\mathrm{in}}^{O(1)}$, parameterized by $k+d$.
The relations below have polynomial-time membership tests.
For graphs, $d=2$. 

\noindent\textbf{Partial Vertex Cover.}\enspace
The goal is to select at most $k$ vertices maximizing the weight of edges
with a selected endpoint. Each vertex is a variable, and each edge gives
$R_{\mathrm{OR}}=\{(0,1),(1,0),(1,1)\}$. Both coordinates are unit
witnesses, so the relation is both \Direct and \Rooted. Thus
Theorem~\ref{thm:arityd-complete-dichotomy} recovers an \epas{} for this
problem. Marx gave an \epas{}~\cite{DBLP:journals/cj/Marx08}, and Manurangsi
later obtained the faster bound $(1/\varepsilon)^{O(k)}$ for the weighted
version~\cite{DBLP:conf/soda/Manurangsi19}.

\noindent\textbf{\maxcoverage{}.}\enspace
Viewing each vertex as the set of its incident edges makes
\textsc{Partial Vertex Cover} a coverage problem in which every element
belongs to two sets. Allowing arbitrary set systems gives \maxcoverage:
select at most $k$ sets maximizing the weight of covered elements. An
element belonging to $r$ sets gives the relation
$\mathrm{OR}_r:=\{\beta\in\{0,1\}^r:|\beta|\ge1\}$. This relation is both
\Direct and \Rooted, so frequency at most $d$ gives an \epas, recovering
the unweighted result of Skowron and
Faliszewski~\cite{DBLP:conf/aaai/SkowronF15}. For the weighted problem,
Sellier's scheme gives the faster bound
$(d/\varepsilon)^{O(k)}$~\cite{DBLP:conf/esa/Sellier23}. Without a frequency bound,
the polynomial-time $(1-1/e)$-approximation~\cite{DBLP:journals/mp/NemhauserWF78}
is optimal unless $\Pclass{}=\NPclass{}$~\cite{DBLP:journals/jacm/Feige98}.
Under \gapeth{}, for every fixed $\delta>0$, no
$(1-1/e+\delta)$-approximation runs in \fpt{} time parameterized by
$k$~\cite{DBLP:conf/icalp/Cohen-AddadG0LL19}; under randomized \gapeth{},
none runs even in time $f(k)\ell_{\mathrm{in}}^{o(k)}$ for any computable
$f$~\cite{Manurangsi20}.

\noindent\textbf{Cardinality-constrained \maxsat{}.}\enspace
The OR constraints above are clauses with only positive literals.
Allowing negative literals gives cardinality-constrained \maxsat:
maximize the total weight of satisfied clauses with at most $k$
variables set to $1$. A clause with a negative literal gives a
zero-valid relation, while a clause with only positive literals
gives $\mathrm{OR}_r$. Thus, formulas with at most $d$ variables
per clause admit an \epas.

More generally, Jain et al.\ gave an \epas{} for unweighted formulas
whose variable-clause incidence graph is $K_{t,t}$-free for fixed $t$
\cite{jain_parameterized_2023}. Faster randomized and deterministic
schemes were subsequently obtained independently by Inamdar et al.\
and Manurangsi
\cite{DBLP:conf/icalp/00020LS0U24,DBLP:journals/tcs/Manurangsi25}.

Biclique-freeness can be further weakened to bounded semi-ladder
index: Kundu et al.\ give deterministic schemes for weighted
\maxcoverage{} and cardinality-constrained \maxsat{} in this
setting~\cite{KunduSSU26}. More generally still, bounded
VC-dimension suffices for \maxcoverage{}
\cite{DBLP:conf/compgeom/BadanidiyuruKL12}
and unweighted cardinality-constrained \maxsat{}
\cite{KunduSSU26}.

For unrestricted formulas, a polynomial-time $(1-1/e)$-approximation
is known~\cite{Sviridenko01}, and the coverage lower bounds above
also apply.

\noindent\textbf{Balanced Set Cover.}\enspace
One can also reward leaving some elements uncovered. In a
cardinality-constrained, weighted version of \textsc{Balanced Set
Cover}~\cite[Definition~4.6]{FeldmanGKP09}, select at most $k$ sets maximizing covered blue
weight plus uncovered red weight. Blue elements give OR constraints, and a
red element in $r$ sets gives the zero-valid relation $\{\mathbf 0_r\}$.
Thus frequency at most $d$ gives an \epas.

\noindent\textbf{Maximum Unique Coverage.}\enspace
Changing the coverage rule to reward only elements contained in exactly
one selected set gives \textsc{Max Unique
Coverage}~\cite{DemaineFHS08,MisraRSS09}. Each element gives
$\mathrm{EO}_r:=\{\beta\in\{0,1\}^r:|\beta|=1\}$. All coordinates are unit
witnesses, so these relations are both \Direct and \Rooted. Frequency at
most $d$ therefore gives an \epas. For unweighted instances, the scheme of
Cervenjak et al.~\cite{CervenjakGUW24} gives
$(d\log(d+1)/\varepsilon)^{O(k)}$ time by enumerating subsets of its
approximate kernel.

\noindent\textbf{Graph and hypergraph Max Cut.}\enspace
For a graph $G=(V,E)$ and $S\subseteq V$, an edge crosses the cut $(S,V\setminus S)$
if and only if it has exactly one endpoint in $S$. Thus selecting
at most $k$ vertices to maximize the total weight of crossing edges
gives the cardinality-constrained version of
\textsc{Max Cut}~\cite{MisraRSS09}, with each edge represented by
$R_{\mathrm{XOR}}:=\mathrm{EO}_2=\{(0,1),(1,0)\}$.
For unweighted graphs with the requirement $|S|=k$, Bonnet et al.\
gave an \epas{}~\cite{DBLP:journals/algorithmica/BonnetEPT15},
and Fomin et al.\ later improved the running time to
$(1/\varepsilon)^{O(k)}$~\cite{FominGIK24}.
These schemes also apply to the version with $|S|\le k$:
after adding $k$ isolated vertices, any set of at most $k$
original vertices can be extended to a set of exactly $k$ vertices
without changing the cut value. Conversely, removing the added
vertices from a solution preserves its cut value.

In a hypergraph with vertex set $V$ and $S\subseteq V$, a hyperedge
crosses $(S,V\setminus S)$ if and only if it contains a vertex of $S$
and a vertex of $V\setminus S$.
A hyperedge of size $r\ge2$ is therefore represented by
$\mathrm{NAE}_r:=\{\beta\in\{0,1\}^r:1\le|\beta|<r\}$.
Every coordinate is a unit witness, so each $\mathrm{NAE}_r$
is both \Direct and \Rooted.
By Theorem~\ref{thm:arityd-complete-dichotomy}, weighted hypergraph
\textsc{Max Cut} with at most $k$ vertices on one side and hyperedges
of size at most $d$ admits an \epas{} parameterized by $k+d$.

\noindent\textbf{Threshold coverage.}\enspace
Raising the coverage requirement from one selected set to at least $t$
selected sets gives threshold coverage. An element in $r$ sets gives the
relation $\mathrm{TH}_{t,r}:=\{\beta\in\{0,1\}^r:|\beta|\ge t\}$.
For $2\le t\le r\le d$, this relation has no unit witness and is inaccessible.
Hence threshold coverage with frequency at most $d$ and any fixed
threshold $2\le t\le d$ admits no \epas{} unless
$\fpt{}=\wone{}$.

\subsection{Related CSP classifications}
\label{sec:intro-related}

Several earlier classifications of Boolean CSPs concern different
objectives. Creignou~\cite{Creignou95} classified the polynomial-time
approximability of \maxcsp{}. Khanna, Sudan, Trevisan, and
Williamson~\cite{KhannaSTW01} classified \minones{}, \maxones{}, and
\mincsp{}.
\maxcsp{} maximizes the weight of satisfied constraints without limiting
the number of ones. \minones{} minimizes the number of ones while
satisfying every constraint. Schaefer's satisfiability
classification~\cite{10.1145/800133.804350} and later classifications
with cardinality constraints~\cite{CreignouSS08} also require every
constraint to be satisfied. Our problem combines the \maxcsp{} objective
with a budget of at most $k$ ones, allowing some constraints to remain
unsatisfied.

Bonnet, Egri, Lin, and Marx~\cite{BonnetELM18} gave a parameterized
approximation classification for \mincsp{}. They determined which finite
Boolean languages admit a constant-factor approximation in
fixed-parameter time. \mincsp{} minimizes the number of unsatisfied
constraints, and its parameter bounds that number. Here, the parameter
instead bounds the number of ones. The approximation guarantee also
differs: they study constant-factor approximation, while we seek an
\epas{}, which allows arbitrarily small error. Their classification
therefore does not determine which languages admit a scheme for our
problem.

The distinction between exactly $k$ ones and at most $k$ ones is also
important. Marx~\cite{DBLP:journals/cc/Marx05} showed that satisfying all
constraints with exactly $k$ ones is either fixed-parameter tractable or
$\wone{}$-complete, depending on the Boolean language.
Bulatov and Marx~\cite{BulatovM11} studied related questions over finite
domains. Under our upper bound of $k$ ones, however, satisfying all
constraints is fixed-parameter tractable for every fixed
arity~\cite[Section~11, p.~181]{DBLP:journals/cc/Marx05}.
Our classification addresses the remaining optimization question:
which languages admit an \epas{} for maximizing the satisfied weight
when not every constraint can be satisfied within the budget?

\paragraph*{Organization.}
Section~\ref{sec:preliminaries} gives the notation and preliminaries,
and Section~\ref{sec:overview} explains the main proof ideas.
Section~\ref{sec:arityd-complete-distinct} develops the classification
of zero-free relations and the common ingredients for the algorithms.
Sections~\ref{ss:arityd-rooted} and~\ref{arityd-direct-sec-main}
give the \Rooted and \Direct algorithms, respectively.
Section~\ref{arityd-dist-hard-subsec-necessity} proves hardness.
Section~\ref{sec:classification-completion} completes the classification
and proves the arity-two boundary.
Section~\ref{sec:conclusion} concludes with open questions.

\section{Notation and Preliminaries}
\label{sec:preliminaries}

Throughout this paper, $X$ denotes a finite set of Boolean variables.

\paragraph*{Assignments and Hamming weight.}
A \emph{truth assignment} (or simply an assignment) is a function $\tau \colon X \to \{0,1\}$. The \emph{Hamming weight} of an assignment $\tau$, denoted $\wt(\tau) := |\tau^{-1}(1)|$, is the number of variables set to $1$. We denote the all-zero assignment by $\mathbf{0}$. For a subset of variables $S \subseteq X$, its indicator assignment is denoted by $\tau_S$, where $\tau_S(x) = 1$ if $x \in S$ and $0$ otherwise. 

A \emph{partial assignment} is a function $\rho\colon Y\to\{0,1\}$ for
some $Y\subseteq X$. If $\rho_1\colon Y_1\to\{0,1\}$ and
$\rho_2\colon Y_2\to\{0,1\}$ are partial assignments with
$Y_1\cap Y_2=\varnothing$, then $\rho_1\sqcup\rho_2\colon Y_1\cup Y_2\to\{0,1\}$
is the partial assignment that agrees with $\rho_1$ on $Y_1$ and with
$\rho_2$ on $Y_2$.

\begin{definition}[Boolean constraints and constraint languages]
For a positive integer $r$, let $[r]:=\{1,\ldots,r\}$.  A \emph{Boolean
relation} of arity $r$ is a set $R\subseteq\{0,1\}^r$, and a
\emph{Boolean constraint language} $\Gamma$ is a finite set of Boolean
relations.

A constraint $C$ over $X$ is specified by a relation
$R_C\subseteq\{0,1\}^{r_C}$ and an ordered tuple
$(x_1,\ldots,x_{r_C})$ of pairwise distinct variables from $X$.
This tuple is the \emph{scope tuple} of $C$, and
$\scope(C):=\{x_1,\ldots,x_{r_C}\}$ is its
\emph{underlying scope set}.
The same relation and scope tuple may occur more than once;
each occurrence is treated as a separate constraint.

For an assignment $\tau$, write
\[
  \tau[C]:=(\tau(x_1),\ldots,\tau(x_{r_C})).
\]
We say that $\tau$ satisfies $C$, written $\tau\models C$,
if $\tau[C]\in R_C$.
A constraint $C$ is a $\Gamma$-constraint if $R_C\in\Gamma$.
We also write $\iota_C(i):=x_i$ for $i\in[r_C]$.
\end{definition}

The feasibility problem $\operatorname{CSP}(\Gamma)$ asks, given a
finite variable set $X$ and a finite collection $\mathcal C$ of
$\Gamma$-constraints over $X$, whether there is an assignment
satisfying every constraint.
We instead maximize the weight of the satisfied constraints subject to a
Hamming-weight budget.

\begin{definition}[\maxwtcsp{}$(\Gamma)$]
For a fixed Boolean constraint language $\Gamma$, an instance is a tuple
$\mathcal I=(X,\mathcal C,w,k)$, where $X$ is a finite set of Boolean
variables, $\mathcal C$ is a finite collection of $\Gamma$-constraints
over $X$, $w:\mathcal C\to\mathbb Q_{\ge0}$ assigns a nonnegative rational
weight to each constraint, and $k\in\mathbb Z_{\ge0}$ is the
Hamming-weight budget.  A feasible solution is an assignment $\tau:X\to
\{0,1\}$ with $\wt(\tau)\le k$.  Its value and the optimum are
\[
  F^{\mathcal I}(\tau)
    :=\sum_{\substack{C\in\mathcal C:\ \tau\models C}}w(C),
  \qquad
  \OPT(\mathcal I):=\max_{\wt(\tau)\le k}F^{\mathcal I}(\tau).
\]
The task is to find a feasible assignment of maximum value.
\end{definition}

When the instance is clear, we write $F(\tau)$ instead of
$F^{\mathcal I}(\tau)$.
The \emph{support} of $\tau$ is the set
$\{x\in X:\tau(x)=1\}$.
For $S\subseteq X$, we write $F(S):=F(\tau_S)$.
Thus a feasible solution can be represented by a set $S\subseteq X$
with $|S|\le k$.

An instance is \emph{unweighted} if $w(C)=1$ for every
$C\in\mathcal C$.
It has \emph{arity at most $d$} if $r_C\le d$ for every
$C\in\mathcal C$.
We omit $(\Gamma)$ from the problem name when the constraint language
is clear from context or when arbitrary Boolean relations are allowed.

\paragraph*{Feasible completions and residual instances.}
Let $\mathcal I=(X,\mathcal C,w,k)$ be an instance and let
$\rho\colon H\to\{0,1\}$ be a partial assignment with
$|\rho^{-1}(1)|\le k$.
A \emph{feasible completion} of $\rho$ is an assignment
$\sigma\colon X\setminus H\to\{0,1\}$ that selects at most
$k-|\rho^{-1}(1)|$ variables. Together, $\rho$ and $\sigma$
give the feasible assignment $\rho\sqcup\sigma$ to the original
instance, with value $F^{\mathcal I}(\rho\sqcup\sigma)$.

The \emph{residual instance} under $\rho$ is obtained by substituting
the values prescribed by $\rho$ into every constraint and keeping
the remaining variables in their original coordinate order.
The remaining budget is $k-|\rho^{-1}(1)|$.
Constraints that become always unsatisfied are discarded.
The weights of constraints that become always satisfied are kept
as a constant contribution to the objective.
The remaining constraints keep their original weights.
The resulting relations need not belong to $\Gamma$.

An algorithm can also require its output to extend $\rho$ while
keeping the original constraints unchanged. Our algorithms for
zero-free constraints use this form, because substituting fixed
values need not preserve zero-freeness or the repair rules.

\begin{remark}[Repeated variables within a constraint]
\label{rem:repeated-variables}
Allowing repeated variables can change whether an \epas{} exists.
Consider
\[
  R:=\{(1,0,0),(0,1,0),(1,0,1),(0,1,1),(1,1,1)\}.
\]
Its unit-witness coordinates are $1$ and $2$.
Every satisfying tuple has a $1$ in at least one of these coordinates.
Keeping one such coordinate at $1$ and setting the other to $0$
gives $(1,0,c)$ or $(0,1,c)$, both of which belong to $R$
for either value of $c$. Setting the third coordinate to $0$
also preserves membership. Thus $R$ is \Rooted, and
\maxwtcsp{}$(\{R\})$ admits an \epas{} by
Theorem~\ref{thm:arityd-complete-dichotomy}.

If repeated variables are allowed, a constraint with relation $R$
and scope tuple $(x,x,y)$ is satisfied exactly when both $x$ and $y$
are assigned $1$. It therefore expresses $R_{\mathrm{AND}}$.
Replacing each AND constraint in this way preserves its weight
and the number of variables assigned $1$.
Hence the version allowing repeated variables admits no \epas{}
unless $\fpt{}=\wone{}$.

For each constraint $C$, delete from its relation every tuple whose entries
disagree at positions occupied by the same variable. Then retain only the
first position occupied by each variable, obtaining a relation $R'_C$ on the
distinct variables of $C$. Different constraints using the same relation may
yield different $R'_C$. Let $\Gamma^{\mathrm{id}}$ be the set of all such
relations obtainable from constraints over $\Gamma$, allowing all possible
scope tuples. These replacements work in both directions and preserve weights,
budget, and assignment values. Theorem~\ref{thm:arityd-complete-dichotomy}
therefore gives an \epas{} if all nonempty zero-free relations in
$\Gamma^{\mathrm{id}}$ are \Direct, or all are \Rooted; otherwise, none exists
unless $\fpt{}=\wone{}$.
\end{remark}

\paragraph*{Multiplicity and instance notation.}
For $S\subseteq X$, its \emph{multiplicity} is the number of
constraints $C\in\mathcal C$ with $\scope(C)=S$.
Constraints are counted separately, regardless of their relations,
weights, or the order of variables in their scope tuples.
An instance has \emph{multiplicity one} if no two constraints have
the same underlying scope set.
The requirement that variables be distinct within each constraint
places no restriction on how many constraints may share a scope set.

For an instance $\mathcal I=(X,\mathcal C,w,k)$, let $n:=|X|$
be the number of variables and $m:=|\mathcal C|$ the number of
constraints.

\paragraph*{Approximation schemes.}
Let $d$ bound the arities of the relations in $\Gamma$.
An \epas{} for \maxwtcsp{}$(\Gamma)$, parameterized by $k+d$,
is an algorithm that, given an instance $\mathcal I$ and a rational
$\varepsilon\in(0,1)$, returns a feasible assignment $\tau$ satisfying
\[
  F^{\mathcal I}(\tau)\ge(1-\varepsilon)\OPT(\mathcal I)
\]
in time $f(k,d,\varepsilon)\ell_{\mathrm{in}}^{O(1)}$ for some computable
function $f$. The polynomial exponent is independent of
$k$, $d$, and $\varepsilon$.
For fixed $\Gamma$, the arity bound $d$ is constant, so this is
an \epas{} parameterized by $k$.

The \emph{decision version} of \maxwtcsp{}$(\Gamma)$ asks, given
an instance $\mathcal I$ and a rational threshold $t$, whether
there is a feasible assignment of value at least $t$,
equivalently, whether $\OPT(\mathcal I)\ge t$.

\section{Technical Overview}
\label{sec:overview}

We explain the proof of Theorem~\ref{thm:arityd-complete-dichotomy}
in four parts: the treatment of zero-valid constraints, the \Rooted
algorithm, the \Direct algorithm, and the hardness reduction.
The full proofs are in Sections~\ref{sec:arityd-complete-distinct}--\ref{sec:classification-completion}.
Selecting a variable means setting it to $1$; all other variables
are set to $0$. For $S\subseteq X$, write $\tau_S$ for the assignment
that selects exactly $S$, as in Section~\ref{sec:preliminaries}.
Throughout the overview, $\OPT$ is the optimum
of the original instance. Since weights are nonnegative, the value
of any feasible assignment on just the zero-free constraints is
also at most $\OPT$.

\smallskip
\noindent\textbf{Zero-valid lifting.}
A zero-valid constraint accepts the all-zero assignment; a zero-free
constraint rejects it. Our treatment of zero-valid constraints adapts
Manurangsi's degree-threshold idea for clauses with negative
literals~\cite{DBLP:journals/tcs/Manurangsi25}. We fix the variables
that can affect a large weight of these constraints and bound the
loss caused by the others.

If $k=0$, return the all-zero assignment. Otherwise fix $\xi\in(0,1)$,
eventually chosen as $\varepsilon/2$. Let $w_0$ be the total weight
of zero-valid constraints, and let $D_0(v)$ be the total weight
of those containing $v$. The all-zero assignment shows that $w_0\le\OPT$.
Call $v$ \emph{heavy} if $D_0(v)>\xi w_0/(2k)$, and let $H$
contain these variables. Each constraint contains at most $d$
variables, so $\sum_v D_0(v)\le dw_0$. If $H$ is nonempty, adding
the defining inequalities for all heavy variables gives
$|H|\xi w_0/(2k)<\sum_{v\in H}D_0(v)\le dw_0$.
Dividing by $\xi w_0/(2k)>0$ gives $|H|<2dk/\xi$.
The same bound holds when $H$ is empty; in particular, $w_0=0$
implies that there are no heavy variables.

Try every partial assignment $\rho:H\to\{0,1\}$ with
$|\rho^{-1}(1)|\le k$. For each $\rho$, run the appropriate zero-free
algorithm to obtain a feasible completion $\sigma$ on $X\setminus H$.
Evaluate the combined assignment $\rho\sqcup\sigma$ on all original constraints. There are at most
$(|H|+1)^k$ guesses, since each is determined by at most $k$ selected
variables. Keep the best answer.
The initial fixing imposes requirements on the returned assignment;
it does not substitute values into the Boolean relations, which
could destroy their repair properties.

Consider the guess agreeing with an optimal solution. That solution
and the returned assignment agree on $H$ and each select at most
$k$ variables. Any constraint whose satisfaction changes contains
a variable on which the two assignments differ. Since they agree on $H$,
every such variable is nonheavy. Moving from the optimal assignment
to the returned assignment can change at most $k$ variables from
$1$ to $0$ and at most $k$ variables from $0$ to $1$.
Thus at most $2k$ variables change. Each belongs to zero-valid
constraints of total weight at most $\xi w_0/(2k)$, so the loss
is at most $2k\cdot\xi w_0/(2k)=\xi w_0\le\xi\OPT$.
The algorithms below lose at most $\varepsilon\OPT/2$ on the
zero-free constraints compared with any feasible assignment that
agrees with the guess. Taking $\xi=\varepsilon/2$ therefore gives
total loss at most $\varepsilon\OPT$.

\smallskip
\noindent\textbf{The \Rooted case.}
The \Rooted algorithm builds on the branching and exchange
ideas of Jain et al.~\cite{jain_parameterized_2023}.
The scores $m_K$ and the bounds on the loss from adding
variables extend these ideas to general \Rooted relations,
where adding a variable can destroy satisfaction.

Assume every nonempty zero-free relation is \Rooted.
For a constraint $C$, let $U_C$ contain its unit-witness variables:
selecting any one of them alone satisfies $C$. Let $N_C$ contain
its remaining variables, called controls. These roles depend on
$C$. The \Rooted rule says that, from every satisfying assignment
to $C$, we can keep one of its selected unit witnesses at $1$ and set
all other unit witnesses to $0$. The constraint remains satisfied even
if we also set any subset of its selected controls to $0$.

\smallskip
\noindent\emph{Scores and the graph.}
Write $F(S)$ for the weight of zero-free constraints satisfied by
selecting $S$. The \emph{single-variable score} $M(v)=F(\{v\})$
is the weight satisfied by selecting $v$ alone, with every other
variable set to $0$. It stays fixed throughout the search.

The search maintains a set $K$ of variables required to be selected,
initially $K=\rho^{-1}(1)$. A candidate is a variable $v$ outside
$K$ that has not been fixed to $0$. Its \emph{candidate score at $K$},
$m_K(v)$, is the total weight of constraints that would be satisfied
with $v$ as their only selected unit witness if we added $v$ to $K$.
For each $C$ with $v\in U_C$ and $K\cap U_C=\varnothing$,
count $w(C)$ if $\tau_{K\cup\{v\}}\models C$.
Every constraint counted by $m_K(v)$ is also counted by
$M(v)$, so $0\le m_K(v)\le M(v)$.

The candidate score need not equal the actual change
$F(K\cup\{v\})-F(K)$: it counts only these constraints
and does not subtract any weight lost by adding $v$.

We use a directed graph to record pairs of variables that could
cause large losses or make us count the same constraint twice.
It has one vertex for each variable.
Fix a parameter $\theta>0$, to be chosen below. For distinct
variables $v,x$, let $I(v\to x)$ be the total weight of all
constraints counted by $M(v)$ that also contain $x$.
Draw an arrow $v\to x$ exactly when
$I(v\to x)>\theta M(v)$, and give it weight $I(v\to x)$.
Each sum is defined even when no arrow is drawn. The sums and the graph are computed once before the search.

Suppose we add two distinct candidates $v$ and $x$ to $K$.
Among the constraints counted by $m_K(v)$, only those
involving $x$ can also be counted by $m_K(x)$ or become
unsatisfied when $x$ is added after $v$.
These constraints are also counted by $M(v)$, so their
total weight is at most $I(v\to x)$.
We call these arrows \emph{heavy interactions}.
They need not be symmetric: $v\to x$ does not necessarily
imply $x\to v$.

Each constraint counted by $M(v)$ contributes to at most $d-1$
sums $I(v\to x)$, so
\[
 \sum_{x\ne v}I(v\to x)\le(d-1)M(v).
\]
Hence at most $d_{\max}:=\lceil(d-1)/\theta\rceil$ arrows leave
any variable. If $M(v)=0$, all these sums are zero and no arrow
leaves $v$. A variable may still receive many arrows.

\smallskip
\noindent\emph{The search.}
Set $\ell=k(3d_{\max}+1)$.
At state $K$, let $T$ contain the $\ell$ candidates with largest
positive candidate scores $m_K(v)$, or all such candidates if fewer
exist. Let $E$ contain candidates receiving an arrow from a variable
in $K$. There are at most
$|K|d_{\max}\le kd_{\max}$ such candidates. Record $K$ as a
possible answer.

When $|T|=\ell$, remove $E$ from $T$.
Since $|E|\le kd_{\max}$, at least
$\ell-kd_{\max}=k(2d_{\max}+1)$ candidates remain.
Form an undirected graph on $T\setminus E$, joining two
candidates whenever there is an arrow between them in
either direction. We seek an \emph{independent set} of
at least $k$ candidates, meaning that no two are joined
by an edge.

For any set $W$ of vertices in this graph, there are at most
$d_{\max}|W|$ edges between its vertices, since each vertex
has at most $d_{\max}$ outgoing arrows.
Thus every nonempty induced subgraph has average degree
at most $2d_{\max}$ and contains a vertex with at most
$2d_{\max}$ neighbors. The standard coloring
algorithm~\cite{Diestel2000,MatulaBeck1983} therefore uses
at most $2d_{\max}+1$ colors. Its largest color class
gives an independent set $Q$ of size at least
$|T\setminus E|/(2d_{\max}+1)\ge k$.
This takes deterministic polynomial time.

Since $Q\subseteq T\setminus E$, no arrow goes from $K$
to $Q$, and independence excludes arrows between
candidates in $Q$. By the definition of the arrows,
\[
 I(u\to v)\le\theta M(u)
 \qquad (u\in K\cup Q,\ v\in Q,\ u\ne v).
\]
We use this bound below to control losses and repeated
counting when adding the replacements.

Choose $k$ candidates from the independent set and order
them as $q_1,\ldots,q_k$. At the current call, record
$K\cup\{q_1,\ldots,q_j\}$ for every $0\le j\le k-|K|$.
The case $j=0$ records $K$ itself.

If $|T|<\ell$, the only solution recorded directly at
the current state is $K$.
In either case, if $|K|<k$, the algorithm repeats the
entire procedure from $K\cup\{v\}$ for every $v\in T\cup E$.
Each call recomputes its candidate scores and sets and,
whenever its top set has size $\ell$, constructs its own
independent set and records the corresponding solutions.
The single-variable scores $M(x)$ and the directed graph
are computed once before the search.
For the fixed initial guess, the algorithm returns the
solution with the largest $F$-value recorded anywhere
in the search.

\smallskip
\noindent\emph{Why the search works.}
To see why the search suffices, fix a feasible solution $S$
agreeing with the fixed initial guess.
The algorithm does not know $S$; we use it only in the proof.
The initial set $K$ satisfies $K\subseteq S$.
Whenever $T\cup E$ contains a variable of $S\setminus K$,
follow a recursive call that adds such a variable.
Each call preserves $K\subseteq S$ and increases $|K|$.

Eventually either $K=S$, in which case $S$ has been recorded,
or $(S\setminus K)\cap(T\cup E)=\varnothing$.
It remains to consider the latter case.
In this case, every variable of $S\setminus K$ lies outside
$T$ and receives no arrow from $K$.

We first compare $F(S)$ with $F(K)$ by considering the
constraints that become unsatisfied when we remove
$S\setminus K$. We show that each such constraint contributes its weight
to $m_K(v)$ for some $v\in S\setminus K$, or to
$I(u\to x)$ for some $u\in K$ and $x\in S\setminus K$.

Consider such a constraint $C$. If $K$ contains no unit
witness of $C$, the \Rooted rule lets us keep one unit
witness $v\in S\setminus K$ and exactly the controls
selected by $K$. Thus $\tau_{K\cup\{v\}}$ satisfies $C$, with
$v$ as its only selected unit witness, so $w(C)$ is
counted by $m_K(v)$.

Otherwise, choose a unit witness $u\in K$ of $C$.
Selecting $u$ alone satisfies $C$, so $w(C)$ is counted
by $M(u)$. Since removing $S\setminus K$ makes $C$
unsatisfied, $C$ must contain some variable
$x\in S\setminus K$. Hence $w(C)$ is counted by
$I(u\to x)$. There is no arrow from $K$ to
$S\setminus K$, so $I(u\to x)\le\theta M(u)$.
For each $u\in K$, there are $|S\setminus K|$ choices
of $x$, and each term $I(u\to x)$ is at most $\theta M(u)$.
Thus the total weight of constraints in this case is at most
\[
 \sum_{u\in K}\sum_{x\in S\setminus K} I(u\to x)
 \le \theta|S\setminus K|\sum_{u\in K}M(u).
\]

Every constraint lost when passing from $S$ to $K$
is counted in one of these two cases. Counting a
constraint more than once only increases the bound.
Therefore,
\[
 F(S)-F(K)
 \le \sum_{v\in S\setminus K}m_K(v)
 +\theta|S\setminus K|\sum_{u\in K}M(u).
\]

If $T$ is not full, it contains every candidate with $m_K(v)>0$.
Every $v\in S\setminus K$ therefore has $m_K(v)=0$,
and returning $K$ loses at most
$\theta k\sum_{u\in K}M(u)$.

If $T$ is full, put $s=|S\setminus K|$ and let $Q_s$ be the
first $s$ candidates in the chosen independent set.
Every candidate in $T$ has $m_K$-score at least as large
as that of every variable in $S\setminus K$: $T$ contains
the largest positive $m_K$-scores, and $S\setminus K$ lies
outside $T$. Thus $Q_s$ has total candidate score at least
that of the variables in $S\setminus K$.
Since $s\le k-|K|$, trying every $j$ includes $j=s$.
The algorithm therefore records $A=K\cup Q_s$, even though
it does not know $S$.

We compare the gain $F(A)-F(K)$ with
$\sum_{v\in Q_s}m_K(v)$, the sum of the candidate scores
before any replacements are added.
To do this, consider adding the candidates in $Q_s$
one at a time. The actual gain can be smaller than
the sum of their scores for two reasons.

First, consider a constraint $C$ counted by $m_K(v)$.
It is unsatisfied by $\tau_K$ and satisfied by $\tau_{K\cup\{v\}}$.
If $C$ contains no candidate from $Q_s$ added before $v$,
the preceding additions do not change any variable of $C$,
so adding $v$ gives a new gain of $w(C)$.

Thus only constraints containing a candidate $x\in Q_s$
added before $v$ can fail to give the gain counted by
$m_K(v)$.
Each such constraint is also counted by $M(v)$, so its
weight is included in $I(v\to x)$.
Since no arrow joins two replacements,
$I(v\to x)\le\theta M(v)$.
Summing over at most $s$ candidates added before each
$v\in Q_s$, the total weight counted by the scores
that may fail to give a new gain is at most
$\theta s\sum_{v\in Q_s}M(v)$.

Second, consider a constraint $C'$ that is satisfied
before $v$ is added but becomes unsatisfied afterwards.
The \Rooted rule guarantees that, before the addition,
$C'$ has a selected unit witness $u$, either in $K$
or among the candidates from $Q_s$ added before $v$.
Selecting $u$ alone satisfies $C'$, so its weight is
counted by $M(u)$. Also, $C'$ must contain $v$, since
only the value of $v$ changes. Hence $w(C')$ is counted
by $I(u\to v)$.

There is no arrow from $K$ to $Q_s$ or between replacements,
so $I(u\to v)\le\theta M(u)$.
For each $u\in A$, there are at most $s$ later additions
of candidates $v$. Summing over these additions and
then over $u$ bounds the total weight lost by
$\theta s\sum_{u\in A}M(u)$.

The two bounds for adding $Q_s$ give
\[
 F(A)-F(K)\ge \sum_{v\in Q_s}m_K(v)
 -\theta s\sum_{v\in Q_s}M(v)
 -\theta s\sum_{z\in A}M(z).
\]
Since $F(A)-F(S)=[F(A)-F(K)]-[F(S)-F(K)]$,
we subtract the upper bound on $F(S)-F(K)$ from the
lower bound on $F(A)-F(K)$ to obtain
\[
 \begin{aligned}
 F(A)\ge{}& F(S)+\sum_{v\in Q_s}m_K(v)
                -\sum_{x\in S\setminus K}m_K(x)\\
          &-\theta s\left(
             \sum_{u\in K}M(u)+\sum_{v\in Q_s}M(v)
             +\sum_{z\in A}M(z)\right).
 \end{aligned}
\]
The difference between the two candidate-score sums
is nonnegative. Also, $K$ and $Q_s$ are disjoint and
together form $A$, so the first two sums of
single-variable scores add up to $\sum_{z\in A}M(z)$.
Therefore,
\[
 F(A)\ge F(S)-2\theta s\sum_{z\in A}M(z).
\]
We apply this bound only once along the proof's path,
when replacing $S\setminus K$ by $Q_s$.
Earlier calls on this path simply add variables of $S$ to $K$.

\smallskip
\noindent\emph{Bounding the error.}
For every $Q\subseteq X$ with $|Q|\le k$,
\begin{equation}
\label{eq:overview-score-bound}
 \sum_{v\in Q}M(v)\le c_d\OPT,
\end{equation}
where $c_1=1$ and $c_d=d(1-1/d)^{-(d-1)}\le ed$ for $d\ge2$.
For $d\ge2$, retain each variable of $Q$ independently with
probability $1/d$, obtaining $Q'$. For a constraint $C$ counted
by $M(v)$ with $v\in Q$, the event
$Q'\cap\scope(C)=\{v\}$ satisfies $C$ and has probability
\[
 \frac1d\left(1-\frac1d\right)^{|Q\cap\scope(C)|-1}
 \ge\frac1{c_d}.
\]
For each $C$, these events are disjoint. Hence
$\mathbb E[F(Q')]\ge\sum_{v\in Q}M(v)/c_d$.
Since $Q'$ is feasible, its expected value is at most $\OPT$.
For $d=1$, $\sum_{v\in Q}M(v)=F(Q)\le\OPT$.
This random choice only proves the bound; the search is deterministic.

Applying the bound to $A$ and using $s\le k$, the
replacement loss is at most $2\theta k c_d\OPT$.
Choosing $\theta=\varepsilon/(4kc_d)$ makes this
at most $\varepsilon\OPT/2$.
The same choice bounds the loss when $|T|<\ell$
by $\varepsilon\OPT/4$ and gives
$\ell=O(d^2k^2/\varepsilon)$.

Each recursive call adds one variable and has
$O(k(d_{\max}+1))$ children. The depth is at most $k$, and the
running time, including the initial guesses, is
$2^{O(k\log(1+dk/\varepsilon))}\ell_{\mathrm{in}}^{O(1)}$.
The common reduction supplies the remaining $\varepsilon\OPT/2$
loss allowance for zero-valid constraints.

\smallskip
\noindent\textbf{The \Direct case.}
Now assume every nonempty zero-free relation is \Direct.
Keep $F$, $M(v)$, and $U_C$ from above.
For distinct variables $x,y$, let $J(x,y)$ be the total
weight of constraints for which both $x,y\in U_C$.
Fix a threshold $\theta\in(0,1)$, to be chosen below.
Construct an undirected graph with one vertex per variable,
joining $x,y$ when $J(x,y)>0$ and
$J(x,y)\ge\theta\max\{M(x),M(y)\}$.
Since $\sum_{y\ne x}J(x,y)\le(d-1)M(x)$ and every
neighbor contributes at least $\theta M(x)$, the maximum
degree is at most $d_{\max}:=\lceil(d-1)/\theta\rceil$.
Variables with $M(x)=0$ are isolated.

Color each vertex independently green with probability
$1/(d_{\max}+1)$ and red otherwise. For a feasible solution
$S$ agreeing with the initial guesses, call a coloring
\emph{good} if every variable of $S$ is green and every
neighbor outside $S$ is red. At most $k$ variables must
be green and $kd_{\max}$ must be red, so a random coloring
is good with probability at least $[e(d_{\max}+1)]^{-k}$.

In a good coloring, every neighbor of $S$ outside $S$
is red. Hence every green connected component meeting
$S$ lies within $S$ and has at most $k$ variables.
We use a branching algorithm, described below, to choose
which components to select in full.

Before branching, process each sampled coloring as follows.
Set all red variables to $0$. Also set all variables in
a green component to $0$ if it has more than $k$ variables
or contains a variable fixed to $0$.
Skip the coloring if this assigns $0$ to a variable
fixed to $1$.

Select in full every remaining component containing
a variable fixed to $1$. If the total number of selected
variables exceeds $k$, skip the coloring. Otherwise,
subtract this number from $k$ to obtain the remaining
budget. The components still available for selection
contain no variables fixed by the initial guesses.

Each remaining component is an \emph{item}.
Write $\mathcal K_{\mathrm{free}}$ for this family of components.
Selecting an item $y$ selects every variable in that component
and has cost $\kappa(y):=|y|$.
Leaving it unselected sets all its variables to $0$ and costs zero.

\smallskip
\noindent\emph{Constraints on items.}
For each original zero-free constraint $C$ of positive weight,
define a constraint on the at most $d$ items whose components
meet $\scope(C)$.
An item assignment satisfies this new constraint exactly when
the corresponding variable assignment satisfies $C$ and all
selected unit witnesses of $C$ lie in exactly one component.
All variable values fixed earlier are retained, including
selected unit witnesses in components already fixed to $1$.
The new constraint has weight $w(C)$, and its scope is the set
of items it involves. Write $F_{\rm mask}$ for the total
satisfied weight of these item constraints.

The masking condition makes the repair rule survive the recursive
fixings: every nonempty zero-free relation obtained by fixing items
remains \Direct (Lemma~\ref{arityd-direct-lem-owner-hus}).
Indeed, in a satisfying assignment all selected original unit witnesses
lie in one component. Resetting another selected item changes only
controls, so satisfaction is preserved. If that component were already
fixed to $1$, resetting every remaining item would preserve satisfaction,
contradicting zero-freeness. It is therefore a remaining unit witness,
and any remaining control item can be reset safely.

\smallskip
\noindent\emph{The recursive search.}
Start with these item constraints and the remaining budget.
Write $b_0\le k$ for the initial budget and $b$ for the budget
at a recursive call. All previously fixed choices remain fixed.
If $b=0$ or no constraints remain, select no further items.
Otherwise, let $\mathcal Z$ contain the current constraints
satisfied when no further item is selected, and put
\[
 W:=\sum_{C\in\mathcal Z}w(C).
\]
The thresholds $\eta_1,\eta_2$ are chosen below and remain
fixed throughout the search.

Each branch fixes choices for some undecided items.
Subtract the cost of the items fixed to $1$ from $b$ and
restrict each current constraint to the fixed choices.
Remove constraints that become always satisfied or always
unsatisfied, since their contribution is the same for every
completion. Recurse on the remaining items and constraints
with the reduced budget.

If $W>0$, define the \emph{weighted degree} of an undecided
item $y$ within $\mathcal Z$ by
\[
 \deg_0(y):=
 \sum_{\substack{C\in\mathcal Z:\\y\in\scope(C)}}w(C).
\]
Let $\mathcal K_{\mathrm{hi}}$ contain the items with
$\deg_0(y)\ge\eta_1W/b$.
Each constraint involves at most $d$ items, so
\[
 \sum_y\deg_0(y)\le dW,
 \qquad
 |\mathcal K_{\mathrm{hi}}|\le\frac{db}{\eta_1}.
\]
Guess which items of $\mathcal K_{\mathrm{hi}}$ are selected
by creating one branch for every subset of total cost at
most $b$, including the empty subset.
Fix the items in that subset to $1$ and all other items of
$\mathcal K_{\mathrm{hi}}$ to $0$.
After restricting the constraints, also remove from the
recursive instance every remaining constraint originally in
$\mathcal Z$ that is still satisfied when no further item
is selected. The proof counts its weight before the remaining
choices are known; the total weight later lost is at most
$\eta_1W$, as shown below.

If $W=0$, no current constraint is satisfied without
selecting another item. For each undecided item $y$, let
$\mu(y)$ be its \emph{single-item value}: the total weight
satisfied by selecting $y$ alone among the undecided items,
with all fixed choices retained.
For distinct items $y,z$, let $I_{\rm item}(y\to z)$ be
the part of $\mu(y)$ contributed by constraints involving $z$.

For an item $y$ with $\mu(y)>0$, define its
\emph{interaction candidates} by
\[
 \mathsf{Bad}(y):=
 \{z\ne y:I_{\rm item}(y\to z)\ge(\eta_2/b)\mu(y)\}.
\]
Put $\mathsf{Bad}(y)=\varnothing$ when $\mu(y)=0$.
Every constraint counted by $\mu(y)$ involves at most
$d-1$ other items. Consequently,
\[
 \sum_{z\ne y}I_{\rm item}(y\to z)\le(d-1)\mu(y),
 \qquad
 |\mathsf{Bad}(y)|\le\frac{(d-1)b}{\eta_2}.
\]
Set $\ell:=\lfloor(d-1)b^2/\eta_2\rfloor+1$.
For each cost bound $c\in\{1,\ldots,b\}$, the
\emph{top candidates} $L_c$ are the $\ell$ items with
largest positive values $\mu(y)$ among those with
$\kappa(y)\le c$, or all such items if fewer exist.
Use a fixed item order to break ties.
The branch set is the union of the top candidates and their
interaction candidates:
\[
 \mathcal B:=
 \bigcup_{c=1}^b\left(
 L_c\cup\bigcup_{y\in L_c}\mathsf{Bad}(y)\right).
\]
In particular,
$|\mathcal B|\le b\ell(1+(d-1)b/\eta_2)$.
Create one branch selecting each $y\in\mathcal B$ with
$\kappa(y)\le b$, leaving all other undecided items available.
Also keep the choice of selecting no further items as a
candidate, without recursing on it.

For each branch, combine the recursive solution with the
choices fixed in that branch. Evaluate each resulting
assignment on all constraints of the current instance,
including those removed before recursion, and return an
assignment with largest value.

Also record $\rho^{-1}(1)$ as a feasible candidate, setting every
other variable to $0$. Repeat for $[3(d_{\max}+1)]^k$
independent colorings and return the recorded solution with
largest $F$-value, giving constant success probability.
For a deterministic version, use
Lemma~\ref{arityd-direct-lem-universal-isolation}.

More precisely, put $q=k(d_{\max}+1)$.
That lemma constructs at most $(q+1)^{O(k)}\log(n+2)$ colorings in
$(q+1)^{O(k)}\ell_{\mathrm{in}}^{O(1)}$ time, with a good coloring
for every feasible $S$.
A family of maps to $[q^2]$ assigns distinct labels to
$S$ and its external neighbors under some map.
For each map, enumerate the subsets of at most $k$ labels and color
exactly those labels green. Choosing the labels of $S$ makes all its
external neighbors red. Enumerating this family gives the deterministic
algorithm claimed in Theorem~\ref{thm:arityd-complete-dichotomy}.

\smallskip
\noindent\emph{Why the search works.}
Fix a feasible solution $S$. One of the initial guesses
agrees with $S$. Consider this guess and a coloring good
for $S$. Every green component meeting $S$ lies within
$S$, so preprocessing keeps it. Every component selected
in advance also lies in $S$. Thus the remaining items
of $S$ fit the remaining budget.

Write $F_{\rm mask}(S)$ for the value of the corresponding
item assignment. A satisfied original constraint is counted
unless its selected unit witnesses lie in different green
components. Assign each omitted constraint to one such
witness pair $x,y$. Since these vertices are not joined
by an edge, $J(x,y)\le\theta(M(x)+M(y))$.
Summing over pairs and using
$\sum_{x\in S}M(x)\le c_d\OPT$ from~\eqref{eq:overview-score-bound}
gives
\[
 F_{\rm mask}(S)\ge F(S)-\theta k c_d\OPT.
\]
We next bound the additional loss caused by the search.

Let $\sigma$ initially be the item assignment determined by
$S$. At each recursive call, it specifies the choices for
the remaining undecided items. We follow a branch agreeing
with these choices, changing $\sigma$ only when a replacement
is needed.

When $W>0$, follow the branch selecting exactly the items
of $\mathcal K_{\mathrm{hi}}$ selected by $\sigma$.
This fixes all items of large weighted degree consistently
with $\sigma$, and its remaining choices fit the reduced budget.
A removed constraint whose weight was counted before the
remaining choices were known can later become unsatisfied
only if some selected remaining item occurs in its scope.
Every remaining item is outside $\mathcal K_{\mathrm{hi}}$,
so any feasible completion loses weight at most
\[
 \sum_{\text{selected remaining items }y}\deg_0(y)
 \le b\cdot\frac{\eta_1W}{b}
 =\eta_1W.
\]
This bound applies to the completion returned by the recursive
call as well as to the comparison assignment.

When $W=0$, if $\sigma$ satisfies no current constraint,
selecting no further items loses no additional value.
Otherwise, if $\sigma$ selects an item of $\mathcal B$,
follow the branch selecting that item.

Suppose that $\sigma$ has positive current value and selects
no item of $\mathcal B$. Choose a current constraint satisfied
by $\sigma$, and let $C$ be the original constraint from
which it was constructed. The selected unit witnesses of
$C$ lie in a single component.

Keep this component selected and set all other undecided
items to $0$, leaving fixed choices unchanged. This changes
only control entries of $C$ from $1$ to $0$, so the \Direct
property preserves satisfaction of $C$. Its selected unit
witnesses remain in the same component, so the item constraint
also remains satisfied.

If this component were already fixed to $1$, selecting no
undecided items would satisfy the current constraint,
contradicting $W=0$. Thus the component is an undecided item
$u$ selected by $\sigma$. Selecting $u$ alone among the
undecided items satisfies this constraint, so $\mu(u)>0$.
Put $c:=\kappa(u)$.

Since $\mu(u)>0$ and $\kappa(u)=c$, item $u$ is eligible
for $L_c$. However, $u\notin L_c$, because
$L_c\subseteq\mathcal B$ and $\sigma$ selects no item of
$\mathcal B$. Thus $|L_c|=\ell$, and every item in $L_c$
costs at most $c$ and has single-item value at least $\mu(u)$.

Each selected item $z$ has at most $(d-1)b/\eta_2$
interaction candidates. Since $\sigma$ selects at most $b$ items,
the union of their interaction candidates has size at most
$(d-1)b^2/\eta_2<\ell$.
Choose $v\in L_c$ outside this union.
Then $\kappa(v)\le\kappa(u)$ and $\mu(v)\ge\mu(u)$.
Also, $v$ is unselected because $L_c\subseteq\mathcal B$.

For every selected item $z$, we have
$v\notin\mathsf{Bad}(z)$ by the choice of $v$.
We also have $z\notin\mathsf{Bad}(v)$, since
$\mathsf{Bad}(v)\subseteq\mathcal B$.
Therefore,
\[
 I_{\rm item}(z\to v)\le\frac{\eta_2}{b}\mu(z),
 \qquad
 I_{\rm item}(v\to z)\le\frac{\eta_2}{b}\mu(v).
\]
Selecting $v$ alone, or any selected item $z$ alone,
together with the previously fixed choices is feasible.
Each satisfied current constraint contributes the weight
of a satisfied original constraint. Hence
$\mu(v)\le\OPT$ and $\mu(z)\le\OPT$.
Summing over the at most $b$ selected items gives
\[
 \sum_{z:\sigma(z)=1}I_{\rm item}(z\to v)
 \le\eta_2\OPT,
 \qquad
 \sum_{z:\sigma(z)=1}I_{\rm item}(v\to z)
 \le\eta_2\OPT.
\]

Replace $u$ by $v$ in $\sigma$.
Since $\kappa(v)\le\kappa(u)$, the assignment remains feasible.
Lemma~\ref{arityd-direct-lem-item-calculus} shows that removing
$u$ loses at most $\mu(u)$, whereas adding $v$ gains at least
$\mu(v)$ minus the two interaction sums.
Since $\mu(v)\ge\mu(u)$, the replacement loses at most
$2\eta_2\OPT$. Follow the branch selecting the top candidate
$v$ with this modified assignment.

The value $W$ is also at most $\OPT$.
Each recursive step therefore introduces loss at most
$\max\{\eta_1,2\eta_2\}\OPT$.
A branch selecting an item spends at least one unit of budget.
The only branch that spends no budget fixes every item of
$\mathcal K_{\mathrm{hi}}$ to $0$.
All constraints previously in $\mathcal Z$ are then removed,
and every retained constraint rejects the choice of selecting
no further items. Thus a branch that spends no budget is
followed by a step that either spends budget or terminates.
There are at most $2b_0+1$ recursive steps along any path.
Combining the loss bounds gives a returned selected set $A$
satisfying
\[
 F(A)\ge F_{\rm mask}(S)
 -(2b_0+1)\max\{\eta_1,2\eta_2\}\OPT.
\]
Here we add the losses along a recursion path. The formal proof
of Theorem~\ref{arityd-direct-thm-item-hus} uses induction on
$\phi(\mathcal J)=2b+\mathbf1[W>0]$.

\smallskip
\noindent\emph{Error and running time.}
Set $\theta=\varepsilon/(4kc_d)$.
For each coloring, let $b_0$ be the initial remaining budget
of its item search and set
$\eta_1=\varepsilon/[4(2b_0+1)]$ and $\eta_2=\eta_1/2$.
Passing to item constraints omits at most $\varepsilon\OPT/4$,
and the search loses at most another $\varepsilon\OPT/4$.
Thus $F(A)\ge F(S)-\varepsilon\OPT/2$.
Together with the common reduction's loss, this gives
the required $(1-\varepsilon)$ approximation.

Choose $C=(1+dk/\varepsilon)^{O(1)}\ge2$ so that, for each
$r\ge1$, a call has at most $C^r$ children that reduce the
budget by $r$. There is at most one child that leaves the
budget unchanged, and its next step spends budget or terminates.
Counting children by $r$ gives at most $(4C)^k$ nodes.
Each item constraint involves at most $d$ items, so its
satisfying assignments can be listed using at most $2^d$
evaluations of the original constraint.
At each node and edge, restricting constraints, computing
degrees and scores, and forming candidate lists take
$k^{O(1)}2^d\ell_{\mathrm{in}}^{O(1)}$ time.
Including the colorings and initial guesses gives total time
$2^{O(d+k\log(1+dk/\varepsilon))}\ell_{\mathrm{in}}^{O(1)}$;
see Section~\ref{arityd-direct-sec-main}.

\smallskip
\noindent\textbf{Hardness: a reduction from Clique.}
Suppose $\Gamma_+$ contains a relation that fails \Direct and a relation
that fails \Rooted. They may be the same relation.
We reduce \clique{} to an unweighted instance $\mathcal I$ with budget
$k'$ and integer target $t\ge1$ satisfying
\[
 \OPT(\mathcal I)\le t,
 \qquad
 \OPT(\mathcal I)=t
 \iff G\text{ has a }k\text{-clique}.
\]
Both $k'$ and $t$ are bounded by functions of $k,d$.

\smallskip
\noindent\emph{An inaccessible relation.}
Choose a minimal satisfying support $A$ consisting only of controls,
and put $q=|A|\ge2$.
For each graph edge, place its endpoints in two coordinates of $A$,
use $q-2$ shared variables in the other coordinates of $A$, and use
fresh variables in every coordinate outside $A$.
Set $k'=k+q-2$ and $t=\binom{k}{2}$.
If all shared variables are selected, an edge constraint with no fresh
variable selected is satisfied exactly when both endpoints are selected.
With $s$ selected vertices and $r$ selected fresh variables,
$s+r\le k$ and the score is at most $\binom{s}{2}+r$.
For $k\ge q+3$, reaching $t$ forces $s=k$ and a clique.
If a shared variable is missing, every satisfied constraint needs its
own selected fresh variable, giving score at most $k'<t$.

\smallskip
\noindent\emph{Two accessible relations.}
Otherwise, choose accessible $P,Q\in\Gamma_+$, with $P$ failing
\Direct and $Q$ failing \Rooted. From these relations, we construct
a fixed \emph{switch gadget}: a small collection of constraints with
a switch variable $x$, private variables, and shared variables called
\emph{pins}.

With all pins set to $1$, selecting a designated set of $h$ private
variables satisfies $h$ constraints when $x=0$ and $h+1$ when $x=1$.
Any other choice of private variables satisfies at most as many
constraints as the number of private variables selected, regardless
of $x$. Thus, the switch gives one extra satisfied constraint only
for the designated choice.

Represent each graph vertex $v$ by a variable $x_v$. For each edge
$uv$, use two copies of the gadget with switches $x_u,x_v$,
shared private variables, and separate pin groups.
The designated choice then gives one extra point for each selected
endpoint. Selecting both $x_u$ and $x_v$ gives two extra points using
the same $h$ private variables, so the construction rewards edges
whose endpoints are both selected.

Additional constraints ensure that every pin is set to $1$ in any
assignment reaching the target score. We choose a budget $k'$ and
a target $t$, both depending only on $k$ and $d$, so that the optimum
is at most $t$ and equals $t$ exactly when the graph contains a
$k$-clique. The instance is unweighted, and no two constraints use
the same set of variables.

Finally, an \epas{} with $\varepsilon=1/(2t)$ must return score $t$
when a clique exists and at most $t-1$ otherwise, since scores are integers.
For fixed $\Gamma$, this decides \clique{} in \fpt{} time.
The full construction and proof are in
Section~\ref{arityd-dist-hard-subsec-necessity}.

\section{Classification and common reductions}
\label{sec:arityd-complete-distinct}

The following sections prove Theorem~\ref{thm:arityd-complete-dichotomy}: a
classification of the fixed finite Boolean constraint languages $\Gamma$ for
which \maxwtcsp{}$(\Gamma)$ admits an \epas{} parameterized by the budget $k$.
We prove the approximation guarantee for each of the two repair rules,
\Direct and \Rooted, and prove hardness when neither rule holds for all
the nonempty zero-free relations. As in Section~\ref{sec:preliminaries},
each constraint uses distinct variables, listed in coordinate order.
Section~\ref{ss:arityd-binary-specialization} proves
Corollary~\ref{cor:arity-two} for relations of arity at most two.

\subsection{Classification of zero-free relations}

A relation is \emph{zero-valid} if it contains the all-zero tuple and
\emph{zero-free} otherwise.  Empty relations can be discarded because their
constraints are never satisfied.  Define
\[
 \Gamma_+
   :=\{R\in\Gamma:R\ne\varnothing,\ \mathbf0\notin R\},
\]
the nonempty zero-free part of the original language.

For a nonempty zero-free relation $R\subseteq\{0,1\}^r$, let $e_i$ be the
tuple whose unique $1$ is in coordinate $i$. Coordinate $i$ is a
\emph{unit-witness coordinate} if $e_i\in R$ and a \emph{control}
otherwise. For a tuple
$\alpha=(\alpha_1,\ldots,\alpha_r)$, its support is
$\supp(\alpha):=\{i:\alpha_i=1\}$.

\begin{definition}
\label{def:arityd-two-rules}
Let $R\subseteq\{0,1\}^r$ be a nonempty zero-free relation.
\begin{enumerate}
\item $R$ is \Direct if, for every $\alpha\in R$ and every
      control coordinate $c$ with $\alpha_c=1$,
      \[
        \alpha-e_c\in R.
      \]
      Here $\alpha-e_c$ changes coordinate $c$ from $1$ to $0$
      and leaves every other coordinate unchanged.

\item $R$ is \Rooted if, for every $\alpha\in R$, there is a
      unit-witness coordinate $u$ with $\alpha_u=1$ such that,
      for every $B\subseteq\supp(\alpha)$ consisting only of
      control coordinates, $R$ contains a tuple $\beta$ with
      \[
        \supp(\beta)=\{u\}\cup B.
      \]
\end{enumerate}
\end{definition}

\begin{lemma}
\label{lem:arityd-testing-rules}
For a nonempty zero-free relation $R$ of arity $r$, the
\Direct and \Rooted rules can be checked using
$O(r2^r)$ and $O(r4^r)$ membership queries, respectively.
For a finite language $\Gamma$ of arity at most $d$, checking
both rules for every relation in $\Gamma_+$ uses
$O(|\Gamma|\,d\,4^d)$ queries.
\end{lemma}

\begin{proof}
Each membership query asks whether a specified tuple belongs
to a relation. For each relation $R\in\Gamma$ of arity $r$,
first query the all-zero tuple. If it belongs to $R$, skip
this relation: zero-valid relations need not satisfy either
repair rule.

Otherwise, enumerate all $2^r$ Boolean tuples and record
which belong to $R$. Discard $R$ if none belongs to it.
The remaining relations are exactly those in $\Gamma_+$.
For each such relation, the answers for the unit tuples $e_i$
identify which coordinates are unit witnesses and which are controls.

For \Direct, for each $\alpha\in R$ and each control coordinate
$c$ with $\alpha_c=1$, query whether
$\alpha-e_c\in R$. There are at most $r$ such tests per tuple,
giving $O(r2^r)$ queries.

For \Rooted, for each $\alpha\in R$, try every unit-witness
coordinate $u$ with $\alpha_u=1$. For each such $u$ and every
$B\subseteq\supp(\alpha)$ consisting only of control coordinates,
query the
tuple whose support is $\{u\}\cup B$.
The rule holds exactly when every $\alpha\in R$ has a choice
of $u$ for which all these queries succeed.
There are at most $r$ choices of $u$ and $2^r$ subsets $B$
per tuple, giving $O(r4^r)$ queries.

Summing over $\Gamma$ gives $O(|\Gamma|\,d\,4^d)$ queries.
For fixed $\Gamma$ and $d$, this bound is constant.
\end{proof}

\begin{definition}
\label{def:arityd-accessibility}
A nonempty zero-free relation $R\subseteq\{0,1\}^r$ is
\emph{accessible} if every $\alpha\in R$ satisfies
\[
 \text{there is }u\in[r]\text{ with }\alpha_u=1\text{ and }e_u\in R.
 \tag{\(*\)}\label{eq:arityd-accessibility}
\]
It is \emph{inaccessible} otherwise.
\end{definition}

\begin{lemma}
\label{lem:arityd-repair-accessibility}
Every nonempty zero-free relation that is \Direct or \Rooted
is accessible.
\end{lemma}

\begin{proof}
For \Rooted, the definition requires every $\alpha\in R$ to have
$\alpha_u=1$ at some unit-witness coordinate $u$.

For \Direct, suppose that some $\alpha\in R$ has no $1$
in a unit-witness coordinate. Every coordinate set to $1$
is then a control. Set these coordinates to $0$ one at a time.
The \Direct rule preserves membership in $R$ at every step,
so the all-zero tuple belongs to $R$. This contradicts
zero-freeness.
\end{proof}

Both sufficient conditions may hold. A language that satisfies neither
condition belongs to the hard case.

\paragraph*{Order of the proofs.}
Section~\ref{ss:arityd-zero-valid-lifting-distinct} reduces arbitrary
instances to zero-free instances in a form that both algorithms can use.
Section~\ref{ss:arityd-random-isolation} bounds the sum of single-variable
scores, which both algorithms use to control their approximation error.
Section~\ref{ss:arityd-rooted} gives the algorithm for \Rooted relations,
and Section~\ref{arityd-direct-sec-main} gives the algorithm for \Direct
relations. Section~\ref{arityd-dist-hard-subsec-necessity} proves hardness
when neither rule holds for all relations in $\Gamma_+$.
Section~\ref{ss:arityd-proof} combines these results to prove the
classification, and Section~\ref{ss:arityd-binary-specialization} handles
relations of arity at most two.
\subsection{Reduction to zero-free constraints}
\label{ss:arityd-zero-valid-lifting-distinct}

Fix a Boolean constraint language $\Gamma$ of arity at most $d$.
We reduce approximation over $\Gamma$ to approximation over its
zero-free part $\Gamma_+$, while respecting prescribed variable values.
The same reduction applies to both the \Direct and \Rooted algorithms.

Suppose $\mathcal A$ is an algorithm with the following guarantee.
Given an instance $\mathcal J$ over $\Gamma_+$ with budget $k$,
a partial assignment $\rho$ with $|\rho^{-1}(1)|\le k$,
and a rational $\delta\in(0,1)$, it returns a feasible assignment
$\tau$ extending $\rho$ such that
\[
 F^{\mathcal J}(\tau)
 \ge F^{\mathcal J}(\tau')-\delta\OPT(\mathcal J)
\]
for every feasible assignment $\tau'$ extending $\rho$.
Here $\OPT(\mathcal J)$ is measured before imposing $\rho$.
The running time of $\mathcal A$ is
$g(k,d,\delta)\ell_{\mathrm{in}}^{O(1)}$.

\begin{lemma}
\label{lem:arityd-zero-valid-lifting}
Using $\mathcal A$, for every rational $\varepsilon\in(0,1)$,
we can compute a $(1-\varepsilon)$-approximate solution
to any instance over $\Gamma$ in time
\[
 2^{O(k\log(1+dk/\varepsilon))}
 g(k,d,\varepsilon/2)\ell_{\mathrm{in}}^{O(1)}.
\]
\end{lemma}

We prove the lemma by constructing a small family of partial assignments.
Consider an instance $\mathcal I=(X,\mathcal C,w,k)$ over $\Gamma$.
Discard empty relations, and split the remaining constraints into
the zero-valid constraints $\mathcal C_0$ and the zero-free constraints
$\mathcal C_+$. Let
$\mathcal I_+:=(X,\mathcal C_+,w|_{\mathcal C_+},k)$.
Write $F_0$ and $F_+$ for the two contributions to the objective, so
$F=F_0+F_+$, and put $\OPT_+:=\OPT(\mathcal I_+)$.
Define
\[
 w_0:=\sum_{C\in\mathcal C_0}w(C),
 \qquad
 D_0(v):=\sum_{\substack{C\in\mathcal C_0:\ v\in\scope(C)}}w(C).
\]
Thus $w_0$ is the value of the all-zero assignment, and $D_0(v)$ bounds
the weight of zero-valid constraints that changing $v$ can affect.
Nonnegative weights give $\OPT_+\le\OPT$, and the all-zero assignment
gives $w_0\le\OPT$.

For $k\ge1$ and a rational $\varepsilon\in(0,1)$, define
\[
 H:=\left\{v\in X:D_0(v)>\frac{\varepsilon w_0}{4k}\right\}.
\]
For each $T\subseteq H$ with $|T|\le k$, let
$\rho_T:H\to\{0,1\}$ be the restriction of $\tau_T$ to $H$.
Thus $\rho_T$ assigns $1$ to $T$ and $0$ to $H\setminus T$.
The choice of $H$ adapts Manurangsi's degree-threshold argument for
clauses with negative literals~\cite{DBLP:journals/tcs/Manurangsi25}.
The next lemma bounds the loss on zero-valid constraints once these
values are fixed.

\begin{lemma}
\label{lem:arityd-zero-valid-family}
The family $\{\rho_T:T\subseteq H,\ |T|\le k\}$ has at most
$2^{O(k\log(1+dk/\varepsilon))}$ members.
For every $S\subseteq X$ with $|S|\le k$, the assignment $\tau_S$
extends $\rho_T$ for $T=S\cap H$, and every feasible assignment
$\tau$ extending this $\rho_T$ satisfies
\[
 |F_0(\tau)-F_0(S)|\le\frac{\varepsilon}{2}w_0.
\]
\end{lemma}

\begin{proof}
If $w_0=0$, then $D_0(v)=0$ for every $v$, so $H=\varnothing$ and the
family has one member. Otherwise, since each constraint has arity at
most $d$,
\[
 |H|\frac{\varepsilon w_0}{4k}
 \le\sum_{v\in H}D_0(v)
 \le\sum_{v\in X}D_0(v)
 \le dw_0.
\]
Hence $|H|\le4dk/\varepsilon$ in either case. The number of members
is at most
\[
 \sum_{i=0}^{\min\{k,|H|\}}\binom{|H|}{i}
 \le(|H|+1)^k
 \le(1+4dk/\varepsilon)^k.
\]

Now fix $S$ and put $T=S\cap H$. By construction, $\tau_S$ extends
$\rho_T$. Let $\tau$ be any feasible assignment extending $\rho_T$,
and put $A:=\tau^{-1}(1)$. Then $|A|\le k$ and
$A\cap H=S\cap H$. Every zero-valid constraint whose satisfaction
changes contains a variable on which $A$ and $S$ differ. Such a variable
lies in $(A\cup S)\setminus H$. This set has at most $2k$ variables,
and each has $D_0(v)\le\varepsilon w_0/(4k)$. Therefore
\begin{equation}
 \label{eq:arityd-zero-compare}
 |F_0(\tau)-F_0(S)|
 \le\sum_{v\in(A\cup S)\setminus H}D_0(v)
 \le\frac{\varepsilon}{2}w_0.
\end{equation}
\end{proof}

\begin{proof}[Proof of Lemma~\ref{lem:arityd-zero-valid-lifting}]
If $k=0$, return the all-zero assignment. Otherwise construct
$\mathcal I_+$ and the family of partial assignments above.
For each $T\subseteq H$ with $|T|\le k$, run
$\mathcal A(\mathcal I_+,\rho_T,\varepsilon/2)$.
Evaluate each returned assignment on the original instance and return
one with the largest value. Every call uses the original zero-free
relations; $\rho_T$ only restricts the permitted assignments.

Let $O$ be an optimal selected set for the original instance, and let
$\tau$ be the assignment returned by
$\mathcal A(\mathcal I_+,\rho_{O\cap H},\varepsilon/2)$.
Since $\tau_O$ is a feasible assignment extending $\rho_{O\cap H}$,
the guarantee of $\mathcal A$ gives
\begin{equation}
 \label{eq:arityd-core-compare}
 F_+(\tau)\ge F_+(O)-\frac{\varepsilon}{2}\OPT_+.
\end{equation}
Lemma~\ref{lem:arityd-zero-valid-family} gives
$F_0(\tau)\ge F_0(O)-\frac{\varepsilon}{2}w_0$.
Adding the two bounds yields
\[
 F(\tau)\ge F(O)-\frac{\varepsilon}{2}(\OPT_++w_0)
          \ge(1-\varepsilon)\OPT.
\]
The best returned assignment has at least this value.

By Lemma~\ref{lem:arityd-zero-valid-family}, there are at most
$2^{O(k\log(1+dk/\varepsilon))}$ calls to $\mathcal A$.
Each call takes time $g(k,d,\varepsilon/2)\ell_{\mathrm{in}}^{O(1)}$.
Constructing $\mathcal I_+$, preparing each $\rho_T$, and evaluating
each returned assignment take polynomial time in $\ell_{\mathrm{in}}$.
For a fixed language, membership tests take constant time; for input
relation descriptions, we use the assumption in
Remark~\ref{rem:arityd-uniform-algorithms}.
This proves the running-time bound.
\end{proof}

\subsection{Bounding single-variable scores}
\label{ss:arityd-random-isolation}

For any instance with zero-free relations, the next lemma bounds
the sum of single-variable scores over a set of at most $k$ variables
by $c_d$ times the optimum. The factor $c_d$ depends only on the
arity bound $d$. Both algorithms use this bound to control their
approximation loss.
Put $c_1:=1$ and, for $d\ge2$,
\begin{equation}
\label{eq:arityd-cd}
   c_d:=d\left(1-\tfrac1d\right)^{-(d-1)}
       =\frac{d^{\,d}}{(d-1)^{d-1}}\le ed .
\end{equation}

\begin{lemma}
\label{lem:arityd-random-isolation}
Let $\mathcal J=(X,\mathcal C,w,k)$ be an instance of arity at most
$d$, with $k\ge1$ and nonempty zero-free relations.
For $v\in X$, let $M(v):=F^{\mathcal J}(\tau_{\{v\}})$.
Then every $Q\subseteq X$ with $|Q|\le k$ satisfies
\begin{equation}
\label{eq:arityd-random-isolation}
 \sum_{v\in Q}M(v)\le c_d\OPT(\mathcal J).
\end{equation}
\end{lemma}

\begin{proof}
For $d=1$, every nonempty zero-free relation accepts only the
value $1$. Each constraint contributes to the score of its
unique variable. Hence
\[
 \sum_{v\in Q}M(v)=F^{\mathcal J}(\tau_Q)
 \le\OPT(\mathcal J).
\]
Since $c_1=1$, this proves the bound.

Now suppose $d\ge2$. Retain each variable of $Q$ independently
with probability $p:=1/d$, obtaining $Q'\subseteq Q$.
The assignment $\tau_{Q'}$ selects exactly the retained variables.

For each $v\in Q$, consider every constraint $C$ satisfied by
$\tau_{\{v\}}$. If there is no such constraint, then $M(v)=0$
and $v$ contributes nothing to the sum.
For any such $C$, we have $v\in\scope(C)$: otherwise
$\tau_{\{v\}}[C]$ would be the all-zero tuple, contradicting
zero-freeness.

Consider the event that $v$ is the only retained variable
in $\scope(C)$:
\[
 Q'\cap\scope(C)=\{v\}.
\]
This requires retaining $v$ and omitting the other
$|Q\cap\scope(C)|-1$ variables of $Q$ in the scope.
By independence, its probability is
\[
 p(1-p)^{|Q\cap\scope(C)|-1}
 \ge p(1-p)^{d-1}
 =\frac1{c_d},
\]
where the inequality uses $|Q\cap\scope(C)|\le d$.
On this event,
$\tau_{Q'}[C]=\tau_{\{v\}}[C]$, so $\tau_{Q'}$ satisfies $C$.

For a fixed constraint $C$, these events are disjoint for
different choices of $v$, since $C$ cannot have two different
variables as its only retained variable. Thus the probability
that $\tau_{Q'}$ satisfies $C$ is at least the sum of these
event probabilities. Multiplying by the nonnegative weights
and summing over constraints gives
\begin{align*}
 \mathbb E[F^{\mathcal J}(\tau_{Q'})]
 &=\sum_{C\in\mathcal C}w(C)\Pr[\tau_{Q'}\models C]\\
 &\ge\frac1{c_d}
   \sum_{C\in\mathcal C}
   \sum_{\substack{v\in Q:\\ \tau_{\{v\}}\models C}}w(C)\\
 &=\frac1{c_d}\sum_{v\in Q}M(v).
\end{align*}
The last equality follows by changing the order of summation:
for each $v$, the inner sum is precisely $M(v)$.

Every outcome has $|Q'|\le|Q|\le k$, so $\tau_{Q'}$ is feasible
and its value is at most $\OPT(\mathcal J)$. Therefore
\[
 \frac1{c_d}\sum_{v\in Q}M(v)
 \le\mathbb E[F^{\mathcal J}(\tau_{Q'})]
 \le\OPT(\mathcal J),
\]
which proves \eqref{eq:arityd-random-isolation}.
\end{proof}

\section{Algorithm for \Rooted relations}
\label{ss:arityd-rooted}

In this section, we provide an \epas{} for \maxwtcsp{}$(\Gamma)$,
parameterized by $k$, when every relation in $\Gamma_+$ is \Rooted.
The algorithm builds on the branching and exchange ideas of
Jain et al.~\cite{jain_parameterized_2023}.

For the rest of the zero-free analysis, fix an instance
$\mathcal I=(X,\mathcal C,w,k)$ of arity at most $d$ whose relations are
nonempty, zero-free and \Rooted. Empty relations may be discarded, since
their constraints are never satisfied. We assume $k\ge1$, since for
$k=0$ the all-zero assignment is the unique feasible assignment.

For each constraint $C\in\mathcal C$, define
\[
 U_C:=\{x\in\scope(C):\tau_{\{x\}}\models C\},
 \qquad
 N_C:=\scope(C)\setminus U_C.
\]
Thus $U_C$ contains the variables that satisfy $C$ when selected
alone, and $N_C$ contains the remaining variables. We call these
unit-witness variables and controls, respectively. Abbreviate
\[
   F(S):=F^{\mathcal I}(\tau_S)
   \qquad(S\subseteq X).
\]

\subsection{Single-variable scores and candidate scores}

For $v\in X$, the \emph{single-variable score} $M(v)$ is the
total weight of constraints satisfied by selecting $v$ alone.
For distinct $v,x\in X$, let $I(v\to x)$ be the total weight
of these constraints whose scope also contains $x$. Define
\begin{align}
\label{eq:arityd-rooted-global-mass}
   M(v)&:=F(\{v\}),\\
\label{eq:arityd-rooted-global-interaction}
   I(v\to x)&:=
   \sum_{\substack{C\in\mathcal C:\{v,x\}\subseteq\scope(C)\\
                    \tau_{\{v\}}\models C}}w(C)
   \qquad(v\ne x).
\end{align}
Both quantities depend only on the instance.

Let $K\subseteq X$ be the set of variables required to be
selected. For $v\in X\setminus K$, define
\[
  \mathcal A_K(v):=
  \left\{C\in\mathcal C:
  \begin{array}{l}
       v\in U_C,\quad K\cap U_C=\varnothing,\\
       \tau_{K\cup\{v\}}\models C
  \end{array}\right\}.
\]
Thus $\mathcal A_K(v)$ consists of the constraints satisfied
by selecting $K\cup\{v\}$ for which $v$ is the only selected
unit-witness variable. The \emph{candidate score at $K$} is
\begin{equation}
\label{eq:arityd-rooted-context-mass}
   m_K(v):=\sum_{C\in\mathcal A_K(v)}w(C).
\end{equation}
Every constraint in $\mathcal A_K(v)$ is also satisfied by
selecting $v$ alone, since $v\in U_C$. Therefore,
\begin{equation}
\label{eq:arityd-rooted-context-below-global}
   0\le m_K(v)\le M(v).
\end{equation}
The candidate score $m_K(v)$ depends on $K$. It need not equal
$F(K\cup\{v\})-F(K)$: it counts only constraints in
$\mathcal A_K(v)$ and does not subtract the weight of constraints
that become unsatisfied when $v$ is added.

\begin{lemma}
\label{lem:arityd-rooted-outgoing}
For every $v\in X$,
\[
   \sum_{x\ne v}I(v\to x)\le(d-1)M(v).
\]
\end{lemma}

\begin{proof}
Each constraint counted by $M(v)$ contains at most $d-1$ other variables.
Its weight therefore appears in at most $d-1$ of the sums $I(v\to x)$.
Summing over these constraints proves the bound.
\end{proof}

The next two lemmas bound the loss from removing variables and
the gain from adding variables, in terms of the candidate scores
$m_K$ and the sums $I$.

\begin{lemma}
\label{lem:arityd-rooted-projection}
For all $K\subseteq S\subseteq X$,
\begin{equation}
\label{eq:arityd-rooted-projection}
 F(S)-F(K)\le
   \sum_{v\in S\setminus K}m_K(v)
   +\sum_{\substack{u\in K\\x\in S\setminus K}}I(u\to x).
\end{equation}
\end{lemma}

\begin{proof}
The difference $F(S)-F(K)$ is at most the total weight of
constraints satisfied by $\tau_S$ but not by $\tau_K$.
We show that each such constraint contributes its weight
to at least one sum on the right.

Consider such a constraint $C$. First suppose that
$K\cap U_C=\varnothing$. Since
$K\cap N_C\subseteq S\cap N_C$, the \Rooted property gives
a variable $v\in S\cap U_C$ such that $C$ remains satisfied
when we keep only $v$ and the controls in $K\cap N_C$.
These are precisely the variables of $\scope(C)$ selected
by $\tau_{K\cup\{v\}}$. Hence
\[
   \tau_{K\cup\{v\}}\models C.
\]
Moreover, $v\notin K$. Thus $C\in\mathcal A_K(v)$, and
$w(C)$ is included in $m_K(v)$.

Now suppose that $K\cap U_C\ne\varnothing$, and choose
$u\in K\cap U_C$. By the definition of $U_C$,
$\tau_{\{u\}}\models C$. Since $C$ is satisfied by $\tau_S$
but not by $\tau_K$, its scope must contain a variable
$x\in S\setminus K$. Therefore, $w(C)$ is included in
$I(u\to x)$.

These cases cover every lost constraint. Since weights
are nonnegative, counting a constraint more than once
only increases the bound. This proves
\eqref{eq:arityd-rooted-projection}.
\end{proof}

\begin{lemma}
\label{lem:arityd-rooted-insertion}
Let $K\subseteq X$ and let
$P=\{p_1,\ldots,p_s\}\subseteq X\setminus K$, with its
distinct elements listed in a fixed order. Put
$P_{<i}:=\{p_1,\ldots,p_{i-1}\}$. Then
\begin{align}
\label{eq:arityd-rooted-insertion}
 F(K\cup P)-F(K)\ge{}&
 \sum_{i=1}^s m_K(p_i)
 -\sum_{i=1}^s\sum_{j<i}I(p_i\to p_j)\notag\\
 &-\sum_{i=1}^s\sum_{u\in K\cup P_{<i}}I(u\to p_i).
\end{align}
\end{lemma}

\begin{proof}
Put $K_i:=K\cup\{p_1,\ldots,p_i\}$ for $i=0,\ldots,s$,
so $K_0=K$ and $K_{i-1}=K\cup P_{<i}$.
We bound the gain and loss when $p_i$ is added.
By Lemma~\ref{lem:arityd-repair-accessibility}, every
assignment satisfying a constraint $C$ selects at least
one variable in $U_C$.

First consider $C\in\mathcal A_K(p_i)$ whose scope contains
no variable from $P_{<i}$. Before adding $p_i$, the assignment
on $C$ agrees with $\tau_K[C]$. Since
$K\cap U_C=\varnothing$, this assignment does not satisfy $C$.
After adding $p_i$, the assignment on $C$ agrees with
$\tau_{K\cup\{p_i\}}[C]$, which satisfies $C$ by the definition
of $\mathcal A_K(p_i)$. Thus $C$ gives a new gain of $w(C)$.

A constraint in $\mathcal A_K(p_i)$ containing an earlier
variable $p_j$ may fail to give this gain. Its weight is
included in $I(p_i\to p_j)$, since selecting $p_i$ alone
satisfies $C$. Therefore, the total weight of candidate
constraints that may fail to give a new gain is at most
\[
   \sum_{j<i}I(p_i\to p_j).
\]
Counting a constraint more than once only increases this
upper bound.

Now consider a constraint $C$ satisfied by $\tau_{K_{i-1}}$
but not by $\tau_{K_i}$. Since only $p_i$ changes value,
$p_i\in\scope(C)$. Accessibility gives a variable
$u\in K_{i-1}\cap U_C$. Selecting $u$ alone satisfies $C$,
so $w(C)$ is included in $I(u\to p_i)$.
Thus the total weight lost is at most
$\sum_{u\in K_{i-1}}I(u\to p_i)$.

Combining the gain and loss bounds gives
\begin{align*}
 F(K_i)-F(K_{i-1})\ge{}&
 m_K(p_i)-\sum_{j<i}I(p_i\to p_j)\\
 &-\sum_{u\in K_{i-1}}I(u\to p_i).
\end{align*}
Since $K_0=K$, $K_s=K\cup P$, and
$K_{i-1}=K\cup P_{<i}$, summing the preceding inequalities gives
\begin{align*}
 F(K\cup P)-F(K)
 &=\sum_{i=1}^s\bigl(F(K_i)-F(K_{i-1})\bigr)\\
 &\ge \sum_{i=1}^s m_K(p_i)
   -\sum_{i=1}^s\sum_{j<i}I(p_i\to p_j)\\
 &\qquad
   -\sum_{i=1}^s\sum_{u\in K\cup P_{<i}}I(u\to p_i).
\end{align*}
This proves \eqref{eq:arityd-rooted-insertion}.
\end{proof}

\subsection{Choosing candidates with small losses}

Fix a rational $\delta\in(0,1)$ and a partial assignment
$\rho:H\to\{0,1\}$, where $H\subseteq X$ and
$|\rho^{-1}(1)|\le k$. Put
\[
   H^1:=\rho^{-1}(1),\qquad H^0:=\rho^{-1}(0),
   \qquad Y:=X\setminus H^0.
\]
The search starts at $K=H^1$ and maintains
$H^1\subseteq K\subseteq Y$ and $|K|\le k$.
At a call with required selected set $K$, the available variables
are those in $Y\setminus K$.
Fix an order on $X$ for breaking ties.

\paragraph*{Parameters and the graph.}
By Lemma~\ref{lem:arityd-random-isolation},
\begin{equation}
\label{eq:arityd-rooted-random-isolation}
   \sum_{v\in Q}M(v)\le c_d\OPT(\mathcal I)
   \qquad(Q\subseteq X,\ |Q|\le k),
\end{equation}
where $c_d$ is defined in \eqref{eq:arityd-cd}.
Set
\begin{equation}
\label{eq:arityd-rooted-parameters}
   \theta:=\frac{\delta}{2kc_d},
   \qquad
   d_{\max}:=\left\lceil\frac{d-1}{\theta}\right\rceil,
   \qquad
   \ell:=k(3d_{\max}+1).
\end{equation}
Construct a directed graph on $X$. For distinct $v,x\in X$,
draw an arrow $v\to x$ exactly when
$I(v\to x)>\theta M(v)$, and give it weight $I(v\to x)$.
The graph is computed once before the search. Write
\[
   E(v):=\{x\in X\setminus\{v\}:I(v\to x)>\theta M(v)\}
\]
for the outgoing neighbors of $v$.
If $M(v)=0$, Lemma~\ref{lem:arityd-rooted-outgoing} gives
$I(v\to x)=0$ for every $x\ne v$, so $E(v)=\varnothing$.
If $M(v)>0$, each outgoing arrow has weight greater than
$\theta M(v)$, and their total weight is at most $(d-1)M(v)$.
Consequently,
\begin{equation}
\label{eq:arityd-rooted-heavy-degree}
   |E(v)|\le d_{\max}\qquad(v\in X).
\end{equation}

\paragraph*{Candidate sets.}
At a call with required set $K$, define
\begin{equation}
\label{eq:arityd-rooted-ek}
   E^{\mathrm{out}}(K):=
   \left(\bigcup_{u\in K}E(u)\right)\cap(Y\setminus K),
   \qquad
   |E^{\mathrm{out}}(K)|\le |K|d_{\max}\le kd_{\max}.
\end{equation}
These are the available variables receiving an arrow from $K$.
If $x\in Y\setminus K$ lies outside $E^{\mathrm{out}}(K)$, then
\[
   I(u\to x)\le\theta M(u)\qquad(u\in K).
\]
Let $T_K$ consist of the $\ell$ available variables with largest
positive scores $m_K(v)$. If fewer than $\ell$ available variables
have positive scores, take all of them. Break ties by the fixed
order on $X$. The search branches on
\begin{equation}
\label{eq:arityd-rooted-branch-set}
   B_K:=T_K\cup E^{\mathrm{out}}(K)\subseteq Y\setminus K.
\end{equation}

\paragraph*{Replacement sets.}
When $|T_K|=\ell$, form an undirected graph $G_K$ on
$T_K\setminus E^{\mathrm{out}}(K)$. Distinct vertices $p,q$
are adjacent exactly when
\begin{equation}
\label{eq:arityd-rooted-candidate-edges}
   I(p\to q)>\theta M(p)
   \quad\text{or}\quad
   I(q\to p)>\theta M(q).
\end{equation}

\begin{lemma}
\label{lem:arityd-rooted-independent-set}
If $|T_K|=\ell$, we can deterministically find, in polynomial time,
a set of $k$ variables
\begin{equation}
\label{eq:arityd-rooted-independent-set}
   Q_K=\{q_1,\ldots,q_k\}\subseteq
   T_K\setminus E^{\mathrm{out}}(K)
\end{equation}
such that
\[
   I(u\to v)\le\theta M(u)
   \qquad
   (v\in Q_K,\ u\in K\cup Q_K,\ u\ne v).
\]
\end{lemma}

\begin{proof}
Since $|T_K|=\ell$ and
$|E^{\mathrm{out}}(K)|\le kd_{\max}$,
\[
   |T_K\setminus E^{\mathrm{out}}(K)|
   \ge \ell-kd_{\max}=k(2d_{\max}+1).
\]

Let $W$ be a nonempty subset of the vertices of $G_K$.
Every edge in the subgraph induced by $W$ comes from at
least one directed arrow with both endpoints in $W$.
There are at most $d_{\max}|W|$ such arrows, since each
vertex has at most $d_{\max}$ outgoing arrows.
The induced subgraph therefore has at most
$d_{\max}|W|$ edges and average degree at most $2d_{\max}$.
Hence it contains a vertex with at most $2d_{\max}$ neighbors.

Repeatedly remove a vertex of minimum degree, breaking
ties by the fixed order on $X$. Each vertex has at most
$2d_{\max}$ neighbors still present when it is removed.
Now color the vertices in reverse removal order.
The already colored neighbors of a vertex are exactly
those that were still present when it was removed.
Thus at most $2d_{\max}$ colors are forbidden.
Assigning the first available color uses at most
$2d_{\max}+1$ colors. This procedure is deterministic
and takes polynomial time~\cite{Diestel2000,MatulaBeck1983}.

A largest color class has at least
\[
   \frac{|T_K\setminus E^{\mathrm{out}}(K)|}
        {2d_{\max}+1}\ge k
\]
vertices. Choose such a class, breaking ties by color
number, and let $q_1,\ldots,q_k$ be its first $k$ vertices
in the fixed order on $X$.
These vertices are pairwise nonadjacent in $G_K$,
so no arrow joins two distinct variables in $Q_K$.
Moreover, no arrow goes from $K$ to $Q_K$, because
$Q_K\subseteq T_K\setminus E^{\mathrm{out}}(K)$.
The definition of the arrows now gives the required
inequality for every $v\in Q_K$ and
$u\in K\cup Q_K$ with $u\ne v$.
\end{proof}

\subsection{The search}

\paragraph*{Algorithm.}
Algorithm~\ref{alg:arityd-rooted} takes a partial assignment
$\rho$, which fixes some variables to $0$ or $1$. It returns
a selected set $A$ with $|A|\le k$ such that $\tau_A$ agrees
with $\rho$ on every fixed variable.
The scores $M$, the sums $I$, and the directed graph are
computed once. At each call with $|K|<k$, the algorithm
computes $m_K(v)$ for every available variable
$v\in Y\setminus K$.

At a call with required set $K$, the algorithm records $K$.
If $|K|<k$ and $|T_K|=\ell$, it also records
\begin{equation}
\label{eq:arityd-rooted-recorded-prefixes}
   K\cup\{q_1,\ldots,q_j\}
   \qquad(0\le j\le k-|K|),
\end{equation}
where $Q_K=\{q_1,\ldots,q_k\}$ is constructed by
Lemma~\ref{lem:arityd-rooted-independent-set}.
When $|K|<k$, it calls the search on $K\cup\{v\}$ for every
$v\in B_K$.
The call returns a set with largest $F$-value among its recorded
sets and the sets returned by its children.

\begin{algorithm}[H]
\small
\DontPrintSemicolon
\SetAlgoNoEnd
\SetKwFunction{RootSearch}{RootedSearch}
\SetKwProg{Fn}{Procedure}{:}{}
\caption{\Rooted algorithm}
\label{alg:arityd-rooted}
\KwIn{Zero-free \Rooted instance $\mathcal I$ with budget $k$;
rational $\delta\in(0,1)$; partial assignment
$\rho:H\to\{0,1\}$ with $|\rho^{-1}(1)|\le k$.}
\KwOut{A selected set $A$ with $|A|\le k$ such that $\tau_A$ extends $\rho$.}
\If{$k=0$}{\Return $\varnothing$\;}
Set $d$ to the largest constraint arity, taking $d=1$
if $\mathcal C=\varnothing$\;
Fix an order on $X$ for enumerating branches and breaking ties\;
Set $H^1\gets\rho^{-1}(1)$, $H^0\gets\rho^{-1}(0)$, and
$Y\gets X\setminus H^0$\;
Set $\theta,d_{\max},\ell$ as in \eqref{eq:arityd-rooted-parameters}\;
Compute $M$, $I$, and the directed graph once\;
\Return \RootSearch{$H^1$}\;
\BlankLine
\Fn{\RootSearch{$K$}}{
  $\mathcal L\gets\{K\}$\;
  \If{$|K|=k$}{\Return $K$\;}
  Recompute $m_K(v)$ for all $v\in Y\setminus K$\;
  Form $T_K$, $E^{\mathrm{out}}(K)$, and $B_K$ as defined above\;
  \If{$B_K=\varnothing$}{\Return $K$\;}
  \If{$|T_K|=\ell$}{
    Form $G_K$ on $T_K\setminus E^{\mathrm{out}}(K)$\;
    Construct $Q_K=\{q_1,\ldots,q_k\}$ by
    Lemma~\ref{lem:arityd-rooted-independent-set}\;
    \For{$j=0,\ldots,k-|K|$}{
      $\mathcal L\gets\mathcal L\cup\{K\cup\{q_1,\ldots,q_j\}\}$\;
    }
  }
  \ForEach{$v\in B_K$}{
    $A_v\gets$\RootSearch{$K\cup\{v\}$}\;
    $\mathcal L\gets\mathcal L\cup\{A_v\}$\;
  }
  \Return a set in $\mathcal L$ with largest $F$-value\;
}
\end{algorithm}

\paragraph*{Replacement at a call.}
The next lemma bounds the loss from replacing a feasible
selected set $S\supseteq K$ by a set recorded at the current
call when $(S\setminus K)\cap B_K=\varnothing$.

\begin{lemma}
\label{lem:arityd-rooted-terminal}
Let $K$ be the required selected set at a search call, and let
$K\subseteq S\subseteq Y$ with $|S|\le k$.
If $(S\setminus K)\cap B_K=\varnothing$, then the call records
a set $A$ with $K\subseteq A\subseteq Y$ and $|A|\le k$ such that
\begin{equation}
\label{eq:arityd-rooted-terminal}
   F(A)\ge F(S)-\delta\OPT(\mathcal I).
\end{equation}
\end{lemma}

\begin{proof}
Put $s:=|S\setminus K|$. If $s=0$, take the recorded set
$A:=K=S$. Assume henceforth that $s\ge1$, so $|K|<k$.

Since $E^{\mathrm{out}}(K)\subseteq B_K$, the assumption gives
$I(u\to x)\le\theta M(u)$ for every $u\in K$ and
$x\in S\setminus K$. Lemma~\ref{lem:arityd-rooted-projection}
therefore implies
\begin{equation}
\label{eq:arityd-rooted-terminal-projection}
 F(K)\ge F(S)-\sum_{v\in S\setminus K}m_K(v)
       -\theta s\sum_{u\in K}M(u).
\end{equation}

Suppose first that $|T_K|<\ell$. Then $T_K$ contains every
available variable with positive candidate score.
Since $(S\setminus K)\cap T_K=\varnothing$, we have
$m_K(v)=0$ for every $v\in S\setminus K$.
Taking the recorded set $A:=K$ and applying
\eqref{eq:arityd-rooted-random-isolation}, we obtain
\[
 F(A)\ge F(S)-\theta s c_d\OPT(\mathcal I)
       \ge F(S)-\frac{\delta}{2}\OPT(\mathcal I),
\]
where the last inequality uses $s\le k$ and
$\theta kc_d=\delta/2$. Thus $A$ satisfies
\eqref{eq:arityd-rooted-terminal}.

Now suppose that $|T_K|=\ell$. Let
$Q_K=\{q_1,\ldots,q_k\}$ be the replacement set constructed
at this call, and put
\[
   P:=\{q_1,\ldots,q_s\},\qquad A:=K\cup P.
\]
Since $s\le k-|K|$, the algorithm records $A$.
Moreover, $P\subseteq Y\setminus K$, so
$K\subseteq A\subseteq Y$ and $|A|=|S|\le k$.

The sets $P$ and $S\setminus K$ both have $s$ variables.
Every variable in $P$ belongs to $T_K$, whereas
$S\setminus K$ is disjoint from $T_K$.
By the choice of $T_K$, every candidate score in $P$ is
at least every candidate score in $S\setminus K$. Hence
\begin{equation}
\label{eq:arityd-rooted-top-domination}
   \sum_{p\in P}m_K(p)
   \ge\sum_{v\in S\setminus K}m_K(v).
\end{equation}

Lemma~\ref{lem:arityd-rooted-independent-set} gives
$I(u\to v)\le\theta M(u)$ for every $v\in P$ and
$u\in K\cup P$ with $u\ne v$. Thus the two sums subtracted
in Lemma~\ref{lem:arityd-rooted-insertion} satisfy
\begin{equation}
\label{eq:arityd-rooted-overlap-error}
 \sum_{i=1}^s\sum_{j<i}I(q_i\to q_j)
 \le\theta\sum_{i=1}^s(i-1)M(q_i)
 \le\theta s\sum_{p\in P}M(p)
\end{equation}
and
\begin{equation}
\label{eq:arityd-rooted-global-error}
 \sum_{i=1}^s
 \sum_{u\in K\cup\{q_1,\ldots,q_{i-1}\}}I(u\to q_i)
 \le\theta s\sum_{z\in A}M(z).
\end{equation}

Using these bounds in Lemma~\ref{lem:arityd-rooted-insertion},
then applying \eqref{eq:arityd-rooted-terminal-projection}
and \eqref{eq:arityd-rooted-top-domination}, gives
\begin{align*}
 F(A)
 &\ge F(K)+\sum_{p\in P}m_K(p)
       -\theta s\left(
          \sum_{p\in P}M(p)+\sum_{z\in A}M(z)\right)\\
 &\ge F(S)-\theta s\left(
          \sum_{u\in K}M(u)+\sum_{p\in P}M(p)
          +\sum_{z\in A}M(z)\right)\\
 &=F(S)-2\theta s\sum_{z\in A}M(z)\\
 &\ge F(S)-2\theta kc_d\OPT(\mathcal I)\\
 &=F(S)-\delta\OPT(\mathcal I).
\end{align*}
The first equality uses $A=K\cup P$ and $K\cap P=\varnothing$.
The last two steps use $|A|\le k$, $s\le k$,
\eqref{eq:arityd-rooted-random-isolation}, and
$2\theta kc_d=\delta$.
\end{proof}

\paragraph*{Correctness.}
We prove feasibility and the approximation guarantee for
every search call by induction on the remaining budget
$k-|K|$.

\begin{lemma}
\label{lem:arityd-rooted-search}
Let $H^1\subseteq K\subseteq Y$ with $|K|\le k$.
The search call in Algorithm~\ref{alg:arityd-rooted}
with required selected set $K$ returns a set $A$ with
\[
   K\subseteq A\subseteq Y,\qquad |A|\le k,
\]
such that
\[
   F(A)\ge F(S)-\delta\OPT(\mathcal I)
\]
for every $S$ with $K\subseteq S\subseteq Y$ and $|S|\le k$.
\end{lemma}

\begin{proof}
We prove feasibility and the value bound by induction on
the remaining budget $k-|K|$.

If $|K|=k$, the call returns $A=K$. Any set $S$ containing
$K$ and having size at most $k$ must equal $K$.
Thus $A$ is feasible and $F(A)=F(S)$, which proves
the value bound.

Suppose that $|K|<k$, and assume that the lemma holds
for every required set with smaller remaining budget.
Each recursive call has required set $K\cup\{v\}$ for
some $v\in B_K\subseteq Y\setminus K$. Hence
\[
   H^1\subseteq K\cup\{v\}\subseteq Y,
   \qquad |K\cup\{v\}|=|K|+1\le k.
\]
Its remaining budget is $k-|K|-1$, so the induction
hypothesis applies to every recursive call.

Every set recorded directly at the current call is either
$K$ or a set $K\cup\{q_1,\ldots,q_j\}$ with
$0\le j\le k-|K|$. Whenever such a prefix is recorded,
its variables belong to the replacement set
$Q_K\subseteq Y\setminus K$. Thus every directly recorded
set contains $K$, lies in $Y$, and has size at most $k$.
For every child call, the induction hypothesis gives
\[
   K\cup\{v\}\subseteq A_v\subseteq Y,
   \qquad |A_v|\le k.
\]
Thus every directly recorded set and every set returned
by a child is feasible for the current call.
The returned set $A$ is one of these sets, so
$K\subseteq A\subseteq Y$ and $|A|\le k$.

It remains to prove the value bound.
Fix any $S$ with $K\subseteq S\subseteq Y$ and $|S|\le k$.
Suppose first that
$(S\setminus K)\cap B_K\ne\varnothing$.
Choose $v$ in this intersection.
The algorithm makes a recursive call with required set
$K\cup\{v\}$, and $S$ satisfies
\[
   K\cup\{v\}\subseteq S\subseteq Y,\qquad |S|\le k.
\]
The induction hypothesis therefore gives
\[
   F(A_v)\ge F(S)-\delta\OPT(\mathcal I).
\]
The current call returns a set of value at least that of
every set returned by a child. Consequently,
\[
   F(A)\ge F(A_v)\ge F(S)-\delta\OPT(\mathcal I).
\]

Now suppose that
$(S\setminus K)\cap B_K=\varnothing$.
Lemma~\ref{lem:arityd-rooted-terminal} gives a set $A'$
recorded directly at the current call with
\[
   F(A')\ge F(S)-\delta\OPT(\mathcal I).
\]
The returned set has value at least that of every directly
recorded set, so
\[
   F(A)\ge F(A')\ge F(S)-\delta\OPT(\mathcal I).
\]
Since $S$ was arbitrary, the value bound holds for every
permitted set $S$. This completes the induction.
\end{proof}

\paragraph*{Running time.}
\begin{theorem}
\label{thm:arityd-rooted-zero-free}
Let $\mathcal I=(X,\mathcal C,w,k)$ be an instance of arity
at most $d$ whose relations are nonempty, zero-free and
\Rooted. For every rational $\delta\in(0,1)$ and partial
assignment $\rho$ fixing at most $k$ variables to $1$,
Algorithm~\ref{alg:arityd-rooted} returns a selected set $A$
with $|A|\le k$ such that $\tau_A$ extends $\rho$ and
\[
   F(A)\ge F(S)-\delta\OPT(\mathcal I)
\]
for every selected set $S$ with $|S|\le k$ such that
$\tau_S$ extends $\rho$.

With membership tests as specified in
Remark~\ref{rem:arityd-uniform-algorithms}, the algorithm
runs in time
\[
   2^{O(k\log(1+dk/\delta))}\ell_{\mathrm{in}}^{O(1)}.
\]
\end{theorem}

\begin{proof}
If $k=0$, the assumption on $\rho$ implies that every
fixed variable is assigned $0$. The algorithm returns
the empty set, which is the unique feasible selected set
and respects $\rho$. Both the value bound and the
running-time bound follow. Assume henceforth that $k\ge1$.

The initial search call has required set
$K=H^1=\rho^{-1}(1)$, and
$Y=X\setminus\rho^{-1}(0)$. A selected set $S$ satisfies
$H^1\subseteq S\subseteq Y$ exactly when $\tau_S$ extends
$\rho$. Lemma~\ref{lem:arityd-rooted-search} therefore
gives the stated feasibility and value guarantees.
For the empty partial assignment, applying the lemma
to an optimal selected set gives
$F(A)\ge(1-\delta)\OPT(\mathcal I)$.

For the running-time analysis, it is enough to use the
largest constraint arity as $d$, taking $d=1$ if
$\mathcal C=\varnothing$. Using a larger arity bound only
enlarges the stated running-time bound.

Each call has at most
\[
   |B_K|\le\ell+kd_{\max}=k(4d_{\max}+1)
\]
children. Every child adds one variable to $K$, so the
recursion depth is at most $k-|H^1|\le k$.
The number of calls is therefore at most
\[
   \bigl(1+k(4d_{\max}+1)\bigr)^k.
\]
By \eqref{eq:arityd-cd}, we have $c_d\le ed$. Hence
\[
   d_{\max}
   \le 1+\frac{d-1}{\theta}
   =1+\frac{2kc_d(d-1)}{\delta},
\]
and thus $d_{\max}+1=O(d^2k/\delta)$.
Since $\delta<1$,
\[
   1+\frac{d^2k^2}{\delta}
   \le \left(1+\frac{dk}{\delta}\right)^2.
\]
It follows that the number of calls is
$2^{O(k\log(1+dk/\delta))}$.

We next bound the work done before the search and at
each call. For every constraint $C$, we compute $U_C$
by testing $\tau_{\{x\}}\models C$ for each
$x\in\scope(C)$, using at most $d$ membership tests.
For each $v\in U_C$, the weight $w(C)$ contributes to
$M(v)$ and to $I(v\to x)$ for every
$x\in\scope(C)\setminus\{v\}$. Thus all scores $M$ and
sums $I$ can be computed using at most $d^2$ additions
per constraint. The directed graph is then computed
once from these quantities.

At a call with required set $K$, a constraint can
contribute to a candidate score only if
$K\cap U_C=\varnothing$. For each available
$v\in U_C$, we test whether $\tau_{K\cup\{v\}}$ satisfies
$C$ and, if so, add $w(C)$ to $m_K(v)$.
This uses at most $d$ membership tests per constraint.
Sorting the scores, forming the candidate sets, and
constructing $Q_K$ by
Lemma~\ref{lem:arityd-rooted-independent-set}
take polynomial time in the input size.

A replacement set has $k$ distinct variables, so whenever
it is constructed, $k\le |X|$. Consequently, each call
records at most $|X|+1$ sets directly and receives at
most $|X|$ sets from its children. Each set can be
evaluated by scanning the original constraints.
All value evaluations at a call therefore take
polynomial time in the input size.

Finally, all scores and interaction sums are sums of
input weights, so their numerators and denominators
have polynomial bit length in $\ell_{\mathrm{in}}$.
For $d\ge2$, the numerator and denominator in
\[
   c_d=\frac{d^d}{(d-1)^{d-1}}
\]
have $O(d\log(d+1))$ bits; for $d=1$, we have $c_1=1$.
Since constraint scopes are explicitly listed, the
maximum arity used by the algorithm is at most
$\ell_{\mathrm{in}}$. Thus $c_d$, $\theta$,
$d_{\max}$ and $\ell$ also have polynomial bit length.
Exact arithmetic and comparisons therefore take
polynomial time in $\ell_{\mathrm{in}}$.

For fixed $\Gamma$, membership tests take constant time.
When relation descriptions are supplied as part of the input,
we assume that membership tests take time polynomial in
$\ell_{\mathrm{in}}$, with a fixed exponent across inputs.
Preprocessing and the work at each call consequently
take $\ell_{\mathrm{in}}^{O(1)}$ bit operations.
Combining this with the bound on the number of calls
proves the stated running time.
\end{proof}

\paragraph*{Adding zero-valid constraints.}
\begin{theorem}
\label{thm:arityd-rooted-epas}
Let $\Gamma$ be a fixed Boolean constraint language of arity
at most $d$. If every relation in $\Gamma_+$ is \Rooted,
then \maxwtcsp{}$(\Gamma)$ admits an \epas{} parameterized by
$k$, with running time
\[
   2^{O(k\log(1+dk/\varepsilon))}\ell_{\mathrm{in}}^{O(1)}.
\]
\end{theorem}

\begin{proof}
Fix a rational $\varepsilon\in(0,1)$.
We apply Lemma~\ref{lem:arityd-zero-valid-lifting}.
It requires an algorithm $\mathcal A$ with the following
guarantee on instances over $\Gamma_+$. Given such an
instance $\mathcal J$ with budget $k$, a partial assignment
$\rho$ fixing at most $k$ variables to $1$, and a rational
$\delta\in(0,1)$, the algorithm must return a feasible
assignment $\tau$ extending $\rho$ such that
\[
 F^{\mathcal J}(\tau)
 \ge F^{\mathcal J}(\tau')-\delta\OPT(\mathcal J)
\]
for every feasible assignment $\tau'$ extending $\rho$.
Here $\OPT(\mathcal J)$ is measured before imposing $\rho$.

Define $\mathcal A(\mathcal J,\rho,\delta)$ to run
Algorithm~\ref{alg:arityd-rooted} with these inputs and
return the assignment $\tau_A$ corresponding to its
selected set $A$.
Every relation in $\Gamma_+$ is nonempty, zero-free and
\Rooted. Thus Theorem~\ref{thm:arityd-rooted-zero-free}
applies. It shows that $\tau_A$ is feasible, extends
$\rho$, and satisfies the required value bound.
The same theorem gives running time
\[
   2^{O(k\log(1+dk/\delta))}
   \ell_{\mathrm{in}}^{O(1)}.
\]
Thus $\mathcal A$ meets the assumptions of
Lemma~\ref{lem:arityd-zero-valid-lifting}.

Applying that lemma with this algorithm gives a
$(1-\varepsilon)$-approximate solution for every instance
over $\Gamma$. The reduction calls $\mathcal A$ with
$\delta=\varepsilon/2$, so its total running time is
\begin{align*}
 &2^{O(k\log(1+dk/\varepsilon))}
  \cdot 2^{O(k\log(1+2dk/\varepsilon))}
  \ell_{\mathrm{in}}^{O(1)}\\
 &\qquad=
  2^{O(k\log(1+dk/\varepsilon))}
  \ell_{\mathrm{in}}^{O(1)}.
\end{align*}
Here we use
$1+2dk/\varepsilon\le2(1+dk/\varepsilon)$.
Since $\Gamma$ is fixed, $d$ is constant, so this
is an \epas{} parameterized by $k$.
\end{proof}

\section{Algorithm for \Direct relations}
\label{arityd-direct-sec-main}

In this section, we provide an \epas{} for languages whose nonempty
zero-free relations are \Direct. We formalize the component construction
and item search described in Section~\ref{sec:overview}.

\subsection{Separating small components}
\label{arityd-direct-sec-color-coding}

\paragraph*{Boolean instance.}
For the zero-free analysis, fix a Boolean instance
$\mathcal I=(X,\mathcal C,w,k)$ of arity at most $d$ whose
relations are nonempty, zero-free and \Direct. Write
\[
 F(S):=F^{\mathcal I}(\tau_S)
 \qquad(S\subseteq X),
 \qquad
 M(x):=F(\{x\})
 \qquad(x\in X).
\]
We discard constraints of weight zero and assume $k\ge1$.
For $k=0$, the empty selected set is the only feasible solution.
For each constraint $C\in\mathcal C$, define
\[
 U_C:=\{x\in\scope(C):\tau_{\{x\}}\models C\},
 \qquad
 N_C:=\scope(C)\setminus U_C.
\]
As in the \Rooted analysis, $U_C$ consists of the unit-witness
variables of $C$, and $N_C$ consists of its controls.

\paragraph*{Shared-witness graph.}
For distinct variables $x,y\in X$, define
\[
 J(x,y):=
 \sum_{\substack{C\in\mathcal C:\ \{x,y\}\subseteq U_C}}w(C).
\]
Thus $J(x,y)=J(y,x)$ is the total weight of constraints for
which both $x$ and $y$ are unit witnesses. These weights and
the scores $M$ are computed once from the original instance.

\begin{lemma}
\label{arityd-direct-lem-pair-degree}
Every $x\in X$ satisfies
\[
 \sum_{y\in X\setminus\{x\}}J(x,y)\le(d-1)M(x).
\]
\end{lemma}

\begin{proof}
A constraint contributes to $M(x)$ exactly when $x\in U_C$.
For such a constraint, its weight appears in $J(x,y)$ for
each $y\in U_C\setminus\{x\}$. There are at most $d-1$
such variables. Therefore
\[
 \sum_{y\ne x}J(x,y)
 =\sum_{\substack{C\in\mathcal C:\ x\in U_C}}
     (|U_C|-1)w(C)
 \le(d-1)\sum_{\substack{C\in\mathcal C:\ x\in U_C}}w(C)
 =(d-1)M(x).
\]
\end{proof}

Fix a rational threshold $\theta\in(0,1)$.
The graph $G_\theta$ has vertex set $X$. Distinct vertices
$x,y$ are adjacent exactly when
\[
 J(x,y)>0
 \qquad\text{and}\qquad
 J(x,y)\ge\theta\max\{M(x),M(y)\}.
\]
This is the graph used in the \Direct part of the technical overview.

\begin{corollary}
\label{arityd-direct-cor-heavy-degree}
The maximum degree of $G_\theta$ is at most
$d_{\max}:=\lceil(d-1)/\theta\rceil$.
\end{corollary}

\begin{proof}
If $M(x)=0$, Lemma~\ref{arityd-direct-lem-pair-degree} and
nonnegativity imply that $J(x,y)=0$ for every $y\ne x$.
Thus $x$ is isolated.
Otherwise, every neighbor $y$ contributes at least
$\theta M(x)$ to $\sum_{y\ne x}J(x,y)$. Hence
\[
 \deg_{G_\theta}(x)\theta M(x)
 \le\sum_{y\ne x}J(x,y)
 \le(d-1)M(x).
\]
Dividing by $\theta M(x)>0$ proves the degree bound.
\end{proof}

\paragraph*{Coloring the variables.}
Let $\rho$ be a partial assignment fixing at most $k$ variables
to $1$, and put
\[
 H^1:=\rho^{-1}(1),\qquad H^0:=\rho^{-1}(0),
 \qquad Y:=X\setminus H^0.
\]
A selected set $S$ respects $\rho$ exactly when
$H^1\subseteq S\subseteq Y$.
For a graph $G$ and a set $S$ of its vertices, write $N_G(S)$
for the vertices adjacent to at least one member of $S$.
A coloring \emph{isolates} $S$ if every vertex of $S$ is green
and every vertex of $N_G(S)\setminus S$ is red.
In such a coloring, every green connected component meeting $S$
is contained in $S$.

Algorithm~\ref{alg:arityd-direct-color-coding} samples colorings
of $G_\theta$. For each coloring, it retains the small green
components that avoid the variables fixed to $0$, and forces
components meeting $H^1$ to be selected in full.
It records the retained components $\mathcal K$, the variables
$P$ forced to be selected, and the remaining budget $b$.
The procedure chooses no further components at this stage.

\begin{algorithm}[H]
\small
\DontPrintSemicolon
\SetAlgoNoEnd
\caption{Random color coding}
\label{alg:arityd-direct-color-coding}
\KwIn{Graph $G_\theta$ on $X$ of maximum degree at most
$d_{\max}$; budget $k\ge1$; partial assignment $\rho$
fixing at most $k$ variables to $1$.}
\KwOut{A list $\mathcal D$ of triples $(\mathcal K,P,b)$.}
Set $H^1\gets\rho^{-1}(1)$, $H^0\gets\rho^{-1}(0)$,
and $Y\gets X\setminus H^0$\;
Set $N\gets[3(d_{\max}+1)]^k$ and
$\mathcal D\gets\varnothing$\;
\For{$j=1,\ldots,N$}{
 Independently color each vertex green with probability
 $1/(d_{\max}+1)$ and red otherwise\;
 Let $\mathcal K$ contain the green connected components
 of size at most $k$ that lie in $Y$\;
 \If{$H^1\not\subseteq\bigcup\mathcal K$}{Skip this coloring\;}
 $P\gets\bigcup\{K\in\mathcal K:K\cap H^1\ne\varnothing\}$\;
 \If{$|P|>k$}{Skip this coloring\;}
 Append $(\mathcal K,P,k-|P|)$ to $\mathcal D$\;
}
\Return $\mathcal D$\;
\end{algorithm}

\begin{lemma}
\label{arityd-direct-lem-random-coloring}
Fix $S$ with $H^1\subseteq S\subseteq Y$ and $|S|\le k$.
A single sampled coloring isolates $S$ with probability at least
$[e(d_{\max}+1)]^{-k}$.
With probability at least $1-e^{-1}$,
Algorithm~\ref{alg:arityd-direct-color-coding} records a triple
$(\mathcal K,P,b)$ such that $P\subseteq S$ and
$S\setminus P$ is a union of components in $\mathcal K$
of total size at most $b$.
\end{lemma}

\begin{proof}
If $d_{\max}=0$, every vertex is green and $S$ has no
external neighbors. Thus every sampled coloring isolates $S$.

For $d_{\max}\ge1$, put
\[
 s:=|S|,
 \qquad h:=|N_{G_\theta}(S)\setminus S|.
\]
The degree bound gives $h\le d_{\max}s$.
The coloring isolates $S$ exactly when all its $s$ vertices
are green and all its $h$ external neighbors are red.
Since vertices are colored independently,
\begin{align*}
 \Pr[\text{$S$ is isolated}]
 &=\left(\frac1{d_{\max}+1}\right)^s
   \left(\frac{d_{\max}}{d_{\max}+1}\right)^h\\
 &\ge\left[
   \frac1{d_{\max}+1}
   \left(\frac{d_{\max}}{d_{\max}+1}\right)^{d_{\max}}
   \right]^s\\
 &\ge[e(d_{\max}+1)]^{-s}\\
 &\ge[e(d_{\max}+1)]^{-k}.
\end{align*}
Here the first inequality uses $h\le d_{\max}s$, and
\[
 \left(\frac{d_{\max}}{d_{\max}+1}\right)^{d_{\max}}
 =\left(1+\frac1{d_{\max}}\right)^{-d_{\max}}
 \ge e^{-1}.
\]
The final inequality uses $s\le k$.

The algorithm samples $N=[3(d_{\max}+1)]^k$ colorings.
Since $e<3$, each trial isolates $S$ with probability at
least $1/N$. This also holds when $d_{\max}=0$.
Independence of the trials gives
\[
 \Pr[\text{no trial isolates $S$}]
 \le(1-1/N)^N
 \le e^{-1}.
\]

Consider a trial that isolates $S$.
A green connected component meeting $S$ cannot contain a
vertex outside $S$: a green path leaving $S$ would have
to pass through an external neighbor of $S$, which is red.
Thus every such component lies in $S$, has size at most
$k$, and lies in $Y$. The algorithm retains all these
components, and they partition $S$.

Since $H^1\subseteq S$, the retained components cover $H^1$.
Every component meeting $H^1$ lies in $S$, so their union
$P$ satisfies $P\subseteq S$ and $|P|\le k$.
Neither condition for skipping the coloring holds, and
the algorithm records $(\mathcal K,P,b)$ with $b=k-|P|$.

The retained components contained in $S$ and disjoint
from $P$ partition $S\setminus P$. Their total size is
\[
 |S\setminus P|=|S|-|P|\le k-|P|=b.
\]
This proves the recorded triple's guarantee.
\end{proof}

\paragraph*{Derandomization.}
We replace the random trials by a deterministic family of
green and red colorings.

\begin{lemma}
\label{arityd-direct-lem-universal-isolation}
Let $k\ge1$, and let $G$ have $n$ vertices and maximum degree
at most $d_{\max}$. Put $q:=k(d_{\max}+1)$.
One can deterministically construct a family of at most
\[
 (q+1)^{O(k)}\log(n+2)
\]
green and red colorings such that every set $S$ of size at
most $k$ is isolated by some coloring in the family.
The construction takes
$(q+1)^{O(k)}\ell_{\mathrm{in}}^{O(1)}$ time.
In a coloring isolating $S$, every green connected component
meeting $S$ is contained in $S$ and has size at most $k$.
\end{lemma}

\begin{proof}
The $(n,q,q^2)$-splitter construction of Naor, Schulman, and
Srinivasan~\cite[Lemma~2 and Theorem~3(iv)]{DBLP:conf/focs/NaorSS95}
lets us
deterministically construct a family $\mathcal H$ of maps
$V(G)\to[q^2]$ such that every set of at most $q$ vertices
receives distinct labels under some map. The family has size
$(q+1)^{O(1)}\log(n+2)$, and its construction takes time
polynomial in $q+n$. For $q=1$, the constant map suffices;
if $q>n$, one map assigning distinct labels to all vertices suffices.

For each $h\in\mathcal H$ and each $B\subseteq[q^2]$ with
$|B|\le k$, color $x$ green if $h(x)\in B$, and red otherwise.
Fix $S$ with $|S|\le k$. The set
\[
 T:=S\cup(N_G(S)\setminus S)
\]
has size at most $k(d_{\max}+1)=q$.
Choose a map assigning distinct labels to $T$ and take
$B=h(S)$. This makes every vertex of $S$ green and every
external neighbor red, so the coloring isolates $S$.

The number of colorings is at most
\[
 |\mathcal H|\sum_{i=0}^{k}\binom{q^2}{i}
 \le |\mathcal H|(q^2+1)^k
 =(q+1)^{O(k)}\log(n+2).
\]
Enumerating and evaluating these colorings gives the stated
construction time. Every green connected component meeting
$S$ lies in $S$, since all external neighbors are red.
Its size is therefore at most $k$.
\end{proof}

Enumerating this family in
Algorithm~\ref{alg:arityd-direct-color-coding} produces a
deterministic list of triples. For every feasible $S$ with
$\tau_S$ extending $\rho$, this list contains a triple
$(\mathcal K,P,b)$ such that $P\subseteq S$ and
$S\setminus P$ is a union of components in $\mathcal K$
with total size at most $b$.

\subsection{Search on item constraints}
\label{arityd-direct-sec-item-hus}
\label{arityd-direct-sec-contraction}

Fix a triple $(\mathcal K,P,b_0)$ recorded by
Algorithm~\ref{alg:arityd-direct-color-coding}.
The items are the components
$\mathcal K_{\mathrm{free}}:=\{K\in\mathcal K:K\cap P=\varnothing\}$,
with initial costs $\kappa(y):=|y|$ and budget $b_0=k-|P|$.
Fix an order on the items.
Write $\mathcal J=(\mathcal K_{\mathrm{free}},
\mathcal C_{\mathrm{item}},w,\kappa,b)$ for the current item instance,
where $w$ gives nonnegative rational weights, $\kappa$ gives
positive integer costs, and $b\in\mathbb Z_{\ge0}$ is the
remaining budget.
An \emph{item constraint} is a Boolean constraint whose variables
are items.
Constraints have arity at most $d\ge1$ and may have an empty scope;
their relations are given by their satisfying tuples.
We use the constraint and objective notation from
Section~\ref{sec:preliminaries}, with items in place of variables.
For a complete or partial item assignment $\sigma$, put
\[
 \supp(\sigma):=\sigma^{-1}(1),
 \qquad
 \kappa(\sigma):=\sum_{y\in\supp(\sigma)}\kappa(y).
\]
For a complete assignment $\sigma$, $F^{\mathcal J}(\sigma)$
is the total weight of its satisfied constraints.
The assignment is feasible when $\kappa(\sigma)\le b$,
and $\OPT(\mathcal J)$ is the largest value of a feasible assignment.
A complete assignment $\sigma$ to the initial instance
selects the variable set
\[
 \begin{aligned}
 S(\sigma)&:=P\cup\bigcup_{K\in\supp(\sigma)}K,\\
 |S(\sigma)|&=|P|+\kappa(\sigma).
 \end{aligned}
\]
For each original constraint $C$, define an item constraint
$\widehat C$ of weight $w(C)$ on the items intersecting
$\scope(C)$, listed in the fixed item order.
Define its relation by
\begin{equation}
\label{arityd-direct-eq-mask}
\begin{aligned}
 \sigma\models\widehat C
 \quad\Longleftrightarrow\quad&
 \tau_{S(\sigma)}\models C
 \ \text{and}\\
 &\varnothing\ne S(\sigma)\cap U_C\subseteq K
 \quad\text{for some }K\in\mathcal K.
\end{aligned}
\end{equation}
Thus $C$ is counted when its selected unit witnesses,
including those in $P$, lie in one component.
Since the components are disjoint, each item constraint
has arity at most $d$.
Call this item instance $\mathcal J_{\rm col}$.
The search requires every nonempty zero-free relation obtained
by fixing items in an input constraint to be \Direct.
Lemma~\ref{arityd-direct-lem-owner-hus} proves that the
constructed instance has this property.

\paragraph*{Recursive calls.}
Discard constraints of weight zero.
Let $\delta\in(0,1)$ be the rational error parameter for the
item search. We seek a feasible assignment of value at least
$(1-\delta)\OPT(\mathcal J)$ for the input item instance
$\mathcal J$. Set
\[
 \eta_1:=\frac{\delta}{2b_0+1},
 \qquad
 \eta_2:=\frac{\eta_1}{2}.
\]
These thresholds remain fixed throughout the search.
At each call, $\mathcal J$ contains the unfixed items and
the constraints after substituting all earlier fixed values.
The assignment $\mathbf0$ sets every unfixed item to $0$.
A constraint is zero-free if this assignment does not satisfy it.
If $b=0$ or no constraints remain, return $\mathbf0$.
Otherwise, define
\[
 \mathcal Z:=\{C\in\mathcal C_{\mathrm{item}}:\mathbf0\models C\},
 \qquad
 W:=\sum_{C\in\mathcal Z}w(C)=F^{\mathcal J}(\mathbf0).
\]
For a partial assignment $\pi$ with $\kappa(\pi)\le b$,
form the child instance $\mathcal J_\pi$ by substituting
$\pi$, removing its fixed items, and reducing the budget
to $b-\kappa(\pi)$.
First remove constraints that become always satisfied or
always unsatisfied.
Then, among the remaining restrictions of constraints in
the parent's $\mathcal Z$, remove those satisfied when all
unfixed items are $0$; the values prescribed by $\pi$ stay fixed.
The remaining constraints and items keep their weights and costs.
\paragraph*{The branching rule when $W>0$.}
We fix every item whose total weight of constraints
in $\mathcal Z$ reaches $\eta_1W/b$.
Define
\[
 \begin{aligned}
 \deg_0(y)&:=
 \sum_{\substack{C\in\mathcal Z\\y\in\scope(C)}}w(C),\\
 \mathcal K_{\mathrm{hi}}&:=
 \left\{y\in\mathcal K_{\mathrm{free}}:
 \deg_0(y)\ge\frac{\eta_1W}{b}\right\}.
 \end{aligned}
\]
Create one branch for every assignment
$\pi:\mathcal K_{\mathrm{hi}}\to\{0,1\}$ with
$\kappa(\pi)\le b$, including the all-zero assignment.
We call the items in $\mathcal K_{\mathrm{hi}}$ marked.
Lemma~\ref{arityd-direct-lem-item-credit} shows that the
additionally removed constraints left unsatisfied by a feasible
completion have total weight at most $\eta_1\OPT(\mathcal J)$.
It also shows that the all-zero branch gives a zero-free child.

\paragraph*{The branching rule when $W=0$.}
Since $W=0$, the current instance is zero-free.
For a set $Q$ of unfixed items, let $\tau_Q$ select exactly
the items in $Q$.
We use two quantities to form the candidate set.
For each item $y$, $\mu(y)$ is the value of selecting $y$ alone.
For distinct items $y,z$, $I_{\mathrm{item}}(y\to z)$ is
the part of that value from constraints whose scopes contain $z$:
\[
 \begin{aligned}
 \mu(y)&:=F^{\mathcal J}(\tau_{\{y\}}),\\
 I_{\mathrm{item}}(y\to z)&:=
 \sum_{\substack{C\in\mathcal C_{\mathrm{item}}:\\
                 y,z\in\scope(C),\ \tau_{\{y\}}\models C}}w(C)
 \qquad(y\ne z).
 \end{aligned}
\]
For rational $\eta\in(0,1)$, put
\[
 \mathsf{Bad}(y):=
 \left\{z\in\mathcal K_{\mathrm{free}}\setminus\{y\}:
 I_{\mathrm{item}}(y\to z)\ge\frac{\eta}{b}\mu(y)\right\}
 \quad\text{if }\mu(y)>0,
\]
and set $\mathsf{Bad}(y)=\varnothing$ when $\mu(y)=0$.
Set
\[
 \ell:=\left\lfloor\frac{(d-1)b^2}{\eta}\right\rfloor+1.
\]
For each $c\in[b]$, let $L_c$ contain the $\ell$ items
of largest positive score $\mu(y)$ among those with
$\kappa(y)\le c$, or all such items if fewer exist,
using the fixed item order for ties.
Define
\[
 \mathcal B:=
 \bigcup_{c=1}^b
 \left(L_c\cup\bigcup_{y\in L_c}\mathsf{Bad}(y)\right).
\]
The search uses $\eta=\eta_2$.
Create one branch fixing $y=1$ for each
$y\in\mathcal B$ with $\kappa(y)\le b$.
Also keep $\mathbf0$ as a candidate, without a recursive call.
Lemma~\ref{arityd-direct-lem-item-swap} shows that every
positive-valued feasible assignment avoiding $\mathcal B$
can exchange one selected item for an item in a list,
losing at most $2\eta\OPT(\mathcal J)$ in value.

For each partial assignment $\pi$ specified by either
branching rule, construct $\mathcal J_\pi$ as above and
apply the search to it.
Combine its returned assignment $\sigma$ with $\pi$
to obtain the candidate $\pi\sqcup\sigma$.
After processing all branches, evaluate every candidate,
including $\mathbf0$ when $W=0$, on all constraints of
the current instance $\mathcal J$, including those removed
before recursion.
Return a candidate with largest $F^{\mathcal J}$-value.
Use the fixed item order to enumerate assignments and
break ties, with $0$ before $1$.

Algorithm~\ref{alg:arityd-direct-items} gives the search.

\begin{algorithm}[!htbp]
\small
\DontPrintSemicolon
\SetAlgoNoEnd
\SetKwFunction{ItemSearch}{SearchItems}
\SetKwFunction{ItemChild}{MakeChild}
\SetKwProg{Fn}{Procedure}{:}{}
\caption{Search on item constraints}
\label{alg:arityd-direct-items}
\KwIn{An item instance $\mathcal J$ with budget $b_0$ such that
every nonempty zero-free relation obtained by fixing items
in any constraint is \Direct; a rational $\delta\in(0,1)$.}
\KwOut{An assignment of cost at most $b_0$.}
Discard constraints of weight zero\;
\If{$b_0=0$}{\Return $\mathbf0$\;}
Set $d$ to the largest constraint arity, taking $d=1$
if all scopes are empty\;
Set $\eta_1\gets\delta/(2b_0+1)$ and $\eta_2\gets\eta_1/2$\;
Keep $d$, $\eta_1$, $\eta_2$, and the item order fixed
throughout all recursive calls\;
\Return \ItemSearch{$\mathcal J$}\;
\BlankLine
\Fn{\ItemSearch{$\mathcal J$}}{
  Let $b$ be the budget of $\mathcal J$\;
  \If{$b=0$ or no constraints remain}{\Return $\mathbf0$\;}
  $\mathcal L\gets\varnothing$\;
  Compute $W\gets F^{\mathcal J}(\mathbf0)$\;
  \eIf{$W>0$}{
    Form $\mathcal Z$, $\deg_0$, and $\mathcal K_{\mathrm{hi}}$
    using $\eta_1$\;
    Let $\mathcal P$ contain every assignment $\pi$ on
    $\mathcal K_{\mathrm{hi}}$ with $\kappa(\pi)\le b$,
    including the all-zero assignment\;
  }{
    Compute $\mu(y)$ and $I_{\mathrm{item}}(y\to z)$
    from the current constraints\;
    Form the lists $L_c$ and $\mathcal B$ using $\eta=\eta_2$\;
    Let $\mathcal P$ contain every fixing $y=1$ with
    $y\in\mathcal B$ and $\kappa(y)\le b$\;
    $\mathcal L\gets\{\mathbf0\}$\;
  }
  \ForEach{$\pi\in\mathcal P$ in the fixed order}{
    $\mathcal J_\pi\gets$\ItemChild{$\mathcal J,\pi$}\;
    $\sigma\gets$\ItemSearch{$\mathcal J_\pi$}\;
    $\mathcal L\gets\mathcal L\cup\{\pi\sqcup\sigma\}$\;
  }
  \Return an assignment in $\mathcal L$ with largest
  $F^{\mathcal J}$-value, evaluated on all constraints received by this
  call, breaking ties by the fixed order\;
}
\BlankLine
\Fn{\ItemChild{$\mathcal J,\pi$}}{
  Let $b$ be the budget of $\mathcal J$\;
  Let $\mathcal Z$ contain the constraints of $\mathcal J$
  satisfied by $\mathbf0$\;
  Form $\mathcal J_\pi$ by substituting $\pi$ into every
  constraint of $\mathcal J$ and removing the fixed items\;
  Set the budget of $\mathcal J_\pi$ to $b-\kappa(\pi)$\;
  Remove constraints that are always satisfied or always unsatisfied\;
  Among the remaining restrictions of constraints in $\mathcal Z$,
  remove those satisfied by $\mathbf0$ on the remaining items\;
  \Return $\mathcal J_\pi$\;
}
\end{algorithm}

\subsection{Correctness}

\paragraph*{The item construction.}
For a complete item assignment $\sigma$, write
$F_{\rm mask}(\sigma):=F^{\mathcal J_{\rm col}}(\sigma)$.
For a union $S$ of retained components containing $P$,
let $\sigma_S$ be the item assignment that selects exactly
the items contained in $S$. Put
$F_{\rm mask}(S):=F_{\rm mask}(\sigma_S)$.

If $\sigma$ is feasible for $\mathcal J_{\rm col}$, then
\[
 |S(\sigma)|=|P|+\kappa(\sigma)\le|P|+b_0=k.
\]
Thus $\tau_{S(\sigma)}$ is feasible for the original instance.
By \eqref{arityd-direct-eq-mask}, each satisfied item constraint
corresponds to a satisfied original constraint of the same weight.
Hence
\[
 F_{\rm mask}(\sigma)\le F(S(\sigma)),
 \qquad
 \OPT(\mathcal J_{\rm col})\le\OPT(\mathcal I).
\]

\begin{lemma}
\label{arityd-direct-lem-owner-hus}
Every nonempty zero-free relation obtained by fixing items in
a constraint of $\mathcal J_{\rm col}$ is \Direct.
\end{lemma}

\begin{proof}
Fix any subset of the items of $\widehat C$ so that the
remaining relation is nonempty and zero-free.
Consider any satisfying assignment, keeping all fixed values
unchanged.
By \eqref{arityd-direct-eq-mask}, the selected unit witnesses
of $C$ lie in a unique component $K^*\in\mathcal K$.
Resetting any remaining item other than $K^*$ changes only
controls of $C$ from $1$ to $0$.
The original \Direct rule preserves satisfaction of $C$,
and its selected unit witnesses remain in $K^*$.
Thus $\widehat C$ remains satisfied.

If $K^*$ were already selected by the fixed choices,
including those in $P$, resetting all remaining items would
preserve satisfaction, contradicting zero-freeness.
Hence $K^*$ is a remaining item.
Selecting it alone among the remaining items satisfies the
restricted constraint, so it is a unit witness.
Every selected item that is not a unit witness differs from
$K^*$, and resetting it preserves satisfaction.
Thus the restricted relation is \Direct.
\end{proof}

The objective $F_{\rm mask}$ omits the weight of a satisfied
original constraint when its selected unit witnesses lie in
different green components. We bound this weight using $J$ and $M$.

\begin{lemma}
\label{arityd-direct-lem-mask-loss}
Let $S$ be feasible with $\tau_S$ extending $\rho$,
and let a coloring isolate $S$.
Then the component construction yields an item instance
$\mathcal J_{\rm col}$ for which $\sigma_S$ is feasible, and
\[
 F_{\rm mask}(S)\ge F(S)
       -\theta(k-1)\sum_{x\in S}M(x).
\]
Consequently,
\[
 F_{\rm mask}(S)\ge F(S)-\theta k c_d\OPT(\mathcal I),
\]
where $c_d$ is defined in \eqref{eq:arityd-cd}.
\end{lemma}

\begin{proof}
Every green component meeting $S$ lies in $S$.
It therefore has size at most $k$ and contains no variable
fixed to $0$ by $\rho$, so it is retained.
These components partition $S$.
The components forced to $1$ also lie in $S$,
so $P\subseteq S$ and the coloring is not skipped.
Thus $\sigma_S$ is well defined and feasible, since
\[
 \kappa(\sigma_S)=|S\setminus P|\le k-|P|=b.
\]
For $S=\varnothing$, both bounds follow from zero-freeness.
Otherwise, Lemma~\ref{lem:arityd-repair-accessibility} gives
a selected unit witness for every satisfied original constraint.

For each omitted constraint, choose two selected unit witnesses
in different components and assign its weight to their
unordered pair $\{x,y\}$.
The total weight assigned to this pair is at most $J(x,y)$.
Since $x,y$ lie in different green components, they are
not adjacent in $G_\theta$, so
$J(x,y)\le\theta\max\{M(x),M(y)\}\le\theta(M(x)+M(y))$.
Summing over all unordered pairs in $S$ gives
\begin{align*}
 F(S)-F_{\rm mask}(S)
 &\le\theta\sum_{\{x,y\}\subseteq S}(M(x)+M(y))\\
 &=\theta(|S|-1)\sum_{x\in S}M(x)\\
 &\le\theta(k-1)\sum_{x\in S}M(x).
\end{align*}
Finally, Lemma~\ref{lem:arityd-random-isolation} gives
$\sum_{x\in S}M(x)\le c_d\OPT(\mathcal I)$,
which proves the second bound.
\end{proof}

\paragraph*{The branching rule when $W>0$.}
Each marked item has $\deg_0(y)\ge\eta_1W/b$.
Each constraint in $\mathcal Z$ contributes its weight
to $\deg_0(y)$ for at most $d$ items. Together, these give
\[
 |\mathcal K_{\mathrm{hi}}|\frac{\eta_1W}{b}
 \le\sum_{y\in\mathcal K_{\mathrm{free}}}\deg_0(y)
 \le dW.
\]
Dividing by $\eta_1W/b>0$ gives
\begin{equation}
 |\mathcal K_{\mathrm{hi}}|\le\frac{db}{\eta_1}.
 \label{arityd-direct-eq-item-credit-high}
\end{equation}

For a partial assignment $\pi$ used by the search, let
$w_{\mathrm{acc}}(\pi)$ be the total weight of removed
constraints satisfied after substituting $\pi$ and setting
every unfixed item to $0$.
Removed constraints that are always satisfied after
substitution contribute their weights for every completion
of $\pi$.
Constraints removed in the second step may become unsatisfied.
For every feasible completion $\sigma$, the next lemma bounds
how much
$w_{\mathrm{acc}}(\pi)+F^{\mathcal J_\pi}(\sigma)$
can exceed $F^{\mathcal J}(\pi\sqcup\sigma)$.

\begin{lemma}
\label{arityd-direct-lem-item-credit}
Let $\mathcal J$ be an item instance with budget $b\ge1$
and $W:=F^{\mathcal J}(\mathbf0)>0$.
Fix a rational $\eta_1\in(0,1)$ and define
$\mathcal K_{\mathrm{hi}}$ as above.

Let $\pi:\mathcal K_{\mathrm{hi}}\to\{0,1\}$ be a partial
assignment with $\kappa(\pi)\le b$.
Use $\mathcal J_\pi$ and $w_{\mathrm{acc}}(\pi)$ defined above.
For every assignment $\sigma$ to the unmarked items with
$\kappa(\sigma)\le b-\kappa(\pi)$,
\begin{align*}
 F^{\mathcal J}(\pi\sqcup\sigma)
 &\le w_{\mathrm{acc}}(\pi)+F^{\mathcal J_\pi}(\sigma)\\
 &\le F^{\mathcal J}(\pi\sqcup\sigma)
       +\eta_1\OPT(\mathcal J).
\end{align*}
If $\pi$ sets every marked item to $0$, then
$\mathcal J_\pi$ is zero-free.
\end{lemma}

\begin{proof}
Fix a feasible completion $\sigma$ of $\pi$.
The retained constraints contribute exactly
$F^{\mathcal J_\pi}(\sigma)$ to the value of $\pi\sqcup\sigma$.
Removed constraints that are always satisfied contribute
their weights to both this value and $w_{\mathrm{acc}}(\pi)$.
Those that are always unsatisfied contribute $0$ to both.
Every constraint removed in the second step has its weight
included in $w_{\mathrm{acc}}(\pi)$, whether or not
$\pi\sqcup\sigma$ satisfies it.
This proves the first inequality.
The excess counted weight is exactly the total weight of
constraints removed in the second step that
$\pi\sqcup\sigma$ fails to satisfy.

Consider any such constraint.
It was satisfied with $\pi$ fixed and all unmarked items
set to $0$.
Since the values fixed by $\pi$ do not change, its remaining
scope must contain an item $y$ selected by $\sigma$.
The constraint belonged to $\mathcal Z$ before substitution,
so its weight is included in $\deg_0(y)$.
Every item selected by $\sigma$ is unmarked, and hence
$\deg_0(y)<\eta_1W/b$ for each $y\in\supp(\sigma)$.

Since item costs are positive integers and $\sigma$ is feasible,
\[
 |\supp(\sigma)|\le\kappa(\sigma)
 \le b-\kappa(\pi)\le b.
\]
The excess counted weight is therefore at most
\[
 \sum_{y\in\supp(\sigma)}\deg_0(y)
 \le|\supp(\sigma)|\frac{\eta_1W}{b}
 \le\eta_1W.
\]
The all-zero assignment is feasible for $\mathcal J$, so
$W\le\OPT(\mathcal J)$.
This proves the second inequality.

Now suppose that $\pi$ sets every marked item to $0$.
Setting every unmarked item to $0$ then sets every item of
$\mathcal J$ to $0$.
Every constraint in $\mathcal Z$ is satisfied by this
assignment and is removed in the first or second step.
Every constraint outside $\mathcal Z$ rejects this assignment,
so every retained restriction rejects $\mathbf0$ on its
remaining items.
Thus $\mathcal J_\pi$ is zero-free.
\end{proof}

Fixing further item values in a retained constraint of
$\mathcal J_\pi$ gives a relation obtained by fixing items
in its corresponding constraint in $\mathcal J$.
Therefore, if every nonempty zero-free restriction of a
constraint of $\mathcal J$ is \Direct, the same holds for
$\mathcal J_\pi$.

\paragraph*{The branching rule when $W=0$.}
For the candidate set $\mathcal B$ formed using $\eta$,
we show that every feasible assignment of positive value
either selects an item in $\mathcal B$ or can be changed
to do so with value loss at most $2\eta\OPT(\mathcal J)$.
First, we show that every satisfied constraint has a selected
unit witness. We then use this fact to bound the change in
value when one item is set to $0$ or $1$.

For a current item constraint $C$, an item $y\in\scope(C)$
is a unit witness if $\tau_{\{y\}}\models C$.
The other items in $\scope(C)$ are controls.

\begin{lemma}
\label{arityd-direct-lem-item-descent}
Let $C$ be an item constraint with a nonempty zero-free
\Direct relation.
Every complete item assignment $\sigma$ satisfying $C$
selects an item $y\in\supp(\sigma)\cap\scope(C)$ such that
$\tau_{\{y\}}\models C$.
\end{lemma}

\begin{proof}
Suppose that no selected item in $\scope(C)$ is a unit
witness of $C$.
Every selected item in $\scope(C)$ is then a control.
Starting from $\sigma$, set these items to $0$ one at a time.
The \Direct rule preserves satisfaction of $C$ at each step.
The resulting assignment is $0$ on every item in $\scope(C)$
and still satisfies $C$, contradicting zero-freeness.
Thus $\sigma$ selects a unit witness $y$ of $C$, and
$\tau_{\{y\}}\models C$ by definition.
\end{proof}

For $a\in\{0,1\}$, we write $\sigma[y\leftarrow a]$ for
the assignment obtained by setting item $y$ to $a$ and
leaving all other values unchanged.
Updates to distinct items may be written together.

For the following lemma, write $F:=F^{\mathcal J}$.

\begin{lemma}
\label{arityd-direct-lem-item-calculus}
Let $\mathcal J$ be a zero-free item instance whose nonempty
relations are \Direct.

For every item $y\in\mathcal K_{\mathrm{free}}$,
\[
 \sum_{\substack{z\in\mathcal K_{\mathrm{free}}\\z\ne y}}
 I_{\mathrm{item}}(y\to z)
 \le(d-1)\mu(y).
\]

For every complete item assignment $\sigma$, its value satisfies
\begin{equation}
 F(\sigma)\le\sum_{y\in\supp(\sigma)}\mu(y).
 \label{arityd-direct-eq-item-subadd}
\end{equation}

For every selected item $y\in\supp(\sigma)$, setting it to $0$ gives
\begin{equation}
 F(\sigma[y\leftarrow0])\ge F(\sigma)-\mu(y).
 \label{arityd-direct-eq-item-remove}
\end{equation}

For every unselected item
$y\in\mathcal K_{\mathrm{free}}\setminus\supp(\sigma)$,
setting it to $1$ gives
\begin{align}
 F(\sigma[y\leftarrow1])
 &\ge F(\sigma)+\mu(y)\notag\\
 &\quad-\sum_{z\in\supp(\sigma)}
 \bigl(I_{\mathrm{item}}(y\to z)
       +I_{\mathrm{item}}(z\to y)\bigr).
 \label{arityd-direct-eq-item-insert}
\end{align}
\end{lemma}

\begin{proof}
Every constraint counted by $\mu(y)$ contains $y$, since it is
zero-free. Its weight therefore occurs in at most $d-1$ terms
$I_{\mathrm{item}}(y\to z)$, proving the first bound.

By Lemma~\ref{arityd-direct-lem-item-descent}, every constraint
satisfied by $\sigma$ has a unit witness $y\in\supp(\sigma)$.
Its weight is therefore included in $\mu(y)$.
Summing over the satisfied constraints proves
\eqref{arityd-direct-eq-item-subadd}.

If setting a selected item $y$ to $0$ makes a satisfied
constraint $C$ unsatisfied, then $y$ must be a unit witness
of $C$. Otherwise, the \Direct rule would preserve satisfaction.
Each such constraint has its weight included in $\mu(y)$.
Thus the total weight of constraints that become unsatisfied
is at most $\mu(y)$, proving
\eqref{arityd-direct-eq-item-remove}.

For \eqref{arityd-direct-eq-item-insert}, fix an unselected item $y$.
Every constraint counted by $\mu(y)$ whose scope is disjoint
from $\supp(\sigma)$ is unsatisfied under $\sigma$ and becomes
satisfied when $y$ is set to $1$.
Every other constraint counted by $\mu(y)$ contains a selected
item $z$, so its weight is included in
$I_{\mathrm{item}}(y\to z)$.
Therefore the total weight of constraints that become satisfied
is at least
\[
 \mu(y)-\sum_{z\in\supp(\sigma)}I_{\mathrm{item}}(y\to z).
\]

Every constraint that becomes unsatisfied contains $y$.
By Lemma~\ref{arityd-direct-lem-item-descent}, it has a unit
witness $z\in\supp(\sigma)$ before the change.
Its weight is therefore included in
$I_{\mathrm{item}}(z\to y)$.
Thus the total weight of constraints that become unsatisfied
is at most
\[
 \sum_{z\in\supp(\sigma)}I_{\mathrm{item}}(z\to y).
\]

The change in value is the weight that becomes satisfied
minus the weight that becomes unsatisfied.
Combining the two bounds proves
\eqref{arityd-direct-eq-item-insert}.
\end{proof}

We use Lemma~\ref{arityd-direct-lem-item-calculus} to bound
$\mathsf{Bad}(y)$ and $\mathcal B$.
If $\mu(y)=0$, the terms $I_{\mathrm{item}}(y\to z)$ with $z\ne y$
are nonnegative and have sum at most $0$, so each is $0$.
If $\mu(y)>0$, the definition of $\mathsf{Bad}(y)$ and
the lemma give
\[
 \begin{aligned}
 |\mathsf{Bad}(y)|\frac{\eta}{b}\mu(y)
 &\le\sum_{z\in\mathcal K_{\mathrm{free}}\setminus\{y\}}
       I_{\mathrm{item}}(y\to z)\\
 &\le(d-1)\mu(y).
 \end{aligned}
\]
For $\mu(y)>0$, divide by $(\eta/b)\mu(y)$.
When $\mu(y)=0$, $\mathsf{Bad}(y)=\varnothing$ by definition.
Thus, for every item $y$,
\begin{equation}
 |\mathsf{Bad}(y)|\le\frac{(d-1)b}{\eta}.
 \label{arityd-direct-eq-item-bad-size}
\end{equation}
For distinct items $y,z$ with $z\notin\mathsf{Bad}(y)$,
the definition and the zero-score case give
\[
 I_{\mathrm{item}}(y\to z)\le\frac{\eta}{b}\mu(y).
\]

Each list $L_c$ has at most $\ell$ items.
The definition of $\mathcal B$ and
\eqref{arityd-direct-eq-item-bad-size} therefore give
\begin{equation}
 \begin{aligned}
 |\mathcal B|
 &\le\sum_{c=1}^b\sum_{y\in L_c}
       \bigl(1+|\mathsf{Bad}(y)|\bigr)\\
 &\le b\ell\left(1+\frac{(d-1)b}{\eta}\right).
 \end{aligned}
 \label{arityd-direct-eq-item-branch-size}
\end{equation}

\begin{lemma}
\label{arityd-direct-lem-item-swap}
Let $\mathcal J$ be a zero-free item instance with budget
$b\ge1$ whose nonempty relations are \Direct.
Let $\eta\in(0,1)$ be rational, and construct the sets
$L_c$ and $\mathcal B$ from $\mathcal J$ as above.

Let $\sigma$ be a feasible item assignment with
$F(\sigma)>0$ and $\supp(\sigma)\cap\mathcal B=\varnothing$.
There are a selected item $u\in\supp(\sigma)$ and an
unselected item $v\in L_{\kappa(u)}$ such that
\[
 \sigma':=\sigma[u\leftarrow0,\ v\leftarrow1]
\]
satisfies
\begin{align*}
 \kappa(\sigma')&\le\kappa(\sigma),\\
 F(\sigma')&\ge F(\sigma)-2\eta\OPT(\mathcal J).
\end{align*}
\end{lemma}

\begin{proof}
By \eqref{arityd-direct-eq-item-subadd} and $F(\sigma)>0$,
some selected item $u$ has $\mu(u)>0$.
Put $c:=\kappa(u)\in[b]$.
The item $u$ meets both conditions for inclusion in $L_c$,
but $u\notin L_c$ because $L_c\subseteq\mathcal B$.
Therefore $|L_c|=\ell$, and every item in $L_c$ has cost
at most $\kappa(u)$ and score at least $\mu(u)$.

Since $\sigma$ selects at most $b$ items,
\eqref{arityd-direct-eq-item-bad-size} gives
\[
 \left|\bigcup_{z\in\supp(\sigma)}\mathsf{Bad}(z)\right|
 \le\frac{(d-1)b^2}{\eta}<\ell.
\]
Choose $v\in L_c$ outside this union.
Thus $v\notin\mathsf{Bad}(z)$ for every selected item $z$.

Since $L_c\subseteq\mathcal B$, the item $v$ is unselected.
Also, $\mathsf{Bad}(v)\subseteq\mathcal B$, so no selected
item belongs to $\mathsf{Bad}(v)$.
Consequently, for every $z\in\supp(\sigma)\setminus\{u\}$,
\[
 I_{\mathrm{item}}(v\to z)\le\frac{\eta}{b}\mu(v),
 \qquad
 I_{\mathrm{item}}(z\to v)\le\frac{\eta}{b}\mu(z).
\]

Replacing $u$ by $v$ does not increase the cost:
\[
 \kappa(\sigma')
 =\kappa(\sigma)-\kappa(u)+\kappa(v)
 \le\kappa(\sigma)\le b.
\]
Thus $\sigma'$ is feasible.

The item $v$ and every item selected by $\sigma$ are
feasible when selected alone.
Their scores are therefore at most $\OPT(\mathcal J)$.
There are at most $b$ items in $\supp(\sigma)\setminus\{u\}$,
so the two bounds above give
\begin{align*}
 \sum_{z\in\supp(\sigma)\setminus\{u\}}
 I_{\mathrm{item}}(v\to z)
 &\le\eta\OPT(\mathcal J),\\
 \sum_{z\in\supp(\sigma)\setminus\{u\}}
 I_{\mathrm{item}}(z\to v)
 &\le\eta\OPT(\mathcal J).
\end{align*}

Finally, apply the removal bound to $u$ and then the
insertion bound to $v$:
\begin{align*}
 F(\sigma')
 &\ge F(\sigma)-\mu(u)+\mu(v)\\
 &\quad-\sum_{z\in\supp(\sigma)\setminus\{u\}}
 \bigl(I_{\mathrm{item}}(v\to z)
       +I_{\mathrm{item}}(z\to v)\bigr)\\
 &\ge F(\sigma)-2\eta\OPT(\mathcal J).
\end{align*}
The last inequality uses $\mu(v)\ge\mu(u)$ and the two
sum bounds above.
\end{proof}

For the following theorem, $\ell_{\mathrm{in}}$ is the total bit
length of the item descriptions, constraints, weights, costs, budget,
and accuracy parameter $\delta$.

\begin{theorem}
\label{arityd-direct-thm-item-hus}
Let $\mathcal J$ be an item instance with a nonnegative
integer budget $b_0$ and constraints of arity at most $d\ge1$.
Assume that every nonempty zero-free relation obtained by
fixing items in any input constraint is \Direct.
For every rational $\delta\in(0,1)$,
Algorithm~\ref{alg:arityd-direct-items} returns a feasible
assignment of value at least $(1-\delta)\OPT(\mathcal J)$ in time
\[
 2^{O(b_0\log(1+db_0/\delta)+d)}
 \ell_{\mathrm{in}}^{O(1)}.
\]
\end{theorem}

\begin{proof}
If $b_0=0$, positive item costs make $\mathbf0$ the only
feasible assignment, and the algorithm returns it in
polynomial time. Assume $b_0\ge1$ below.
The relation assumption holds at every call: any further
restriction of a child constraint is also a restriction of
its corresponding input constraint.

\paragraph*{Termination.}
For a call on $\mathcal J$ with remaining budget $b$, define
\[
 \phi(\mathcal J):=
 2b+\mathbf1[F^{\mathcal J}(\mathbf0)>0],
\]
where $\mathbf1[E]$ equals $1$ if $E$ holds and $0$ otherwise.

If a branch fixes $\pi$ at cost $r:=\kappa(\pi)\ge1$,
its child has budget $b-r$.
The term $2b$ decreases by $2r$, while the indicator
can increase by at most $1$. Hence
\[
 \phi(\mathcal J_\pi)
 \le\phi(\mathcal J)-2r+1
 \le\phi(\mathcal J)-1.
\]
The only recursive branch of cost zero fixes every marked
item to $0$ at a call with $W>0$.
Its child is zero-free by
Lemma~\ref{arityd-direct-lem-item-credit}.
The budget stays the same, and the indicator decreases
from $1$ to $0$, so
\[
 \phi(\mathcal J_\pi)=\phi(\mathcal J)-1.
\]
Thus every recursive step decreases the nonnegative integer
$\phi$ by at least one.
Its initial value is at most $2b_0+1$, so the search terminates
after at most $2b_0+1$ recursive steps along any path.

\paragraph*{Feasibility and approximation.}
Recall that $\eta_1=\delta/(2b_0+1)$ and $2\eta_2=\eta_1$.
We prove by induction on $\phi(\mathcal J)$ that every call
returns a feasible assignment $\sigma_{\mathrm{out}}$ with
\begin{equation}
 F^{\mathcal J}(\sigma_{\mathrm{out}})
 \ge\bigl(1-\phi(\mathcal J)\eta_1\bigr)\OPT(\mathcal J).
 \label{arityd-direct-eq-item-induction}
\end{equation}
If $b=0$, the all-zero assignment is the only feasible
assignment. If no constraints remain, every assignment has
value $0$. Thus terminal calls return an optimal assignment.

At a nonterminal call, every tried assignment $\pi$ has
$\kappa(\pi)\le b$. The induction hypothesis gives a child
assignment $\sigma$ with $\kappa(\sigma)\le b-\kappa(\pi)$.
Hence
\[
 \kappa(\pi\sqcup\sigma)
 =\kappa(\pi)+\kappa(\sigma)\le b.
\]
The additional candidate $\mathbf0$ when $W=0$ is also feasible.

Write $\omega:=\OPT(\mathcal J)$ and put
\[
 \alpha:=1-\bigl(\phi(\mathcal J)-1\bigr)\eta_1.
\]
At a nonterminal call,
$2\le\phi(\mathcal J)\le2b_0+1$, and hence $0<\alpha\le1$.
Every child has $\phi(\mathcal J_\pi)\le\phi(\mathcal J)-1$.
By induction, its returned assignment $\sigma$ satisfies
\[
 F^{\mathcal J_\pi}(\sigma)
 \ge\alpha\OPT(\mathcal J_\pi).
\]

Suppose $W>0$.
Choose an optimal assignment of $\mathcal J$, and let $\pi$
be its restriction to $\mathcal K_{\mathrm{hi}}$.
This $\pi$ is one of the tried assignments.
Write $w_{\mathrm{acc}}:=w_{\mathrm{acc}}(\pi)$ and
$\omega':=\OPT(\mathcal J_\pi)$.
The first inequality of
Lemma~\ref{arityd-direct-lem-item-credit}, applied to the
optimal completion, gives
\[
 w_{\mathrm{acc}}+\omega'\ge\omega.
\]
For the child's returned completion $\sigma$, the second
inequality of that lemma and the induction bound give
\begin{align*}
 F^{\mathcal J}(\pi\sqcup\sigma)
 &\ge w_{\mathrm{acc}}+F^{\mathcal J_\pi}(\sigma)-\eta_1\omega\\
 &\ge w_{\mathrm{acc}}+\alpha\omega'-\eta_1\omega\\
 &\ge\alpha(w_{\mathrm{acc}}+\omega')-\eta_1\omega\\
 &\ge(\alpha-\eta_1)\omega\\
 &=\bigl(1-\phi(\mathcal J)\eta_1\bigr)\omega.
\end{align*}
Here $w_{\mathrm{acc}}\ge0$ and $\alpha\le1$ justify
factoring out $\alpha$.

Suppose $W=0$.
The current instance is zero-free, since constraints of
weight zero were discarded.
If $\omega=0$, the candidate $\mathbf0$ already has optimal value.
Assume $\omega>0$.
If an optimal assignment selects an item in $\mathcal B$,
there is a branch fixing that item to $1$.
Otherwise, Lemma~\ref{arityd-direct-lem-item-swap}, with
$\eta=\eta_2$, gives a feasible assignment selecting an item
in $\mathcal B$ with value at least
\[
 (1-2\eta_2)\omega=(1-\eta_1)\omega.
\]
Thus some tried assignment $\pi$ extends to a feasible
assignment of $\mathcal J$ with value at least
$(1-\eta_1)\omega$.

Since $\mathcal Z=\varnothing$, the child construction removes
only constraints that become always satisfied or always
unsatisfied after substituting $\pi$.
Their contribution is therefore exact for every completion.
With $w_{\mathrm{acc}}:=w_{\mathrm{acc}}(\pi)$ and
$\omega':=\OPT(\mathcal J_\pi)$, this gives
\[
 w_{\mathrm{acc}}+\omega'\ge(1-\eta_1)\omega.
\]
For the child's returned completion $\sigma$,
\begin{align*}
 F^{\mathcal J}(\pi\sqcup\sigma)
 &=w_{\mathrm{acc}}+F^{\mathcal J_\pi}(\sigma)\\
 &\ge w_{\mathrm{acc}}+\alpha\omega'\\
 &\ge\alpha(w_{\mathrm{acc}}+\omega')\\
 &\ge\alpha(1-\eta_1)\omega\\
 &\ge(\alpha-\eta_1)\omega\\
 &=\bigl(1-\phi(\mathcal J)\eta_1\bigr)\omega.
\end{align*}
These inequalities again use $w_{\mathrm{acc}}\ge0$ and
$0<\alpha\le1$.

The search evaluates all candidates on the constraints of
the current instance and returns one of largest value.
Thus \eqref{arityd-direct-eq-item-induction} holds in both cases.
At the initial call,
$\phi(\mathcal J)\eta_1\le(2b_0+1)\eta_1=\delta$,
which proves the required approximation guarantee.

\paragraph*{Running time.}
The recursion depth is at most $2b_0+1$.
Since $b\le b_0$ and the thresholds remain fixed,
\eqref{arityd-direct-eq-item-credit-high} and
\eqref{arityd-direct-eq-item-branch-size} give size bounds
depending only on $d,b_0,\delta$.
Fix one integer $q\ge2$ such that, at every nonterminal call,
\[
 1+|\mathcal K_{\mathrm{hi}}|\le q\quad(W>0),
 \qquad
 |\mathcal B|\le q\quad(W=0).
\]
These bounds allow $q=(1+db_0/\delta)^{O(1)}$.
When $W>0$, a branch of cost $r\ge1$ selects at most $r$
marked items. Writing $h:=|\mathcal K_{\mathrm{hi}}|$, the
number of these branches is at most
$\sum_{i=0}^{\min\{r,h\}}\binom hi\le(1+h)^r\le q^r$.
When $W=0$, there are at most $q$ branches, all of positive cost.

Let $N(b)$ bound the number of recursive calls, including
the current call, starting with budget at most $b$.
Then $N(0)=1$.
Branches of cost $r\ge1$ contribute at most $q^rN(b-r)$ calls.
The only branch of cost zero has a zero-free child,
which contributes at most $1+qN(b-1)$ calls.
Including the current call therefore gives, for $b\ge1$,
\[
 N(b)\le2+qN(b-1)+\sum_{r=1}^b q^rN(b-r).
\]
We prove $N(b)\le(4q)^b$ by induction on $b$.
The case $b=0$ follows from $N(0)=1$.
For $b\ge1$, substitute the induction bounds into the recurrence:
\[
 N(b)\le(4q)^b
 \left(\frac{2}{(4q)^b}+\frac14+\sum_{r=1}^b4^{-r}\right).
\]
Since $q\ge2$, the expression in parentheses is at most
$1/4+1/4+1/3=5/6<1$.
This proves the induction bound.

To list the assignments $\pi$ on $\mathcal K_{\mathrm{hi}}$
with $\kappa(\pi)\le b$, process the marked items in their
fixed order, trying $0$ before $1$.
Stop extending a partial assignment once its cost exceeds $b$.
Every partial assignment kept extends to such a $\pi$ by
setting the remaining marked items to $0$.
Each $\pi$ has $h+1$ such partial assignments,
one for each possible number of processed items.
Hence the number of partial assignments kept is at most
$(h+1)$ times the number of these assignments $\pi$.
Since $h$ is at most the number of input items, this takes
polynomial work per generated branch.

Forming children, computing scores and interactions, forming
lists, and evaluating candidates take
$b_0^{O(1)}2^d\ell_{\mathrm{in}}^{O(1)}$ time per call and branch.
Scores and interactions are sums of input weights.
Their bit lengths, as well as those of the thresholds and
list sizes, are polynomial in $\ell_{\mathrm{in}}$.
Thus exact arithmetic also takes polynomial time.

Each generated branch creates one child call.
Multiplying the bound on the number of calls by the work
per call and branch therefore gives
\[
 (4q)^{b_0}b_0^{O(1)}2^d\ell_{\mathrm{in}}^{O(1)}
 =2^{O(b_0\log(1+db_0/\delta)+d)}
   \ell_{\mathrm{in}}^{O(1)}.
\]
\end{proof}

\subsection{The \Direct algorithm}

Algorithm~\ref{alg:arityd-direct} runs the item search for each
recorded triple and returns the selected set with largest value. It uses the coloring threshold $\theta=\delta/(2kc_d)$
and accuracy $\delta/2$ for the item search.

\begin{algorithm}[htbp]
\small
\DontPrintSemicolon
\SetAlgoNoEnd
\caption{\Direct algorithm for zero-free constraints}
\label{alg:arityd-direct}
\KwIn{Boolean instance $\mathcal I=(X,\mathcal C,w,k)$ with
nonempty zero-free \Direct relations; rational $\delta\in(0,1)$;
partial assignment $\rho$ fixing at most $k$ variables to $1$.}
\KwOut{A selected set $A$ with $|A|\le k$ such that
$\tau_A$ extends $\rho$.}
\If{$k=0$}{\Return $\varnothing$\;}
Set $d$ to the largest constraint arity, taking $d=1$
if $\mathcal C=\varnothing$\;
Set $H^1\gets\rho^{-1}(1)$, $H^0\gets\rho^{-1}(0)$,
and $Y\gets X\setminus H^0$\;
Set $\theta\gets\delta/(2kc_d)$\;
Compute $M$, $J$, and $G_\theta$\;
Set $d_{\max}\gets\lceil(d-1)/\theta\rceil$\;
Construct the coloring family $\mathcal F$ using
Lemma~\ref{arityd-direct-lem-universal-isolation}\;
Set $\mathcal D\gets\varnothing$\;
\ForEach{coloring $\chi\in\mathcal F$}{
 Apply the steps of
 Algorithm~\ref{alg:arityd-direct-color-coding}
 from the construction of $\mathcal K$ through the append step,
 using $\chi$ as the current coloring\;
}
$\mathcal L\gets\varnothing$\;
\ForEach{$(\mathcal K,P,b)\in\mathcal D$}{
 Construct $\mathcal J_{\rm col}$ from $(\mathcal K,P,b)$
 as in Section~\ref{arityd-direct-sec-contraction}\;
 Run Algorithm~\ref{alg:arityd-direct-items} on
 $(\mathcal J_{\rm col},\delta/2)$, obtaining $\sigma$\;
 $S\gets P\cup\bigcup_{K\in\supp(\sigma)}K$\;
 $\mathcal L\gets\mathcal L\cup\{S\}$\;
}
Choose $A\in\mathcal L$ with largest exact $F(A)$\;
\Return $A$\;
\end{algorithm}

\begin{theorem}
\label{arityd-direct-thm-zero-free}
Let $\mathcal I=(X,\mathcal C,w,k)$ be an instance of arity
at most $d\ge1$ such that every constraint has a nonempty
zero-free \Direct relation.
For every rational $\delta\in(0,1)$ and partial assignment
$\rho$ with $|\rho^{-1}(1)|\le k$,
Algorithm~\ref{alg:arityd-direct} returns a set $A\subseteq X$
with $|A|\le k$ such that $\tau_A$ extends $\rho$ and
\[
 F(A)\ge F(S)-\delta\OPT(\mathcal I)
\]
for every $S\subseteq X$ with $|S|\le k$ such that
$\tau_S$ extends $\rho$.
The running time is
\[
 2^{O(d+k\log(1+dk/\delta))}\ell_{\mathrm{in}}^{O(1)},
\]
where $\ell_{\mathrm{in}}$ is the total bit length of
$\mathcal I$, $\rho$, and $\delta$.
If relation descriptions are supplied as part of the input,
this bound assumes that membership in each relation can be
tested in time polynomial in $\ell_{\mathrm{in}}$, with a
fixed exponent for all inputs.
\end{theorem}

\begin{proof}
If $k=0$, then $\rho$ fixes no variable to $1$.
The algorithm returns the empty set, whose assignment
extends $\rho$ and is the unique feasible assignment.
Assume $k\ge1$.

Recall that $H^1=\rho^{-1}(1)$ and
$Y=X\setminus\rho^{-1}(0)$.
Since $|H^1|\le k$, the coloring family contains a coloring
that isolates $H^1$.
Applying Lemma~\ref{arityd-direct-lem-mask-loss} to $H^1$
shows that the algorithm records at least one triple.
Thus $\mathcal D\ne\varnothing$.

For each recorded triple $(\mathcal K,P,b)$,
Lemma~\ref{arityd-direct-lem-owner-hus} verifies the relation
assumption of Theorem~\ref{arityd-direct-thm-item-hus}.
The item search therefore returns a feasible assignment
$\sigma$.
The corresponding variable set $S(\sigma)$ satisfies
\[
 H^1\subseteq S(\sigma)\subseteq Y,
 \qquad
 |S(\sigma)|=|P|+\kappa(\sigma)\le |P|+b=k.
\]
Hence $\mathcal L$ is nonempty, and every candidate has
size at most $k$ and an assignment extending $\rho$.

Fix any $S\subseteq X$ with $|S|\le k$ such that
$\tau_S$ extends $\rho$.
Lemma~\ref{arityd-direct-lem-universal-isolation} gives a
coloring that isolates $S$.
By Lemma~\ref{arityd-direct-lem-mask-loss}, the algorithm
records its triple, and $\sigma_S$ is feasible for the
corresponding item instance $\mathcal J_{\rm col}$.
Since $\theta=\delta/(2kc_d)$, the same lemma gives
\[
 F_{\rm mask}(\sigma_S)
 \ge F(S)-\frac{\delta}{2}\OPT(\mathcal I).
\]

Let $\sigma$ be the item assignment returned for this
coloring.
Theorem~\ref{arityd-direct-thm-item-hus}, applied with
accuracy $\delta/2$, gives
\[
 F_{\rm mask}(\sigma)
 \ge (1-\delta/2)\OPT(\mathcal J_{\rm col}).
\]
Since $\sigma_S$ is feasible and
$\OPT(\mathcal J_{\rm col})\le\OPT(\mathcal I)$, it follows that
\[
 F_{\rm mask}(\sigma)
 \ge F_{\rm mask}(\sigma_S)
       -\frac{\delta}{2}\OPT(\mathcal I).
\]
The algorithm chooses $A$ with largest value among all
candidates.
Also, every constraint counted by $F_{\rm mask}(\sigma)$
is satisfied by $\tau_{S(\sigma)}$.
Therefore
\begin{align*}
 F(A)
 &\ge F(S(\sigma))\\
 &\ge F_{\rm mask}(\sigma)\\
 &\ge F_{\rm mask}(\sigma_S)
       -\frac{\delta}{2}\OPT(\mathcal I)\\
 &\ge F(S)-\delta\OPT(\mathcal I).
\end{align*}

\paragraph*{Running time.}
We may take $d$ to be the largest constraint arity,
or $1$ if $\mathcal C=\varnothing$, because this choice
only reduces the stated bound.
Using $\theta=\delta/(2kc_d)$ and $c_d\le ed$, we obtain
\[
 d_{\max}+1=O(1+d^2k/\delta),
 \qquad
 q:=k(d_{\max}+1)=O(k+d^2k^2/\delta).
\]
In particular, $q+1=(1+dk/\delta)^{O(1)}$.
Lemma~\ref{arityd-direct-lem-universal-isolation} therefore
constructs and enumerates the coloring family in time
\[
 (q+1)^{O(k)}\ell_{\mathrm{in}}^{O(1)}
 =2^{O(k\log(1+dk/\delta))}\ell_{\mathrm{in}}^{O(1)}.
\]

Computing $U_C$, $M$, $J$, and $G_\theta$ takes polynomial
time.
For each original constraint, this uses at most $d$
relation membership tests and $O(d^2)$ additions.

For each coloring, construct the item instance by
enumerating the at most $2^d$ assignments to the scope
of each $\widehat C$ and testing
\eqref{arityd-direct-eq-mask}.
The construction takes
$2^{O(d)}\ell_{\mathrm{in}}^{O(1)}$ time and produces an
item instance of at most this bit length.
Its budget satisfies $b\le k$.
Theorem~\ref{arityd-direct-thm-item-hus}, with accuracy
$\delta/2$, therefore solves this instance in time
\[
 2^{O(d+k\log(1+dk/\delta))}\ell_{\mathrm{in}}^{O(1)}.
\]
Forming the corresponding variable set and evaluating
its value in the original instance take polynomial time.
Multiplying by the number of colorings gives the stated
running time.

All scores are sums of input weights and have polynomial
bit length.
For $d\ge2$, the numerator and denominator of
$c_d=d^d/(d-1)^{d-1}$ have $O(d\log(d+1))$ bits;
$c_1=1$.
Since constraint scopes are explicitly listed, the actual
arity is at most the input size.
Together with the input encodings of $k$ and $\delta$,
this gives polynomial bit length for $\theta$,
$d_{\max}$, $q$, and the coloring parameters.
Thus exact arithmetic takes polynomial time.
Membership tests take constant time for a fixed language
and polynomial time under the assumption in the theorem.
Hence the stated bound counts bit operations.
\end{proof}

\paragraph*{Adding zero-valid constraints.}
\begin{theorem}
\label{arityd-direct-thm-main}
Let $\Gamma$ be a fixed Boolean constraint language of arity
at most $d\ge1$.
If every relation in $\Gamma_+$ is \Direct, then
\maxwtcsp{}$(\Gamma)$ admits an \epas{} parameterized by $k$.
For every rational $\varepsilon\in(0,1)$, its running time is
\[
 2^{O(d+k\log(1+dk/\varepsilon))}
 \ell_{\mathrm{in}}^{O(1)},
\]
where $\ell_{\mathrm{in}}$ is the total bit length of the
instance and $\varepsilon$.
\end{theorem}

\begin{proof}
If $k=0$, the all-zero assignment is the unique feasible
assignment and is therefore optimal.
Assume $k\ge1$.

For an instance $\mathcal J$ over $\Gamma_+$, define
$\mathcal A(\mathcal J,\rho,\delta)$ by running
Algorithm~\ref{alg:arityd-direct} and returning the
assignment $\tau_A$ corresponding to its output set $A$.
By Theorem~\ref{arityd-direct-thm-zero-free}, this assignment
is feasible, extends $\rho$, and satisfies
\[
 F^{\mathcal J}(\tau_A)
 \ge F^{\mathcal J}(\tau')
       -\delta\OPT(\mathcal J)
\]
for every feasible assignment $\tau'$ extending $\rho$.
Here $\OPT(\mathcal J)$ is taken over all feasible
assignments, before imposing $\rho$.
The running time of $\mathcal A$ is
\[
 2^{O(d+k\log(1+dk/\delta))}
 \ell_{\mathrm{in}}^{O(1)}.
\]

Thus $\mathcal A$ satisfies the assumptions of
Lemma~\ref{lem:arityd-zero-valid-lifting}.
That lemma calls $\mathcal A$ with
$\delta=\varepsilon/2$ and returns a
$(1-\varepsilon)$-approximate solution to the original
instance.
Its running time is
\begin{align*}
 &2^{O(k\log(1+dk/\varepsilon))}
  2^{O(d+k\log(1+2dk/\varepsilon))}
  \ell_{\mathrm{in}}^{O(1)}\\
 &\qquad=
  2^{O(d+k\log(1+dk/\varepsilon))}
  \ell_{\mathrm{in}}^{O(1)}.
\end{align*}
Since $\Gamma$ is fixed, we can choose $d$ as a constant.
This proves that the algorithm is an \epas{} parameterized
by $k$.

The parameter values also agree with
Section~\ref{sec:overview}.
With $\delta=\varepsilon/2$, the zero-free algorithm uses
$\theta=\varepsilon/(4kc_d)$ and calls the item search with
accuracy $\varepsilon/4$.
For initial item budget $b_0$, the item search therefore uses
\[
 \eta_1=\frac{\varepsilon}{4(2b_0+1)},
 \qquad
 \eta_2=\frac{\eta_1}{2}.
\]
\end{proof}

\section{Hardness when neither repair rule holds throughout the language}
\label{arityd-dist-hard-subsec-necessity}

Fix a finite Boolean constraint language $\Gamma$
of maximum arity $d$.
Assume that $\Gamma_+$ contains a relation that is
not \Direct{} and a relation that is not \Rooted{}.
We prove the hardness half of the dichotomy using
only its nonempty zero-free relations. We use the
unit-witness and control coordinates and the rules \Direct and \Rooted
defined above. Selecting a variable means setting it to one, and the
score counts satisfied constraints. Throughout this section, $k$ is
the size of the sought clique and $k'$ is the budget of the constructed
instance.

From a \clique{} instance $(G,k)$, where
$G=(V_G,E_G)$ is a finite simple undirected graph,
we construct an unweighted multiplicity-one instance
$\mathcal I$ with budget $k'$ and integer target $t$,
both bounded by functions of $k$ and $d$, such that
\begin{equation}
 \OPT(\mathcal I)\le t,
 \qquad
 \OPT(\mathcal I)=t \iff G\text{ has a }k\text{-clique}.
 \label{arityd-dist-hard-eq-target}
\end{equation}
Recall that $e_i$ is the tuple whose only $1$ is in coordinate $i$.
A nonempty zero-free relation $R$ is \emph{accessible} if every tuple
in $R$ has a $1$ in some coordinate $i$ with $e_i\in R$.
It is \emph{inaccessible} if it contains a tuple whose $1$s all lie
in control coordinates.

We first consider a language containing an inaccessible relation.
We then consider the remaining case, where every nonempty zero-free
relation is accessible. In this second case, the hardness hypothesis
gives two relations: $P$, which fails \Direct, and $Q$, which fails
\Rooted. The relations $P$ and $Q$ may coincide.

\begin{lemma}[A gap of one rules out an \epas]
\label{arityd-dist-hard-lem-exact-cap}
Suppose a parameterized maximization problem has nonnegative integer
solution values and admits an \fpt{} reduction from \clique{}.
For each input $(G,k)$, the reduction produces an instance
$\mathcal I$ with parameter $k'=k'(k)$ and an integer target
$t=t(k)\ge1$ such that
\[
 \OPT(\mathcal I)\le t,
 \qquad
 \OPT(\mathcal I)=t
 \iff G\text{ contains a }k\text{-clique}.
\]
Then the problem admits no \epas{} unless
$\fpt{}=\wone{}$.
\end{lemma}

\begin{proof}
Suppose an \epas{} exists. Run it on $\mathcal I$ with
$\varepsilon=1/(2t)$.
If $G$ contains a $k$-clique, the returned feasible solution has
value at least $(1-\varepsilon)t=t-1/2$. Its value is an integer
and is at most $t$, so it equals $t$.
If $G$ contains no $k$-clique, then $\OPT(\mathcal I)\le t-1$,
so the returned value is less than $t$.

Comparing the returned value with $t$ therefore decides \clique{}.
The reduction runs in \fpt{} time, and both $k'$ and $\varepsilon$
depend only on $k$, so this gives an \fpt{} algorithm for \clique{}.
Since \clique{} is $\wone{}$-complete~\cite{DBLP:journals/tcs/DowneyF95},
we obtain $\fpt{}=\wone{}$.
\end{proof}

Each construction below requires $k\ge k_{\min}$ for a fixed
positive integer $k_{\min}$ depending only on $\Gamma$.
If $k<k_{\min}$, add $k_{\min}-k$ new vertices, each adjacent
to every original vertex and every other new vertex, and replace
$k$ by $k_{\min}$.
A $k$-clique in the original graph extends to a $k_{\min}$-clique
by including all new vertices. Conversely, a $k_{\min}$-clique
in the enlarged graph contains at least $k$ original vertices,
and any $k$ of these form a clique in the original graph.
Since $k_{\min}$ is fixed, this step adds only a constant number
of vertices. We then use $(G,k)$ to denote the resulting instance,
so $k\ge k_{\min}$.
This lets us assume that $k$ is larger than a fixed constant
in either reduction.

Every relation used in the constructions belongs to $\Gamma$.
A fresh variable is a new variable, distinct from every variable
introduced earlier. When pairing two copies of a constraint family,
we state explicitly which variables they share.
Within each constraint, the variables are pairwise distinct,
as required in Section~\ref{sec:preliminaries}.

\subsection{An inaccessible relation}

\begin{lemma}[Reduction using an inaccessible relation]
\label{arityd-dist-hard-lem-inaccessible}
Let $\Gamma$ be a fixed finite Boolean constraint language containing
a nonempty zero-free inaccessible relation. Then there is an \fpt{}
reduction from \clique{} to unweighted instances of
\maxwtcsp{}$(\Gamma)$ with multiplicity one.
For each input $(G,k)$, the constructed instance $\mathcal I$
satisfies \eqref{arityd-dist-hard-eq-target}.
\end{lemma}

\begin{proof}
Choose a nonempty zero-free inaccessible relation $R\in\Gamma$.
By inaccessibility, there is a tuple $\beta\in R$ whose support
consists only of control coordinates. Among the tuples $\alpha\in R$
with $\supp(\alpha)\subseteq\supp(\beta)$, choose one with the
fewest $1$s. Put $A:=\supp(\alpha)$ and $q:=|A|$.

No tuple in $R$ has support strictly contained in $A$: such a tuple
would also have support contained in $\supp(\beta)$ and would have
fewer $1$s than $\alpha$.
Since $R$ is zero-free, $A$ is nonempty.
If $A=\{i\}$, then $\alpha=e_i\in R$, contradicting the fact
that $i$ is a control coordinate. Hence $q\ge2$.

\paragraph*{Construction.}
Assume $k\ge q+3$, using the preprocessing above.
Choose two distinct coordinates $a_1,a_2\in A$
for the edge endpoints.
Introduce $q-2$ distinct shared anchor variables
for the other coordinates of $A$.
For each vertex $v\in V_G$, introduce a variable $x_v$.
For each edge $uv$, choose an order $(u,v)$ and add
one unit-weight $R$-constraint with $x_u$ in coordinate
$a_1$ and $x_v$ in coordinate $a_2$.
Put the anchors in their coordinates of $A$ and a
fresh private variable in every coordinate outside $A$.
Different edges have disjoint private variables.
All vertex, anchor, and private variables are distinct.
Since $G$ is simple, different edges give different scope sets.

Set
\begin{equation}
 k':=k+q-2,
 \qquad
 t:=\binom{k}{2}.
 \label{arityd-dist-hard-eq-inaccessible-parameters}
\end{equation}
Let $X$ be the set of all vertex, anchor, and private
variables introduced above, and let $\mathcal C$
be the collection of the constructed $R$-constraints.
The output instance is $\mathcal I=(X,\mathcal C,w,k')$,
where $w(C)=1$ for every $C\in\mathcal C$.
Its decision threshold is $t$.

Given a $k$-clique, set its vertex variables and
every anchor to $1$, and set every other variable
to $0$.
This assignment has weight $k'$.
Each clique edge receives the tuple $\alpha\in R$.
On any other edge, the tuple has support strictly
contained in $A$, so its constraint is unsatisfied.
Thus the score is exactly $t$.

\paragraph*{Upper bound.}
Consider an assignment within the budget. Suppose first
that it selects every anchor, and let $s$ and $r$ be the numbers of
selected vertex and private variables. Then $s+r\le k$. At most
$\binom{s}{2}$ constraints have both endpoints selected. Every other
satisfied constraint needs a selected private variable: without one,
the tuple assigned to that constraint would have support properly
contained in $A$. Private variables are disjoint across edge constraints,
so the score is at most
\begin{equation}
 \binom{s}{2}+r\le\binom{s}{2}+k-s.
 \label{arityd-dist-hard-eq-inaccessible-all-anchors}
\end{equation}
For $s<k$, this is less than $t$, because
\begin{equation}
 \binom{k}{2}-\left(\binom{s}{2}+k-s\right)
 =\frac{(k-s)(k+s-3)}2>0.
 \label{arityd-dist-hard-eq-inaccessible-deficit}
\end{equation}
If $s=k$, no private variable is selected, and the score is the
number of edges among these $k$ vertices. It equals $t$ exactly when
they form a clique.

If an anchor is set to $0$, a constraint with all
its private variables set to $0$ receives a tuple
whose support is strictly contained in $A$.
Such a constraint is unsatisfied.
Therefore every satisfied constraint requires a
selected private variable.
Each private variable occurs in only one constraint,
so the score is at most $k'$.
Since $k\ge q+3\ge5$,
\[
 k'=k+q-2\le2k-5<2k\le\binom{k}{2}=t.
\]
Thus every feasible score is at most $t$, with
equality exactly for a $k$-clique.

There is one constraint per edge and at most $d-q$
private variables per edge.
Thus the instance has $|E_G|$ constraints and at most
\[
 |V_G|+(q-2)+(d-q)|E_G|
\]
variables.
Since $q\le d$,
\[
 k'\le k+d-2,\qquad t\le\frac{k^2}{2}.
\]
These bounds show that the construction is an
\fpt{} reduction.
\end{proof}

\subsection{A fixed constraint family with a switch}

We now consider the remaining case, where every nonempty zero-free
relation in $\Gamma$ is accessible. We first examine a relation
that fails \Rooted to identify a requirement we can use in the
reduction. The next lemma shows that we can leave at least two
coordinates free, prescribe a nonempty set of control coordinates
as $1$, and prescribe all other coordinates as $0$.
Among tuples with these prescribed values, a tuple belongs to
the relation exactly when the entry in every free coordinate is $1$.

Here, the prescribed values specify which tuples of the relation
we examine. In the reduction, every coordinate contains a Boolean
variable, and every constraint uses a relation from $\Gamma$.
We add constraints so that reaching the target score requires
selecting the variables corresponding to coordinates prescribed
as $1$. The proof also considers assignments that select variables
corresponding to coordinates prescribed as $0$, and bounds the
resulting score while accounting for their budget cost.

\begin{lemma}[A condition equivalent to \Rooted]
\label{arityd-dist-hard-lem-pinned-and}
Let $R$ be a nonempty zero-free relation of arity $r$.
Then $R$ is \Rooted if and only if it has no pinned conjunction.
If $R$ is accessible, this is equivalent to the following condition:
for every set $T\subseteq[r]$ of control coordinates and every
$A\subseteq[r]\setminus T$, if $R$ contains a tuple with support
$A\cup T$ but no tuple with support $A'\cup T$ for any
$A'\subsetneq A$, then $|A|=1$.

In particular, if $R$ is accessible but not \Rooted{}, there are
disjoint sets $A,T\subseteq[r]$, with $|A|\ge2$
and $T$ a nonempty set of control coordinates,
such that the following holds:
when the entries in $T$ are $1$ and those outside
$A\cup T$ are $0$, a tuple belongs to $R$ if and
only if every entry in $A$ is $1$.
\end{lemma}

\begin{proof}
Assume first that $R$ is accessible.
If $R$ is \Rooted, then for any sets $A,T$ with the stated
support property, its definition gives a unit-witness coordinate
$u\in A$ and a tuple in $R$ with support $\{u\}\cup T$,
so the condition on proper subsets forces $A=\{u\}$.

Conversely, suppose the stated condition holds.
Take any $\alpha\in R$. Let $W$ be the set of unit-witness
coordinates in $\supp(\alpha)$, and let $N$ be the set of
control coordinates in $\supp(\alpha)$.
Thus $\supp(\alpha)=W\cup N$.

Choose a set $A\subseteq W$ of minimum size such that $R$
contains a tuple with support $A\cup N$. Such a set exists
because $\alpha\in R$. No proper subset of $A$ has this property,
so the assumed condition gives $A=\{u\}$ for some $u\in W$.

Let $B\subseteq N$ be arbitrary. Choose a set
$A'\subseteq\{u\}\cup(N\setminus B)$ of minimum size such that
$R$ contains a tuple with support $A'\cup B$.
Such a set exists because $R$ contains a tuple with support
$\{u\}\cup N$. Again, no proper subset of $A'$ has this property,
so the assumed condition gives $|A'|=1$.
If $u\notin A'$, every coordinate in $A'\cup B$ would be a
control coordinate, contrary to accessibility.
Hence $A'=\{u\}$.

Thus, for every $B\subseteq N$, the relation $R$ contains a tuple
with support $\{u\}\cup B$. The same unit-witness coordinate $u$
therefore satisfies the requirement of \Rooted for $\alpha$.
Since $\alpha$ was arbitrary, $R$ is \Rooted.

Finally, suppose $R$ is not \Rooted. The condition in the lemma
then fails, so there are a set $T\subseteq[r]$ of control coordinates
and a set $A\subseteq[r]\setminus T$ with $|A|\ne1$ such that
$R$ contains a tuple with support $A\cup T$, but no tuple with
support $A'\cup T$ for any $A'\subsetneq A$.

By accessibility, the tuple with support $A\cup T$ has a $1$
in a unit-witness coordinate $u$. Since $T$ contains only
control coordinates, $u\in A$. Thus $A$ is nonempty, and
$|A|\ne1$ gives $|A|\ge2$.
If $T=\varnothing$, then $e_u\in R$ and
$\supp(e_u)=\{u\}\subsetneq A$, contradicting the condition
on proper subsets. Therefore $T$ is nonempty.
Under these prescribed entries, every possible
tuple has support $A'\cup T$ for some $A'\subseteq A$.
The support condition therefore gives the stated
equivalence.
A pinned conjunction violates the condition on proper subsets.
The preceding argument therefore proves that an accessible relation
is \Rooted{} if and only if it has no pinned conjunction.

If $R$ is inaccessible, choose a satisfying tuple whose support
$A$ consists only of controls and is minimal under inclusion.
Since $R$ is zero-free and every coordinate in $A$ is a control,
$|A|\ge2$. No tuple with support properly contained in $A$
belongs to $R$. Hence $A$, together with $T=\varnothing$, is a
pinned conjunction. Such an $R$ is also not \Rooted.
\end{proof}

We now use this requirement to build a fixed family of constraints.
The following definition specifies its variables and the score
bounds needed for the reduction.

\begin{definition}[Switch gadget]
\label{arityd-dist-hard-def-activator}
A \emph{switch gadget} $H(x;Z,Y)$ is a family of unit-weight
constraints, fixed independently of $(G,k)$.
Its Boolean variables are partitioned into three disjoint groups:
\begin{itemize}
 \item $\{x\}$, where $x$ is called the \emph{switch};
 \item $Z$, whose members are called \emph{private variables};
 \item a nonempty set $Y$, whose members are called \emph{pins}.
\end{itemize}
The family has a specified set $B\subseteq Z$ with
$h:=|B|\ge3$ and contains exactly $h+1$ constraints.

For $b\in\{0,1\}$ and $S\subseteq Z$, let $f_b(S)$ be the number
of constraints satisfied by the assignment that sets
$x=b$, every variable in $Y\cup S$ to $1$, and every variable
in $Z\setminus S$ to $0$.
We require
\begin{equation}
\begin{aligned}
 f_0(S)&\le |S| &&\text{for every }S\subseteq Z,\\
 f_1(S)&\le |S| &&\text{for every }S\subseteq Z\text{ with }S\ne B,\\
 f_1(B)&=h+1.
\end{aligned}
\label{arityd-dist-hard-eq-switch-score}
\end{equation}

The constraints also satisfy the following conditions:
\begin{enumerate}
 \item Every relation used is nonempty, zero-free, and accessible.
       The switch and pins occur only in control coordinates.
 \item Every constraint contains the switch or at least one pin.
 \item Each constraint uses pairwise distinct variables, and
       no two constraints have the same underlying scope set.
\end{enumerate}

We say that $\Gamma$ implements $H$ if every constraint in $H$
uses a relation from $\Gamma$.
\end{definition}

With all pins in $Y$ set to $1$, select every variable in $B$
and leave every variable in $Z\setminus B$ at $0$.
This selects $h$ private variables. Turning the switch $x$ on
then satisfies all $h+1$ constraints. Turning it off satisfies
at most $h$ constraints.

\begin{center}
\begin{tabular}{@{}lcc@{}}
\toprule
Selected private set & Switch off ($x=0$) & Switch on ($x=1$) \\
\midrule
$S\ne B$ & At most $|S|$ points & At most $|S|$ points \\
$S=B$, where $|B|=h$ & At most $h$ points & Exactly $h+1$ points \\
\bottomrule
\end{tabular}
\end{center}

The bounds in the table assume that every pin in
$Y$ is set to $1$.
Under this assumption, all $h+1$ constraints are
satisfied with at most $h$ selected private variables
if and only if $S=B$ and $x=1$.
The switch and pins also count toward the
Hamming-weight budget.

If no private variable is selected, no constraint is satisfied,
regardless of the values of the switch and pins.

\subsection{Construction}
\label{arityd-dist-hard-sec-construction}

Choose accessible relations $P,Q\in\Gamma_+$ such that $P$ is not
\Direct and $Q$ is not \Rooted. They may be the same relation.
We first construct the fixed family $H(x;Z,Y)$ and verify
that it is a switch gadget. We then use two copies of $H$
for each edge of the input graph to construct the CSP instance.
The lemma below proves correctness of the reduction.

\paragraph*{The switch gadget.}
We first build one $P$-constraint with the switch in a control
coordinate. We then add $Q$-constraints, each using a consecutive
group of private variables in cyclic order.

\emph{1. The constraint containing the switch.}
Since $P$ is not \Direct, choose a tuple $\alpha\in P$
and a control coordinate $c$ with $\alpha_c=1$ such that
\begin{equation}
 \alpha-e_c\notin P.
 \label{arityd-dist-hard-eq-nondirect-witness}
\end{equation}
Thus changing only entry $c$ from $1$ to $0$ changes a
satisfying tuple into a tuple outside $P$.
We will place the switch in coordinate $c$.

Define
\[
 A_P:=\{j\in\supp(\alpha):e_j\in P\},
 \qquad
 T_P:=\{j\in\supp(\alpha)\setminus\{c\}:e_j\notin P\}.
\]
Thus $A_P$ consists of the unit-witness coordinates $j$
with $\alpha_j=1$, and $T_P$ consists of the control
coordinates $j\ne c$ with $\alpha_j=1$.
These sets and $\{c\}$ are pairwise disjoint, and
\[
 \supp(\alpha)=A_P\cup T_P\cup\{c\}.
\]
Since $P$ is accessible and $\alpha\in P$, there is a
coordinate $j$ with $\alpha_j=1$ and $e_j\in P$.
Therefore $A_P\ne\varnothing$.
Put $\ell:=|A_P|$ and write
$A_P=\{u_1<\cdots<u_\ell\}$.

Introduce a switch variable $x$ and distinct private variables
$z_1,\ldots,z_\ell$. Add one unit-weight $P$-constraint
using the following variables:

\begin{center}
\begin{tabular}{@{}ll@{}}
\toprule
Coordinate & Variable \\
\midrule
$u_j$, for $j\in[\ell]$ & $z_j$ \\
$c$ & $x$ \\
$j\in T_P$ & A new pin $y_j^P$ \\
Outside $\supp(\alpha)$ & A new private variable \\
\bottomrule
\end{tabular}
\end{center}

We call this the \emph{switch constraint}.

Call the private variables in the last row \emph{fillers}.
The pins and fillers are pairwise distinct and distinct
from $x,z_1,\ldots,z_\ell$.

If $z_1,\ldots,z_\ell$ and every pin $y_j^P$ are set
to $1$, and every filler in the switch constraint is set
to $0$, then its tuple of variable values is $\alpha$
when $x=1$ and $\alpha-e_c$ when $x=0$.
Under these conditions, the switch constraint is satisfied
if and only if $x=1$.

\emph{2. The cycle constraints.}
Since $Q$ is accessible and not \Rooted,
Lemma~\ref{arityd-dist-hard-lem-pinned-and} gives disjoint
coordinate sets $A_Q$ and $T_Q$, where $|A_Q|\ge2$ and $T_Q$ is a
nonempty set of control coordinates.
Put $q:=|A_Q|$ and write $A_Q=\{a_1<\cdots<a_q\}$.
The set $A_Q$ may include control coordinates.

With the entries in $T_Q$ fixed to $1$ and those outside
$A_Q\cup T_Q$ fixed to $0$, a tuple belongs to $Q$ exactly
when the entry in every coordinate in $A_Q$ is $1$.

Set
\[
 h:=\max\{q+1,\ell\}.
\]
Then $h\ge3$, $q<h$, and $\ell\le h$.
If $\ell<h$, introduce additional new private variables
$z_{\ell+1},\ldots,z_h$. Put
\[
 B:=\{z_1,\ldots,z_h\}.
\]
Thus $B$ includes all the variables $z_1,\ldots,z_\ell$
already used in the switch constraint.

Use cyclic subscripts: after $z_h$, continue with $z_1$.
Formally, define $z_{h+j}:=z_j$ for $j\in[h]$.
For each $i\in[h]$, put
\[
 E_i:=\{z_i,z_{i+1},\ldots,z_{i+q-1}\}.
\]
This is the group of $q$ consecutive variables starting at $z_i$.
Since $q<h$, each $E_i$ is a proper subset of $B$,
and the groups $E_1,\ldots,E_h$ are distinct.
For example, when $q=2$ and $h=3$, the groups are
\[
 \{z_1,z_2\},\qquad
 \{z_2,z_3\},\qquad
 \{z_3,z_1\}.
\]

For each $j\in T_Q$, introduce a pin $y_j^Q$.
These pins are pairwise distinct and distinct from all
variables already introduced. Each constraint added below uses
the same variable $y_j^Q$ in coordinate $j$.
For each $i\in[h]$, add one unit-weight $Q$-constraint:

\begin{center}
\begin{tabular}{@{}ll@{}}
\toprule
Coordinate & Variable \\
\midrule
$a_j$, for $j\in[q]$ & $z_{i+j-1}$ \\
$j\in T_Q$ & $y_j^Q$ \\
Outside $A_Q\cup T_Q$ & A new private filler \\
\bottomrule
\end{tabular}
\end{center}

Use a different new filler for each remaining coordinate
of each constraint added in this step. All these fillers are distinct
from one another and from all previously introduced variables.

We call these $h$ constraints \emph{cycle constraints}.
Each uses relation $Q$ and one group $E_i$ of
$q$ consecutive private variables in the cyclic
order of $B$.

For each $i\in[h]$, if every pin $y_j^Q$ is set to $1$
and every filler in the cycle constraint indexed by $i$
is set to $0$, then that constraint is satisfied exactly
when every variable in $E_i$ is set to $1$.

\emph{3. The complete gadget.}
Let
\[
 Y:=\{y_j^P:j\in T_P\}\cup\{y_j^Q:j\in T_Q\},
\]
and let $Z$ consist of the variables in $B$ and all
fillers introduced above.
The groups $\{x\}$, $Z$, and $Y$ are disjoint.
Each filler occurs in exactly one constraint.
Since $T_Q\ne\varnothing$, the pin set $Y$ is nonempty.

The family $H(x;Z,Y)$ contains one switch constraint
using $P$ and $h$ cycle constraints using $Q$.
If every pin in $Y$ and every variable in $B$ is set
to $1$, and every filler is set to $0$, then all $h$
cycle constraints are satisfied. The switch constraint is
satisfied exactly when $x=1$. Thus
\[
 f_0(B)=h,\qquad f_1(B)=h+1.
\]
\paragraph*{Verification of the switch gadget.}
The construction already gives $f_0(B)=h$ and $f_1(B)=h+1$.
We now prove the bounds for $S\ne B$.

First, under any assignment, every satisfied constraint
contains a private variable set to $1$ in a unit-witness
coordinate. This follows from accessibility, since the
switch and pins occur only in control coordinates.
In particular, no constraint is satisfied when every
private variable is $0$, regardless of the switch and
pin values.

Fix $b\in\{0,1\}$ and $S\subseteq Z$ with $S\ne B$,
and set every pin to $1$. Put
\[
 a:=|S\cap B|,\qquad r:=|S\setminus B|.
\]
Thus $|S|=a+r$, and $r$ counts the selected fillers.

\begin{enumerate}
 \item If $a=0$, every satisfied constraint contains a
       selected filler. Each filler occurs in only one
       constraint, so $f_b(S)\le r=|S|$.

 \item Suppose $0<a<h$.
       On the cycle, some selected variable $z_i$ has
       an unselected successor $z_{i+1}$.
       Every group $E_j\subseteq S$ starts at a selected
       variable $z_j$. However, $E_i\not\subseteq S$
       because $q\ge2$. Therefore at most $a-1$ groups
       $E_j$ are contained in $S$.

       A satisfied cycle constraint whose group is not
       contained in $S$ must have a selected filler.
       Otherwise its pins would be $1$ and its fillers
       would be $0$, so the defining property of
       $A_Q,T_Q$ would make it unsatisfied.
       Distinct cycle constraints have distinct fillers,
       so at most $r$ such constraints are satisfied.
       The switch constraint contributes at most one point.
       Hence
       \begin{equation}
        f_b(S)\le(a-1)+r+1=|S|.
        \label{arityd-dist-hard-eq-baseline-bound}
       \end{equation}

 \item If $a=h$, then $r\ge1$ because $S\ne B$.
       Since $H$ has $h+1$ constraints,
       \[
        f_b(S)\le h+1\le h+r=|S|.
       \]
\end{enumerate}

These cases, together with the scores for $S=B$,
prove \eqref{arityd-dist-hard-eq-switch-score}.

Every constraint uses distinct variables.
The cycle constraints have different scope sets because
their intersections with $B$ are the distinct groups
$E_1,\ldots,E_h$.
The switch constraint has a different scope set because
it contains $x$, which occurs in no cycle constraint.
Every cycle constraint contains a pin, and the
switch constraint contains the switch.
All relations used are nonempty, zero-free, and accessible.
Thus $H$ satisfies Definition~\ref{arityd-dist-hard-def-activator}.

\paragraph*{The CSP instance.}
Let $(G,k)$ be a \clique{} instance, where $G=(V_G,E_G)$
is a finite simple undirected graph. Assume $k\ge6$, using
the preprocessing at the start of this section.
We use the switch gadget $H(x;Z,Y)$ constructed above.

\emph{1. Vertex and edge variables.}
For each vertex $v\in V_G$, introduce a variable $x_v$.
Selecting $x_v$ represents selecting vertex $v$.

For each edge $e\in E_G$, introduce $h$ private variables
\[
 B_e:=\{z_{e,1},\ldots,z_{e,h}\}.
\]
These correspond to the variables $z_1,\ldots,z_h$ in $B$.
Use cyclic subscripts, defining
$z_{e,h+j}:=z_{e,j}$ for $j\in[h]$.

Recall that $Z\setminus B$ is the set of fillers introduced
in the switch gadget construction.
For each filler $f\in Z\setminus B$ and each edge $e\in E_G$,
introduce a variable $d_{e,f}$. Define
\[
 Z_e:=B_e\cup\{d_{e,f}:f\in Z\setminus B\}.
\]
Thus $Z_e$ consists of $B_e$ and the fillers for edge $e$.
All vertex variables and all variables in the sets $Z_e$
are pairwise distinct.

\emph{2. Pins shared across all edges.}
For each $r\in\{1,2\}$, introduce pins
\[
 y_{j,r}^P\quad(j\in T_P),
 \qquad
 y_{j,r}^Q\quad(j\in T_Q),
\]
and define
\[
 Y^{(r)}
 :=\{y_{j,r}^P:j\in T_P\}
   \cup\{y_{j,r}^Q:j\in T_Q\}.
\]
All these pins are pairwise distinct and distinct from the
vertex and edge variables.

Introduce these two pin groups once for the entire instance.
The same pins will be used across all edges.

\emph{3. Two gadget copies for each edge.}
For each edge $e=uv$, choose an order $(u,v)$ for its endpoints.
Add two groups of constraints, one with switch $x_u$
and one with switch $x_v$.
Each group contains all $h+1$ constraints of $H$, with
variables replaced according to the following table.
We call these groups the two copies of $H$ on edge $e$.
Keep the relation and coordinate order of every constraint.

\begin{center}
\begin{tabular}{@{}lcc@{}}
\toprule
Variable in $H$ & Copy for $u$ & Copy for $v$ \\
\midrule
Switch $x$ & $x_u$ & $x_v$ \\
Variable $z_j\in B$ & $z_{e,j}$ & $z_{e,j}$ \\
Filler $f$ & $d_{e,f}$ & $d_{e,f}$ \\
Pin $y_j^P$ & $y_{j,1}^P$ & $y_{j,2}^P$ \\
Pin $y_j^Q$ & $y_{j,1}^Q$ & $y_{j,2}^Q$ \\
\bottomrule
\end{tabular}
\end{center}

The two copies on edge $e$ share every variable in $Z_e$:
both the variables in $B_e$ and the fillers.
They use different vertex variables as switches and
different pin groups.

We call the $2(h+1)$ constraints in the two switch
gadgets on each edge its \emph{edge constraints}.

Different edges have disjoint private-variable sets.
They reuse the same two global pin groups.
A vertex variable $x_v$ serves as the switch in one copy
on each edge incident to $v$.

With all pins set to $1$, all variables in $B_e$ set to $1$,
and all edge fillers set to $0$, the two copies together
satisfy $2h$ constraints, plus one additional constraint
for each endpoint variable set to $1$.
In particular, selecting both endpoints gives $2h+2$ points
using only $h$ shared private variables.

\paragraph*{The budget, target, and pin constraints.}
Let
\[
 L:=|Y^{(1)}\cup Y^{(2)}|
   =2(|T_P|+|T_Q|)
\]
be the number of pins.
The sets $A_P$, $T_P$, and $\{c\}$ are pairwise disjoint,
and $|A_P|\ge1$, so $|T_P|\le d-2$.
Similarly, $A_Q$ and $T_Q$ are disjoint and $|A_Q|\ge2$,
so $|T_Q|\le d-2$.
Since $T_Q\ne\varnothing$, we obtain
\[
 2\le L\le4(d-2).
\]
Thus $L$ is a constant depending only on $\Gamma$.
Set
\begin{equation}
 k_0:=k+h\binom{k}{2},
 \qquad
 t_0:=2(h+1)\binom{k}{2}.
 \label{arityd-dist-hard-eq-base-budget-target}
\end{equation}
For a $k$-clique, selecting its $k$ vertex variables
and every variable in $B_e$ on each of its
$\binom{k}{2}$ edges uses $k_0$ variables.
With all pins selected and all edge fillers set to $0$,
the two gadget copies on each clique edge satisfy
$2(h+1)$ constraints. These edges therefore contribute
$t_0$ points.

Set the budget to $k':=L+k_0$.
This allows selecting all $L$ pins in addition to
the vertex and edge variables described above.

We next add constraints to make selecting every pin
necessary for reaching the target. Set
\[
 \varphi:=2(h+1)k'+1.
\]
For each pin $y\in Y^{(1)}\cup Y^{(2)}$, add $\varphi$
unit-weight $P$-constraints with $y$ in coordinate $u_1$.
In every other coordinate, put a new variable.
Use a different new variable for every such coordinate
of every constraint. These variables are also distinct
from all previously introduced variables.
Let $D$ be their set.

Since $u_1\in A_P$, we have $e_{u_1}\in P$.
Thus setting $y=1$ and all the new variables in its
constraints to $0$ satisfies all $\varphi$ constraints
added for $y$.
If $y=0$, satisfying any of these constraints requires
setting one of its new variables to $1$, because
$P$ is zero-free.

Define the target by
\begin{equation}
 t:=L\varphi+t_0.
 \label{arityd-dist-hard-eq-force-and-total-target}
\end{equation}
Selecting all pins and setting every variable in $D$
to $0$ satisfies all $L\varphi$ pin constraints.
The remaining $t_0$ points must come from edge constraints.
The correctness lemma proves that any assignment
selecting at most $k'$ variables and leaving a pin
at $0$ has score less than $t$.

The variable set of the output instance is
\[
 X:=\{x_v:v\in V_G\}\cup Y^{(1)}\cup Y^{(2)}
      \cup\bigcup_{e\in E_G}Z_e\cup D.
\]
Let $\mathcal C$ be the collection of all constraints
added for edges and pins, and set $w(C):=1$ for every
$C\in\mathcal C$.
The output instance is $\mathcal I=(X,\mathcal C,w,k')$.
The decision question is whether some assignment sets
at most $k'$ variables to $1$ and satisfies at least
$t$ constraints.

\subsection{Correctness of the reduction}

\begin{lemma}[Correctness of the reduction]
\label{arityd-dist-hard-lem-activator-clique}
Let $\Gamma$ be a fixed finite Boolean constraint language
of maximum arity $d$.
Suppose $P,Q\in\Gamma_+$ are accessible, $P$ is not \Direct,
and $Q$ is not \Rooted.
For every \clique{} instance $(G,k)$ with $k\ge6$,
the construction above produces an unweighted
multiplicity-one instance $\mathcal I$ of
\maxwtcsp{}$(\Gamma)$ with budget $k'$ and target $t$
such that
\[
 \OPT(\mathcal I)\le t,
 \qquad
 \OPT(\mathcal I)=t
 \iff G\text{ contains a }k\text{-clique}.
\]
The construction runs in \fpt{} time, and
$k'\le dk^2$ and $t\le12d^3k^2$.
\end{lemma}

\begin{proof}
\textbf{Forward direction.}
Suppose $G$ contains a $k$-clique.
Set all pins to $1$, select the vertex variables of
the clique, and select exactly $B_e$ on each edge
between its vertices.
Set every other variable to $0$.
The number of selected variables is
\[
 L+k+h\binom{k}{2}=k'.
\]
On each clique edge $e=uv$, both endpoint variables
are set to $1$, exactly the private variables in $B_e$
are selected, and all pins are set to $1$.
The two switch gadgets on $e$ have switches $x_u$
and $x_v$.
Both switches are set to $1$, and each gadget
satisfies its $h$ cycle constraints and its switch
constraint.
Together, the two gadgets contribute
$2(h+1)=2h+2$ points.
On every edge outside the clique, all private variables
are set to $0$.
In every constraint of the two switch gadgets on
such an edge, vertex variables and pins occur only
in control coordinates.
Thus the variable in every unit-witness coordinate
is $0$, so accessibility makes the constraint unsatisfied.

In every pin constraint, the pin is set to $1$
in coordinate $u_1$, and all other variables are
set to $0$.
Its tuple is therefore $e_{u_1}\in P$, so all
$L\varphi$ pin constraints are satisfied.
There are $\binom{k}{2}$ clique edges, giving the
total score
\[
 2(h+1)\binom{k}{2}+L\varphi=t.
\]

\smallskip\noindent\textbf{Backward direction.}
Suppose an assignment $\tau$ satisfies
$\wt(\tau)\le k'$ and $F^{\mathcal I}(\tau)\ge t$.
We show that $G$ contains a $k$-clique and that
$F^{\mathcal I}(\tau)=t$.
Let $s$ be the number of selected variables in
$\bigcup_{e\in E_G}Z_e$.
Let $r$ be the number of variables in $D$ that $\tau$
sets to $1$.
These are the fresh variables introduced in the
pin constraints.
Write $F_{\mathrm{edge}}(\tau)$ for the number of
edge constraints satisfied by $\tau$.

\emph{1. Every pin is selected.}
Suppose $\tau$ leaves at least one pin at $0$.
Every satisfied edge constraint must have a variable
set to $1$ in a unit-witness coordinate, by accessibility.
Any vertex variable or pin in an edge constraint occurs
in a control coordinate, so this selected variable must
be private.

Each selected private variable belongs to one edge.
It occurs only in the two switch gadgets on that edge,
which together contain $2(h+1)$ constraints.
Thus at most $2(h+1)s$ edge constraints contain a
selected private variable, giving
\[
 F_{\mathrm{edge}}(\tau)\le2(h+1)s.
\]

A satisfied pin constraint whose pin is $0$ must
contain a selected variable from $D$, because $P$
is zero-free.
Each variable in $D$ occurs in only one constraint,
so these constraints contribute at most $r$ points.
At most $L-1$ pins are selected, and the constraints
added for each selected pin contribute at most
$\varphi$ points.
Since $s+r\le k'$, we obtain
\begin{equation}
\begin{aligned}
 F^{\mathcal I}(\tau)
 &\le(L-1)\varphi+2(h+1)s+r\\
 &\le(L-1)\varphi+2(h+1)k'\\
 &=L\varphi-1<t.
\end{aligned}
 \label{arityd-dist-hard-eq-missing-pin}
\end{equation}
The last equality uses
$\varphi=2(h+1)k'+1$.
This contradicts $F^{\mathcal I}(\tau)\ge t$.
Therefore $\tau$ selects every pin.

\emph{2. Edge-score bounds.}
Let $p$ be the number of selected vertex variables.
Since all $L$ pins are selected,
\[
 L+p+s+r=\wt(\tau)\le k'=L+k_0.
\]
Therefore $p+s+r\le k_0$, and hence
\begin{equation}
 s\le k_0-p=h\binom{k}{2}+k-p.
 \label{arityd-dist-hard-eq-private-budget}
\end{equation}
There are $L\varphi$ pin constraints, each of weight one.
They therefore contribute at most $L\varphi$ points.

We next bound the score from the two switch gadgets
on each edge.
For each edge $e=uv$, use the endpoint order chosen
in the construction.
The two switch gadgets on $e$ have switches $x_u$
and $x_v$ and share the private-variable set $Z_e$.
Let $S_e$ be the set of these private variables that
$\tau$ sets to $1$, and let $s_e$ be their number:
\[
 S_e:=\{z\in Z_e:\tau(z)=1\},
 \qquad s_e:=|S_e|.
\]
Since different edges have disjoint private-variable
sets,
\[
 s=\sum_{e\in E_G}s_e.
\]
All pins are selected, so we can apply the switch
gadget bounds to each switch gadget.

If $S_e\ne B_e$, each switch gadget contributes at most $s_e$
points, whether its switch is $0$ or $1$, by
\eqref{arityd-dist-hard-eq-switch-score}.
Both switch gadgets have the same selected private set $S_e$.
Their combined score is therefore at most
$s_e+s_e=2s_e$, regardless of which endpoints are selected.

If $S_e=B_e$, every variable in $B_e$ is set to $1$
and every edge filler is set to $0$.
Each switch gadget then satisfies its $h$ cycle constraints.
Its switch constraint is satisfied exactly when
its endpoint variable is set to $1$.
The combined score is therefore
\[
 2h+\tau(x_u)+\tau(x_v).
\]
These two cases give the following table:

\begin{center}
\begin{tabular}{@{}lccc@{}}
\toprule
Private set & Selected endpoints & Gadget scores & Total score \\
\midrule
$S_e\ne B_e$ & $0$, $1$, or $2$
  & At most $s_e$ each & At most $2s_e$ \\
$S_e=B_e$ & $0$ & $h$ and $h$ & $2h$ \\
$S_e=B_e$ & $1$ & $h$ and $h+1$ & $2h+1$ \\
$S_e=B_e$ & $2$ & $h+1$ and $h+1$ & $2h+2$ \\
\bottomrule
\end{tabular}
\end{center}

The table bounds the score on each edge by $2s_e$,
plus at most two additional points when $S_e=B_e$.
The terms $2s_e$ sum to $2s$.

If $S_e=B_e$, the two switch gadgets on $e$ share exactly
$h$ selected private variables.
Private variables belonging to different edges are
distinct.
Thus each such edge uses $h$ different variables
from the total of $s$ selected private variables.
There can therefore be at most $s/h$ edges with
$S_e=B_e$.

We count the additional points in two ways.

\emph{First bound.}
Allow two additional points for every edge with
$S_e=B_e$.
Since there are at most $s/h$ such edges, these
points contribute at most $2s/h$.

\emph{Second bound.}
First allow one additional point for every edge
with $S_e=B_e$, even if neither endpoint is selected.
These points contribute at most $s/h$.

An edge can contribute a second additional point
only when both endpoint variables are set to $1$.
Let $m_G$ be the number of graph edges with both
endpoint variables set to $1$.
Each of these edges contributes at most one second
additional point, so these points contribute at
most $m_G$.
Since $p$ vertex variables are selected and $G$ is
simple, $m_G\le\binom{p}{2}$.

Adding these estimates to $2s$ gives
\begin{equation}
\begin{aligned}
 F_{\mathrm{edge}}(\tau)&\le2s+\frac{2s}{h},\\
 F_{\mathrm{edge}}(\tau)&\le2s+\frac{s}{h}+m_G.
\end{aligned}
 \label{arityd-dist-hard-eq-master-bound}
\end{equation}

\emph{3. Exactly $k$ vertices are selected, and they form a clique.}
The pin constraints contribute at most $L\varphi$ points.
Since $F^{\mathcal I}(\tau)\ge t$, the edge constraints
must contribute at least
\[
 F_{\mathrm{edge}}(\tau)\ge t-L\varphi=t_0.
\]

If $p>k$, selecting these extra vertex variables leaves
fewer private variables available:
\[
 s\le h\binom{k}{2}-(p-k).
\]
The first edge-score bound therefore gives
\[
 F_{\mathrm{edge}}(\tau)
 \le t_0-\left(2+\frac2h\right)(p-k)<t_0.
\]
This contradicts the required edge score.

If $p<k$, the private-variable budget gives
\[
 s\le h\binom{k}{2}+(k-p).
\]
Compared with $h\binom{k}{2}$, this allows at most
$k-p$ more selected private variables.
In the second edge-score bound, this increases
the term $(2+1/h)s$ by at most $(2+1/h)(k-p)$.
However, the number of possible pairs of selected vertices
decreases by
\[
 \binom{k}{2}-\binom{p}{2}
 =\frac{(k-p)(k+p-1)}2
 \ge\frac52(k-p),
\]
using $k\ge6$ and $p\ge0$.
Since $h\ge3$, we have $2+1/h\le7/3$.
Thus the decrease in possible pairs is larger than
the increase from the private-variable budget.
Formally,
\begin{equation}
\begin{aligned}
 F_{\mathrm{edge}}(\tau)
 &\le t_0+\left(2+\frac1h\right)(k-p)
       -\left(\binom{k}{2}-\binom{p}{2}\right)\\
 &\le t_0+\frac73(k-p)-\frac52(k-p)\\
 &=t_0-\frac{k-p}{6}<t_0.
\end{aligned}
 \label{arityd-dist-hard-eq-fewer-vertices}
\end{equation}
This also contradicts the required edge score.
Therefore $p=k$.

We now have $s\le h\binom{k}{2}$ and $m_G\le\binom{k}{2}$.
The second edge-score bound gives
\[
 F_{\mathrm{edge}}(\tau)
 \le2s+\frac{s}{h}+m_G
 \le(2h+1)\binom{k}{2}+\binom{k}{2}
 =t_0.
\]
The edge score is at least $t_0$, so equality holds
throughout the displayed inequalities.
Both coefficients in $(2+1/h)s+m_G$ are positive.
If either $s<h\binom{k}{2}$ or $m_G<\binom{k}{2}$,
this upper bound would be strictly less than $t_0$.
Therefore
\[
 s=h\binom{k}{2},
 \qquad m_G=\binom{k}{2}.
\]
There are $\binom{k}{2}$ pairs among the $k$ selected
vertices, and $m_G$ counts the graph edges between them.
Thus every pair is an edge, giving a $k$-clique.

We have also shown that $F_{\mathrm{edge}}(\tau)=t_0$.
The pin constraints contribute at most $L\varphi$ points, so
\[
 F^{\mathcal I}(\tau)\le L\varphi+t_0=t.
\]
Our assumption $F^{\mathcal I}(\tau)\ge t$ therefore
forces equality.
Thus no feasible assignment has score greater than $t$,
and any feasible assignment with score $t$ gives a
$k$-clique.
Together with the forward direction, this proves
\eqref{arityd-dist-hard-eq-target}.

\paragraph*{Distinct variables and constraint scopes.}
Every constraint has weight one and uses either $P$
or $Q$, both of which belong to $\Gamma$.

No constraint repeats a variable.
Within each switch gadget, different variables of
$H$ are replaced by different variables.
Each pin constraint uses its pin and a distinct
fresh variable in every other coordinate.

We now show that no two constraints contain
exactly the same set of variables.

\emph{Pin constraints.}
Since $u_1\ne c$, the relation $P$ has arity at
least two.
Every pin constraint therefore contains a fresh
variable from $D$.
This variable occurs in no other constraint,
so it distinguishes this pin constraint from
every other constraint.

\emph{Constraints on different edges.}
Every constraint on an edge $e$ contains a variable
from $B_e$: the switch constraint contains
$z_{e,1}$, and each cycle constraint contains
$q$ variables from $B_e$.
Variables in $B_e$ occur only in constraints on $e$.
Thus constraints on different edges cannot contain
the same set of variables.

\emph{Constraints on the same edge $e=uv$.}
The switch constraint for $u$ contains $x_u$
as its only vertex variable.
The switch constraint for $v$ contains $x_v$
as its only vertex variable.
Since $u\ne v$, these constraints have different
sets of variables.
Neither can have the same set of variables as
a cycle constraint, because cycle constraints
contain no vertex variables.

Within either switch gadget, the cycle constraint
indexed by $i\in[h]$ contains exactly the following
variables from $B_e$:
\[
 \{z_{e,i},z_{e,i+1},\ldots,z_{e,i+q-1}\}.
\]
Since $q<h$, the first variable $z_{e,i}$ is the
unique variable in this group whose predecessor
in the cyclic order lies outside the group.
Thus different starting indices give different
groups.
Consequently, cycle constraints within one
switch gadget have different sets of variables.

Finally, $T_Q\ne\varnothing$.
Every cycle constraint in the gadget with switch
$x_u$ contains a pin from $Y^{(1)}$ and none
from $Y^{(2)}$.
Every cycle constraint in the gadget with switch
$x_v$ contains a pin from $Y^{(2)}$ and none
from $Y^{(1)}$.
Thus cycle constraints from the two gadgets
also have different sets of variables.

Every pair of constraints falls into one of
these cases.
Hence $\mathcal I$ has multiplicity one.

\paragraph*{Parameter bounds and reduction size.}
We first bound the budget and target in terms of
$d$ and $k$, then count the variables and constraints.

The relation $Q$ has $q$ coordinates in $A_Q$ and
at least one coordinate in $T_Q$, so $q\le d-1$.
The control coordinate $c$ lies outside $A_P$,
so $\ell\le d-1$.
Hence
\[
 3\le h=\max\{q+1,\ell\}\le d.
\]
In particular, $d\ge3$.
We also have $2\le L\le4(d-2)$.

\emph{Budget.}
Since $h\le d$ and $k\ge6$,
\[
\begin{aligned}
 k_0
 &=k+h\binom{k}{2}
   \le k+d\binom{k}{2}
   \le\frac{dk^2}{2},\\
 k'
 &=L+k_0
   \le4(d-2)+\frac{dk^2}{2}
   \le dk^2.
\end{aligned}
\]

\emph{Target.}
The bound on $k'$ gives
\[
 \varphi=2(h+1)k'+1
 \le2d(d+1)k^2+1
 \le3d^2k^2.
\]
Since $h\ge1$, the definition
$k'=L+k+h\binom{k}{2}$ gives $\binom{k}{2}\le k'$.
Therefore
\[
 t_0=2(h+1)\binom{k}{2}
 \le2(h+1)k'
 =\varphi-1
 <\varphi.
\]
Using $L+1\le4d$, we obtain
\[
 t=L\varphi+t_0
 <(L+1)\varphi
 \le4d\varphi
 \le12d^3k^2.
\]

\emph{Number of variables and constraints.}
Each constraint of $H$ contains the switch or a pin.
Its remaining coordinates contain at most $d-1$
private variables.
Since $H$ has $h+1$ constraints and every private
variable occurs in one of them,
\[
 |Z|\le(h+1)(d-1)\le d^2-1.
\]
Each pin constraint introduces at most $d-1$ fresh
variables. There are $L\varphi$ such constraints,
so
\[
 |D|\le(d-1)L\varphi\le(d-1)t.
\]

Each edge contributes $2(h+1)$ constraints from
its two switch gadgets.
Adding the $L\varphi$ pin constraints gives
\[
 |\mathcal C|
 =2(h+1)|E_G|+L\varphi
 \le2(d+1)|E_G|+t.
\]

There are $|V_G|$ vertex variables, $L$ pins,
$|Z|$ private variables for each edge, and
$|D|$ fresh variables from the pin constraints.
The two switch gadgets on an edge share all its
private variables, so we count these variables once.
The four groups are disjoint. Thus
\[
\begin{aligned}
 |X|
 &=|V_G|+L+|Z||E_G|+|D|\\
 &\le|V_G|+4(d-2)+(d^2-1)|E_G|+(d-1)t.
\end{aligned}
\]

For fixed $\Gamma$, the value $d$ is constant.
The bounds above therefore give
$O(|V_G|+|E_G|+k^2)$ variables and constraints.
Each constraint has at most $d$ coordinates,
so the construction takes time polynomial in
$|V_G|+|E_G|+k^2$.
The output budget $k'$ and target $t$ are bounded
by functions of $k$ and $d$.
Together with the correctness proved above,
this establishes an \fpt{} reduction.
\end{proof}

\subsection{The hardness theorem}

\begin{theorem}[Unweighted multiplicity-one hardness]
\label{arityd-dist-hard-thm-necessity}
Let $\Gamma$ be a fixed finite Boolean constraint
language of maximum arity $d$.
Suppose $\Gamma_+$ contains a relation that is not
\Direct{} and a relation that is not \Rooted{}.

Then \maxwtcsp{}$(\Gamma)$ admits no \epas{}
on unweighted multiplicity-one instances unless
$\fpt{}=\wone{}$.
\end{theorem}

\begin{proof}
If $\Gamma_+$ contains an inaccessible relation,
apply Lemma~\ref{arityd-dist-hard-lem-inaccessible}.

Otherwise, every relation in $\Gamma_+$ is accessible.
Choose $P,Q\in\Gamma_+$ such that $P$ is not \Direct{}
and $Q$ is not \Rooted{}, and apply
Lemma~\ref{arityd-dist-hard-lem-activator-clique}.

Both reductions produce unweighted multiplicity-one
instances satisfying \eqref{arityd-dist-hard-eq-target}.
Their output budgets and targets are bounded by
functions of the clique parameter $k$ and the arity
bound $d$.
Since $\Gamma$ is fixed, $d$ is constant.

Every feasible score is a nonnegative integer.
Lemma~\ref{arityd-dist-hard-lem-exact-cap} therefore
excludes an \epas{} unless
$\fpt{}=\wone{}$.
\end{proof}

\section{Completing the classification}
\label{sec:classification-completion}
\label{ss:arityd-proof}

\begin{proof}[Proof of Theorem~\ref{thm:arityd-complete-dichotomy}]
If every relation in $\Gamma_+$ is \Direct{},
Theorem~\ref{arityd-direct-thm-main} gives the claimed
\epas{} and running time.
If every relation in $\Gamma_+$ is \Rooted{},
Theorem~\ref{thm:arityd-rooted-epas} gives the claimed
\epas{} and running time.
Both results use Lemma~\ref{lem:arityd-zero-valid-lifting}
to extend the corresponding zero-free algorithm to
instances that also contain zero-valid constraints.
The zero-free algorithms respect prescribed variable
values, as required by that lemma.
If both conditions hold, either algorithm applies.

If neither condition holds, $\Gamma_+$ contains
a relation that is not \Direct{} and a relation
that is not \Rooted{}.
Theorem~\ref{arityd-dist-hard-thm-necessity} therefore
excludes an \epas{} unless
$\fpt{}=\wone{}$, even on unweighted
multiplicity-one instances.

It remains to justify
Remark~\ref{rem:arityd-uniform-algorithms}.
Both algorithms access the original relations only
through membership tests.
Testing the all-zero tuple identifies zero-valid
relations.
Testing the at most $d$ unit tuples identifies
the unit-witness coordinates.
Under either sufficient condition on $\Gamma_+$,
Lemma~\ref{lem:arityd-repair-accessibility} implies
that a zero-free relation with no unit-witness
coordinate is empty.
Its constraints can therefore be discarded without
enumerating satisfying tuples.

The running-time proofs of
Theorems~\ref{thm:arityd-rooted-zero-free}
and~\ref{arityd-direct-thm-zero-free} bound the number
of membership tests and show that the intermediate
rational numbers have polynomial bit length in
$\ell_{\mathrm{in}}$.
Thus polynomial-time membership tests, with a fixed
exponent across all inputs, and exact arithmetic
preserve the claimed running-time bounds.
The proof of Lemma~\ref{lem:arityd-zero-valid-lifting}
establishes the same guarantee for the steps handling
zero-valid constraints.
This proves the remark.
\end{proof}

\subsection{Relations of arity at most two}
\label{ss:arityd-binary-specialization}

\begin{proposition}
\label{prop:arityd-binary-specialization}
Every nonempty zero-free relation of arity at most two
other than $R_{\mathrm{AND}}=\{(1,1)\}$ is both
\Direct{} and \Rooted{}.
The relation $R_{\mathrm{AND}}$ satisfies neither rule.
\end{proposition}

\begin{proof}
The only nonempty zero-free unary relation is $\{1\}$,
which satisfies both rules.

Let $R$ be nonempty, zero-free, binary, and different from
$R_{\mathrm{AND}}$. Every $\alpha\in R$ has a selected unit witness $u$:
this is immediate for a unit tuple, and for $\alpha=(1,1)$ it follows
because $R$ contains a unit tuple $e_u$.
Deleting a selected control gives $e_u$, proving \Direct{}.
Keeping only $u$ among the unit witnesses and any subset of the selected
controls gives either $e_u$ or $\alpha$, proving \Rooted{}.
Finally, $R_{\mathrm{AND}}$ has no unit witness and deleting a control
from $(1,1)$ gives a tuple outside the relation, so it satisfies neither rule.
\end{proof}

\begin{proof}[Proof of Corollary~\ref{cor:arity-two}]
If $R_{\mathrm{AND}}\notin\Gamma$, Proposition~\ref{prop:arityd-binary-specialization}
and the \Rooted{} case of Theorem~\ref{thm:arityd-complete-dichotomy}
with $d=2$ give the stated deterministic \epas{} and running time.
If $R_{\mathrm{AND}}\in\Gamma$, it belongs to $\Gamma_+$ and satisfies
neither rule, so the hardness case of the same theorem excludes an
\epas{} unless $\fpt{}=\wone{}$.
\end{proof}

\section{Conclusion}
\label{sec:conclusion}

We give a complete classification of the fixed finite Boolean constraint
languages of arity at most $d$ that admit an \epas{} for \maxwtcsp{},
parameterized by $k+d$, when variables within each constraint are distinct
(Theorem~\ref{thm:arityd-complete-dichotomy}).
An interesting open question is whether the running time of these approximation schemes can be improved, particularly the dependence
on $k$, $d$, and $\varepsilon$.
Another natural open question is whether an analogous classification
holds over larger finite domains, with a distinguished value $0$ and
a budget of at most $k$ variables assigned a nonzero value.

\section*{Acknowledgments}
Souvik Saha acknowledges support from the Anusandhan National Research
Foundation (ANRF), Government of India, through the National Post Doctoral
Fellowship (N-PDF), File No.~PDF/2026/005409.

\section*{AI disclosure}
ChatGPT (OpenAI) helped prepare the manuscript and provided first drafts
of proofs, including the hardness argument. The algorithms were developed
by the authors. The authors verified and extensively simplified the
hardness proofs through interactive discussions with ChatGPT. They
independently checked the hardness proofs and references, reviewed and
edited the final text, and take full responsibility for the paper's
results, proofs, and text.

\bibliographystyle{plainurl}
\bibliography{references}
\end{document}